\documentclass[twoside,leqno]{article}

\usepackage[letterpaper]{geometry}

\usepackage{siamproceedings}

\usepackage[T1]{fontenc}
\usepackage{amsfonts}
\usepackage{graphicx}
\usepackage{epstopdf}
\usepackage{enumitem}
\usepackage{algorithmic}
\ifpdf
  \DeclareGraphicsExtensions{.eps,.pdf,.png,.jpg}
\else
  \DeclareGraphicsExtensions{.eps}
\fi

\newsiamremark{remark}{Remark}
\newsiamremark{hypothesis}{Hypothesis}
\crefname{hypothesis}{Hypothesis}{Hypotheses}
\newsiamthm{claim}{Claim}

\usepackage{amsopn}

\title{Dynamic Hierarchical j-Tree Decomposition and Its Applications}

\date{}
\author{Gramoz Goranci\thanks{Faculty of Computer Science, University of Vienna, Austria.}
\and Monika Henzinger\thanks{Institute of Science and Technology Austria (ISTA), Klosterneuburg, Austria.}
\and Peter Kiss\footnotemark[1]
\and Ali Momeni\thanks{Faculty of Computer Science, UniVie Doctoral School Computer Science DoCS, University of Vienna, Austria.}
\and Gernot Z\"ocklein\footnotemark[2]
}

\usepackage{graphicx}
\usepackage{subcaption}
\usepackage{float}

\graphicspath{{graphics/}{./graphics/}}

\usepackage{tikz}

\usepackage{tikzscale}

\usetikzlibrary{arrows.meta}

\makeatletter

\pgfdeclarearrow{
  name = ipe _linear,
  defaults = {
    length = +1bp,
    width  = +.666bp,
    line width = +0pt 1,
  },
  setup code = {
    \pgfarrowssetbackend{0pt}
    \pgfarrowssetvisualbackend{
      \pgfarrowlength\advance\pgf@x by-.5\pgfarrowlinewidth}
    \pgfarrowssetlineend{\pgfarrowlength}
    \ifpgfarrowreversed
      \pgfarrowssetlineend{\pgfarrowlength\advance\pgf@x by-.5\pgfarrowlinewidth}
    \fi
    \pgfarrowssettipend{\pgfarrowlength}
    \pgfarrowshullpoint{\pgfarrowlength}{0pt}
    \pgfarrowsupperhullpoint{0pt}{.5\pgfarrowwidth}
    \pgfarrowssavethe\pgfarrowlinewidth
    \pgfarrowssavethe\pgfarrowlength
    \pgfarrowssavethe\pgfarrowwidth
  },
  drawing code = {
    \pgfsetdash{}{+0pt}
    \ifdim\pgfarrowlinewidth=\pgflinewidth\else\pgfsetlinewidth{+\pgfarrowlinewidth}\fi
    \pgfpathmoveto{\pgfqpoint{0pt}{.5\pgfarrowwidth}}
    \pgfpathlineto{\pgfqpoint{\pgfarrowlength}{0pt}}
    \pgfpathlineto{\pgfqpoint{0pt}{-.5\pgfarrowwidth}}
    \pgfusepathqstroke
  },
  parameters = {
    \the\pgfarrowlinewidth,%
    \the\pgfarrowlength,%
    \the\pgfarrowwidth,%
  },
}

\pgfdeclarearrow{
  name = ipe _pointed,
  defaults = {
    length = +1bp,
    width  = +.666bp,
    inset  = +.2bp,
    line width = +0pt 1,
  },
  setup code = {
    \pgfarrowssetbackend{0pt}
    \pgfarrowssetvisualbackend{\pgfarrowinset}
    \pgfarrowssetlineend{\pgfarrowinset}
    \ifpgfarrowreversed
      \pgfarrowssetlineend{\pgfarrowlength}
    \fi
    \pgfarrowssettipend{\pgfarrowlength}
    \pgfarrowshullpoint{\pgfarrowlength}{0pt}
    \pgfarrowsupperhullpoint{0pt}{.5\pgfarrowwidth}
    \pgfarrowshullpoint{\pgfarrowinset}{0pt}
    \pgfarrowssavethe\pgfarrowinset
    \pgfarrowssavethe\pgfarrowlinewidth
    \pgfarrowssavethe\pgfarrowlength
    \pgfarrowssavethe\pgfarrowwidth
  },
  drawing code = {
    \pgfsetdash{}{+0pt}
    \ifdim\pgfarrowlinewidth=\pgflinewidth\else\pgfsetlinewidth{+\pgfarrowlinewidth}\fi
    \pgfpathmoveto{\pgfqpoint{\pgfarrowlength}{0pt}}
    \pgfpathlineto{\pgfqpoint{0pt}{.5\pgfarrowwidth}}
    \pgfpathlineto{\pgfqpoint{\pgfarrowinset}{0pt}}
    \pgfpathlineto{\pgfqpoint{0pt}{-.5\pgfarrowwidth}}
    \pgfpathclose
    \ifpgfarrowopen
      \pgfusepathqstroke
    \else
      \ifdim\pgfarrowlinewidth>0pt\pgfusepathqfillstroke\else\pgfusepathqfill\fi
    \fi
  },
  parameters = {
    \the\pgfarrowlinewidth,%
    \the\pgfarrowlength,%
    \the\pgfarrowwidth,%
    \the\pgfarrowinset,%
    \ifpgfarrowopen o\fi%
  },
}

\newdimen\ipeminipagewidth

\tikzstyle{ipe import} = [
  x=1bp, y=1bp,
  ipe node stretch/.store in=\ipenodestretch,
  ipe stretch normal/.style={ipe node stretch=1},
  ipe stretch normal,
  ipe node/.style={
    anchor=base west, inner sep=0, outer sep=0, scale=\ipenodestretch
  },
  ipe mark scale/.store in=\ipemarkscale,
  ipe mark tiny/.style={ipe mark scale=1.1},
  ipe mark small/.style={ipe mark scale=2},
  ipe mark normal/.style={ipe mark scale=3},
  ipe mark large/.style={ipe mark scale=5},
  ipe mark normal, 
  ipe circle/.pic={
    \draw[line width=0.2*\ipemarkscale]
      (0,0) circle[radius=0.5*\ipemarkscale];
    \coordinate () at (0,0);
  },
  ipe disk/.pic={
    \fill (0,0) circle[radius=0.6*\ipemarkscale];
    \coordinate () at (0,0);
  },
  ipe fdisk/.pic={
    \filldraw[line width=0.2*\ipemarkscale]
      (0,0) circle[radius=0.5*\ipemarkscale];
    \coordinate () at (0,0);
  },
  ipe box/.pic={
    \draw[line width=0.2*\ipemarkscale, line join=miter]
      (-.5*\ipemarkscale,-.5*\ipemarkscale) rectangle
      ( .5*\ipemarkscale, .5*\ipemarkscale);
    \coordinate () at (0,0);
  },
  ipe square/.pic={
    \fill
      (-.6*\ipemarkscale,-.6*\ipemarkscale) rectangle
      ( .6*\ipemarkscale, .6*\ipemarkscale);
    \coordinate () at (0,0);
  },
  ipe fsquare/.pic={
    \filldraw[line width=0.2*\ipemarkscale, line join=miter]
      (-.5*\ipemarkscale,-.5*\ipemarkscale) rectangle
      ( .5*\ipemarkscale, .5*\ipemarkscale);
    \coordinate () at (0,0);
  },
  ipe cross/.pic={
    \draw[line width=0.2*\ipemarkscale, line cap=butt]
      (-.5*\ipemarkscale,-.5*\ipemarkscale) --
      ( .5*\ipemarkscale, .5*\ipemarkscale)
      (-.5*\ipemarkscale, .5*\ipemarkscale) --
      ( .5*\ipemarkscale,-.5*\ipemarkscale);
    \coordinate () at (0,0);
  },
  /pgf/arrow keys/.cd,
  ipe arrow normal/.style={scale=1},
  ipe arrow tiny/.style={scale=.4},
  ipe arrow small/.style={scale=.7},
  ipe arrow large/.style={scale=1.4},
  ipe arrow normal,
  /tikz/.cd,
  ipe arrows/.style={
    ipe normal/.tip={
      ipe _pointed[length=1bp, width=.666bp, inset=0bp,
                   quick, ipe arrow normal]},
    ipe pointed/.tip={
      ipe _pointed[length=1bp, width=.666bp, inset=0.2bp,
                   quick, ipe arrow normal]},
    ipe linear/.tip={
      ipe _linear[length = 1bp, width=.666bp,
                  ipe arrow normal, quick]},
    ipe fnormal/.tip={ipe normal[fill=white]},
    ipe fpointed/.tip={ipe pointed[fill=white]},
    ipe double/.tip={ipe normal[] ipe normal},
    ipe fdouble/.tip={ipe fnormal[] ipe fnormal},
    ipe arc/.tip={ipe normal},
    ipe farc/.tip={ipe fnormal},
    ipe ptarc/.tip={ipe pointed},
    ipe fptarc/.tip={ipe fpointed},
  },
  ipe arrows, 
]

\tikzset{
  rgb color/.code args={#1=#2}{%
    \definecolor{tempcolor-#1}{rgb}{#2}%
    \tikzset{#1=tempcolor-#1}%
  },
}

\makeatother

\tikzstyle{ipe stylesheet} = [
  ipe import,
  nonzero rule,
  line join=miter,
  line cap=butt,
  ipe pen normal/.style={line width=0.4},
  ipe pen thin/.style={line width=0.4},
  ipe pen semithick/.style={line width=0.6},
  ipe pen thick/.style={line width=0.8},
  ipe pen ultra thick/.style={line width=1.6},
  ipe pen ultra thin/.style={line width=0.1},
  ipe pen very thick/.style={line width=1.2},
  ipe pen very thin/.style={line width=0.2},
  ipe pen heavier/.style={line width=0.8},
  ipe pen fat/.style={line width=1.2},
  ipe pen ultrafat/.style={line width=2},
  ipe pen uultrafat/.style={line width=2.7},
  ipe pen uuultrafat/.style={line width=3.3},
  ipe pen normal,
  ipe mark normal/.style={ipe mark scale=3},
  ipe mark large/.style={ipe mark scale=5},
  ipe mark small/.style={ipe mark scale=2},
  ipe mark tiny/.style={ipe mark scale=1.1},
  ipe mark Large/.style={ipe mark scale=5},
  ipe mark LARGE/.style={ipe mark scale=6},
  ipe mark huge/.style={ipe mark scale=7},
  ipe mark normal,
  /pgf/arrow keys/.cd,
  ipe arrow normal/.style={scale=4},
  ipe arrow large/.style={scale=1.4},
  ipe arrow small/.style={scale=0.7},
  ipe arrow tiny/.style={scale=0.4},
  ipe arrow normal,
  /tikz/.cd,
  ipe arrows, 
  <->/.tip = ipe normal,
  ipe dash normal/.style={dash pattern=},
  ipe dash dotted/.style={dash pattern=on 1bp off 3bp},
  ipe dash dashed/.style={dash pattern=on 3bp off 3bp},
  ipe dash dash dot dot/.style={dash pattern=on 3bp off 2bp on 1bp off 2bp on 1bp off 2bp},
  ipe dash dash dot/.style={dash pattern=on 3bp off 2bp on 1bp off 2bp},
  ipe dash densely dash dot dot/.style={dash pattern=on 3bp off 1bp on 1bp off 1bp on 1bp off 1bp},
  ipe dash densely dash dot/.style={dash pattern=on 3bp off 1bp on 1bp off 1bp},
  ipe dash densely dashed/.style={dash pattern=on 3bp off 2bp},
  ipe dash densely dotted/.style={dash pattern=on 1bp off 1bp},
  ipe dash loosely dash dot dot/.style={dash pattern=on 3bp off 4bp on 1bp off 4bp on 1bp off 4bp},
  ipe dash loosely dash dot/.style={dash pattern=on 3bp off 4bp on 1bp off 4bp},
  ipe dash loosely dashed/.style={dash pattern=on 3bp off 6bp},
  ipe dash loosely dotted/.style={dash pattern=on 1bp off 4bp},
  ipe dash solid/.style={dash pattern=},
  ipe dash dash dotted/.style={dash pattern=on 4bp off 2bp on 1bp off 2bp},
  ipe dash dash dot dotted/.style={dash pattern=on 4bp off 2bp on 1bp off 2bp on 1bp off 2bp},
  ipe dash dashed (a bit finer)/.style={dash pattern=on 3bp off 3bp},
  ipe dash dashed (almost dotted)/.style={dash pattern=on 1bp off 1bp},
  ipe dash dashed (fine)/.style={dash pattern=on 2bp off 2bp},
  ipe dash dashed (very fine)/.style={dash pattern=on 1.5bp off 1.5bp},
  ipe dash dotted (heavy)/.style={dash pattern=on 2.4bp off 4bp},
  ipe dash fine dotted/.style={dash pattern=on 0.5bp off 1.5bp},
  ipe dash normal,
  ipe node/.append style={font=\normalsize},
  ipe stretch normal/.style={ipe node stretch=1},
  ipe stretch normal,
  ipe opacity opaque/.style={opacity=1},
  ipe opacity nearly opaque/.style={opacity=0.75},
  ipe opacity nearly transparent/.style={opacity=0.25},
  ipe opacity semitransparent/.style={opacity=0.5},
  ipe opacity transparent/.style={opacity=0.001},
  ipe opacity ultra nearly opaque/.style={opacity=0.95},
  ipe opacity ultra nearly transparent/.style={opacity=0.05},
  ipe opacity very nearly opaque/.style={opacity=0.9},
  ipe opacity very nearly transparent/.style={opacity=0.1},
  ipe opacity 10 opacity/.style={opacity=0.1},
  ipe opacity 25 opacity/.style={opacity=0.25},
  ipe opacity 33 opacity/.style={opacity=0.333},
  ipe opacity 50 opacity/.style={opacity=0.5},
  ipe opacity 66 opacity/.style={opacity=0.666},
  ipe opacity 75 opacity/.style={opacity=0.75},
  ipe opacity 90 opacity/.style={opacity=0.9},
  ipe opacity 10/.style={opacity=0.1},
  ipe opacity 30/.style={opacity=0.3},
  ipe opacity 50/.style={opacity=0.5},
  ipe opacity 75/.style={opacity=0.75},
  ipe opacity opaque,
]
\definecolor{red}{rgb}{1,0,0}
\definecolor{blue}{rgb}{0,0,1}
\definecolor{green}{rgb}{0,1,0}
\definecolor{yellow}{rgb}{1,1,0}
\definecolor{orange}{rgb}{1,0.5,0}
\definecolor{purple}{rgb}{0.75,0,0.25}
\definecolor{gray}{rgb}{0.5,0.5,0.5}
\definecolor{brown}{rgb}{0.75,0.5,0.25}
\definecolor{pink}{rgb}{1,0.75,0.75}
\definecolor{violet}{rgb}{0.5,0,0.5}
\definecolor{darkgray}{rgb}{0.25,0.25,0.25}
\definecolor{lightgray}{rgb}{0.75,0.75,0.75}
\definecolor{cyan}{rgb}{0,1,1}
\definecolor{lime}{rgb}{0.75,1,0}
\definecolor{magenta}{rgb}{1,0,1}
\definecolor{olive}{rgb}{0.5,0.5,0}
\definecolor{teal}{rgb}{0,0.5,0.5}
\definecolor{OI black}{rgb}{0,0,0}
\definecolor{OI dark blue}{rgb}{0,0.448,0.699}
\definecolor{OI green}{rgb}{0,0.618,0.452}
\definecolor{OI light blue}{rgb}{0.337,0.706,0.915}
\definecolor{OI orange}{rgb}{0.901,0.625,0}
\definecolor{OI purple}{rgb}{0.799,0.476,0.655}
\definecolor{OI red}{rgb}{0.835,0.367,0}
\definecolor{OI white}{rgb}{1,1,1}
\definecolor{OI yellow}{rgb}{0.942,0.894,0.26}
\definecolor{black}{rgb}{0,0,0}
\definecolor{white}{rgb}{1,1,1}
\definecolor{gold}{rgb}{1,0.843,0}
\definecolor{navy}{rgb}{0,0,0.502}
\definecolor{seagreen}{rgb}{0.18,0.545,0.341}
\definecolor{turquoise}{rgb}{0.251,0.878,0.816}
\definecolor{dark blue}{rgb}{0.121,0.47,0.705}
\definecolor{dark brown}{rgb}{0.651,0.337,0.157}
\definecolor{dark cyan}{rgb}{0.106,0.62,0.467}
\definecolor{dark gray}{rgb}{0.5,0.5,0.5}
\definecolor{dark green}{rgb}{0.2,0.627,0.172}
\definecolor{dark magenta}{rgb}{0.545,0,0.545}
\definecolor{dark orange}{rgb}{1,0.498,0}
\definecolor{dark pink}{rgb}{0.969,0.506,0.749}
\definecolor{dark purple}{rgb}{0.415,0.239,0.603}
\definecolor{dark red}{rgb}{0.63,0.066,0.078}
\definecolor{dark yellow}{rgb}{1,1,0.2}
\definecolor{light blue}{rgb}{0.651,0.807,0.89}
\definecolor{light brown}{rgb}{0.898,0.847,0.741}
\definecolor{light cyan}{rgb}{0.553,0.827,0.78}
\definecolor{light gray}{rgb}{0.8,0.8,0.8}
\definecolor{light green}{rgb}{0.698,0.874,0.541}
\definecolor{light orange}{rgb}{0.992,0.749,0.435}
\definecolor{light pink}{rgb}{0.992,0.855,0.925}
\definecolor{light purple}{rgb}{0.792,0.698,0.839}
\definecolor{light red}{rgb}{0.984,0.603,0.6}
\definecolor{light yellow}{rgb}{1,1,0.8}
\definecolor{mid red}{rgb}{0.89,0.102,0.109}
\definecolor{darkblue}{rgb}{0,0,0.545}
\definecolor{darkcyan}{rgb}{0,0.545,0.545}
\definecolor{darkgreen}{rgb}{0,0.392,0}
\definecolor{darkmagenta}{rgb}{0.545,0,0.545}
\definecolor{darkorange}{rgb}{1,0.549,0}
\definecolor{darkred}{rgb}{0.545,0,0}
\definecolor{lightblue}{rgb}{0.678,0.847,0.902}
\definecolor{lightcyan}{rgb}{0.878,1,1}
\definecolor{lightgreen}{rgb}{0.565,0.933,0.565}
\definecolor{lightyellow}{rgb}{1,1,0.878}

\usepackage[
        ]{xcolor}

\definecolor{myOrange}{RGB}{230, 159, 0}
\definecolor{myLightBlue}{RGB}{86, 180, 233}
\definecolor{myGreen}{RGB}{0, 158, 115}
\definecolor{myYellow}{RGB}{240, 228, 66}
\definecolor{myDarkBlue}{RGB}{0, 114, 178}
\definecolor{myRed}{RGB}{213, 94, 0}
\definecolor{myPink}{RGB}{204, 121, 167}

\usepackage{soul}
\sethlcolor{myYellow}

\usepackage{amsmath}
\usepackage{amssymb}
\usepackage{thmtools, thm-restate}
\usepackage{ntheorem}

\newtheorem{observation}{Observation}[section]

\newtheorem{question}{\color{red}{Question}}[section]

\renewtheorem{theorem}{Theorem}[section]

\usepackage[algo2e, linesnumbered, boxruled, vlined]{algorithm2e}

\DontPrintSemicolon

\SetKwInput{Maintain}{Maintain}

\SetKwProg{Procedure}{Procedure}{}{}

\SetArgSty{textnormal}

\SetKw{Break}{break}

\usepackage{xspace}

\usepackage{mleftright}

\newcommand{\vect}[1]{\ensuremath{\boldsymbol{#1}}}
\newcommand{\atmost}[1]{\ensuremath{ O\mleft( #1 \mright)}\xspace}

\newcommand{\littleo}[1]{\ensuremath{o\mleft( #1 \mright)}\xspace}

\newcommand{\atmosttilde}[1]{\ensuremath{ \widetilde{O}\mleft( #1 \mright)}\xspace}
\newcommand{\poly}[1]{\ensuremath{\operatorname{poly}\mleft( #1 \mright)}\xspace}

\newcommand{\rec}{\ensuremath{\gamma_{\text{rec}}}\xspace}

\usepackage{tcolorbox}

\def\setof#1{\left\{#1  \right\}}

\newcommand\f{\boldsymbol{f}}

\def\tn#1{\textnormal{#1}}

\def\norm#1{\left\| #1 \right\|}
\def\rbrack#1{\left( #1 \right)}

\newcommand{\proj}{\mathsf{proj}}
\newcommand{\fcong}{\mathsf{cong}}
\newcommand{\cross}{\mathsf{cross}}

\newcommand{\str}{\textnormal{str}}

\def\cD{\mathcal{D}}
\def\cP{\mathcal{P}}

\newcommand\vecone{\boldsymbol{1}}

\renewcommand{\deg}{\textnormal{deg}}

\newcommand\R{\mathbb{R}}

\newcommand{\E}[1]{\mathop{{}\mathbb{E}}\left[#1\right]}

\newcommand{\polylog}{\textnormal{polylog}}

\renewcommand{\root}{\ensuremath{\mathsf{root}}}

\renewcommand{\tilde}{\widetilde}

\newcommand{\cH}{\mathcal{H}}

\newcommand{\cE}{\mathcal{E}}

\newcommand{\Otil}{\tilde{O}}

\renewcommand{\epsilon}{\ensuremath\varepsilon}

\renewcommand{\phi}{\ensuremath{\varphi}}

\newcommand{\gernot}[1]{{ \textbf{\color{purple} Gernot: #1}}}

\usepackage[nameinlink]{cleveref}
\crefname{problem}{problem}{problems}
\crefname{claim}{claim}{claims}
\crefname{lemma}{lemma}{lemmas}
\Crefname{algocf}{Algorithm}{Algorithms}
\crefname{proof}{proof}{proofs}
\crefname{observation}{observation}{observations}
\crefname{invariant}{invariant}{invariants}

\begin{document}
\maketitle
\pagenumbering{roman}

\begin{abstract}

We develop a new algorithmic framework for designing approximation algorithms for cut-based optimization problems on capacitated undirected graphs that undergo edge insertions and deletions. Specifically, our framework dynamically maintains a variant of the hierarchical $j$-tree decomposition of [Madry FOCS’10], achieving a poly-logarithmic approximation factor to the graph's cut structure and supporting edge updates in $O(n^{\epsilon})$ amortized update time, for any arbitrarily small constant $\epsilon \in (0,1)$.


Consequently, we obtain new trade-offs between approximation and update/query time for fundamental cut-based optimization problems in the fully dynamic setting, including all-pairs minimum cuts, sparsest cut, multi-way cut, and multi-cut. For the last three problems, these trade-offs give the first fully-dynamic algorithms achieving poly-logarithmic approximation in sub-linear time per operation.   

The main technical ingredient behind our dynamic hierarchy is a dynamic cut-sparsifier algorithm that can handle vertex splits with low recourse. This is achieved by white-boxing the dynamic cut sparsifier construction of [Abraham et al. FOCS’16], based on forest packing, together with new structural insights about the maintenance of these forests under vertex splits. Given the versatility of cut sparsification in both the static and dynamic graph algorithms literature, we believe this construction may be of independent interest.


\end{abstract}


\pagenumbering{arabic}
\section{Introduction}

Cut-based graph problems lie at the heart of combinatorial optimization and theoretical computer science. They have been extensively studied over the years from both theoretical and practical perspectives, with applications in areas such as VLSI design, community detection, clustering, and graph partitioning~\cite{Shmoys:1996aa}. Notable examples include the $s$-$t$ minimum cut (and its dual $s$-$t$ maximum flow), global minimum cut, and sparsest cut problems. Research efforts to solve these problems have led to the development of several algorithmic gems, such as Karger's global min-cut algorithm based on tree packing~\cite{Karger:2000aa}, the seminal work on expander flow of Arora, Rao, and Vazirani~\cite{Arora:2009aa} for approximating the sparsest cut, and interior point methods for solving structured linear programs arising from flow problems on graphs \cite{Daitch:2008aa,Madry:2016aa,Brand:2021aa,Dong:2022aa,Brand:2023aa,Brand:2023ab}. 

Designing algorithms for cut-based problems that are both fast and offer competitive approximation guarantees is a recurring theme in algorithm design. Key to enabling the majority of these fast algorithms is the algorithmic framework of approximating graphs with simpler ones. One of the most powerful notions in this context is that of a \emph{tree flow sparsifier}~\cite{Racke:2002aa}; a tree $T$ that approximates the values of all cuts of an undirected graph $G$ within a factor $q$, known as the \emph{quality} of the sparsifier. A closely related concept is the \emph{$j$-tree decomposition}~\cite{Madry:2010ab}, which is a family of $j$-trees that together approximate the cut structure of $G$ (for a formal definition, see \Cref{def:jtree}). The seminal work of R{\"{a}}cke~\cite{Racke:2002aa} showed the existence of tree flow sparsifiers in the context of oblivious routing. Since then, a series of works~\cite{Bienkowski:2003aa,Harrelson:2003aa,Racke:2008aa,Madry:2010ab,Racke:2014aa,Goranci:2021ab, Li:2025aa} has resulted in fast constructions of both tree flow sparsifiers and $j$-tree decompositions of poly-logarithmic quality that run in close to linear time.\footnote{Concretely, \cite{Racke:2014aa,Li:2025aa} construct a tree flow sparsifier of poly-logarithmic quality in $\tilde{O}(m)$ time, and \cite{Madry:2010ab} gives a $j$-tree decomposition algorithm with quality $\tilde{O}(\log n)^{L}$ in $O(m^{1+2/L})$ time, for any integer $L \geq 1$.} These structures have led to polynomial-time approximation algorithms for many fundamental problems, including minimum bisection, $k$-multicut, and min-max partitioning. They have also played a pivotal role in recent breakthrough advances in fast algorithms for both undirected~\cite{sherman2013nearlymaximumflowsnearly,Kelner:2014aa,peng2015approximateundirectedmaximumflows,Sherman:2017aa} and directed~\cite{Chen:2023aa} $s$-$t$ maximum flow problems.

Motivated by the wide applicability of tree flow sparsifiers and $j$-tree decompositions, there has been a recent surge of interest in designing dynamic algorithms that maintain these structures under edge insertions and deletions. Although constructing such structures is already non-trivial in the static setting, Goranci, R{\"{a}}cke, Saranurak, and Tan~\cite{Goranci:2021ab} were able to obtain a dynamic tree flow sparsifier for unit capacity graphs, achieving quality $n^{o(1)}$ and handling edge updates in $n^{o(1)}$ update time. This result was later extended to weighted graphs by van den Brand, Chen, Kyng, Liu, Meierhans, Gutenberg, and Sachdeva~\cite{brand2024almostlineartimealgorithmsdecremental}, achieving the same guarantees. In parallel, Chen, Goranci, Henzinger, Peng, and Saranurak~\cite{Chen:2020aa} designed a fully dynamic algorithm for maintaining $j$-tree decompositions with quality $\tilde{O}(\log n)$ and update time $\tilde{O}(n^{2/3})$. Unfortunately, if we seek to approximate cuts with poly-logarithmic quality, none of these approaches currently overcome the $n^{2/3}$ update time barrier. This raises the following natural question: 
\begin{center}
\emph{Is there a fully dynamic algorithm for cut-based optimization problems that achieves \emph{poly-logarithmic} approximation with very small polynomial update time?}
\end{center}

We answer this question in the affirmative by designing an algorithm that dynamically maintains a variant of the \emph{hierarchical} $j$-tree decomposition of Madry~\cite{Madry:2010ab}. As a consequence, we obtain new trade-offs between approximation and update time for fundamental cut-based optimization problems in the fully dynamic setting, including all-pairs minimum cuts, sparsest cut, multi-way cut, and multi-cut. 

\paragraph{Our results: The Dynamic Hierarchical $j$-tree Decomposition}
In order to be able to state our main result, we need to first review some notation. We start with the notion of $j$-trees due to Madry~\cite{Madry:2010ab}. 
For an integer $j \geq 1$, an undirected graph $H$ is called a \emph{$j$-tree} if it is the union of a (1) rooted forest $F \subseteq H$ with root set $R$ such that $|R| \leq j$ and a (2) a \emph{core} $C(H) = H[R]$, which is the induced subgraph on the root set $R$. This notion generalizes that of a tree, since every $1$-tree is a tree, and as we will show, it offers a natural trade-off between the complexity of the decomposition structure and the time required to maintain it.

We next formalize what it means for a collection of $j$-trees to approximately preserve the cut structure of graphs. Given an undirected graph $G = (V, E, u)$ with an edge capacity function $u$,  and a collection $\cal{H}$ of $j$-trees whose vertex set is $V$, we say $\cal{H}$ \emph{$\alpha$-preserves} the cuts of $G$ with probability $p$, if 
\begin{itemize}
    \item[(1)]  for every cut $(S, V \setminus S)$, we have $U_G(S) \leq \min_{H \in \cH} U_H(S)$,\footnote{For an undirected, capacitated graph $G= (V, E,u)$, and any subset of vertices $\emptyset \neq S \subset V$, we let $U_G(S)$ denote the total capacity of edges crossing the cut $(S, V \setminus S).$} and
    \item[(2)] for any cut $S$, $\min_{H \in \cH} U_H(S) \leq \alpha \cdot U_G(S)$ with probability at least $p$. 
\end{itemize}

We remark that $\cal{H}$ does not preserve exponentially many cuts of $G$. Instead, it only guarantees that \emph{every fixed} cut is preserved with probability $p$. Since in our construction $\cal{H}$ will be used with $p = 1-1/\textrm{poly}(n)$, we can always use $\cal{H}$ to preserve any set of polynomially many cuts by simply applying a union bound over them. 

Our main contribution is the first fully dynamic algorithm for maintaining a collection of \emph{hierarchical} $j$-trees that approximate the cut structure of graphs, as summarized in the theorem below.  

\begin{restatable}{theorem}{main} \label{th:main}
Let $c \ge 1$ be a constant.
Given an \(n\)-vertex capacitated, undirected graph \(G = (V, E, u)\) with capacity ratio\footnote{Recall that the \emph{capacity ratio} of a graph is the ratio between the largest and the smallest edge capacity.} \(U = \poly{n}\) undergoing edge insertions and deletions, and integers \(j \geq 1\) and  \(1 \leq L  \leq \littleo{ \sqrt{ \log (n/j) / \log \log n }}\),
there is a data structure that 
maintains
a collection \(\mathcal H\) of \atmost{j}-trees of \(G\)
that $\Otil(\log n)^L$-preserves the cuts of \(G\) with probability at least $1-1/n^{c}$.

The data structure explicitly maintains \(\mathcal H\)
in \((n/j )^{2/L} \log^{\atmost{L}} n \) amortized update time,
guarantees that \(|\mathcal H| = \log ^{\atmost{L}} n\)
and that the capacity ratio of each \atmost{j}-tree \(H \in \mathcal H\) is \(n^{2L+1}U\log ^{\atmost{L}} n\).
%
\end{restatable}

For any constant $\epsilon = O(1/L) \in (0,1)$, our result achieves poly-logarithmic quality in the cut approximation and $\tilde{O}(n^{\epsilon})$ amortized update time. In contrast, state-of-the-art dynamic tree flow sparsifier constructions~\cite{Goranci:2021ab,brand2024almostlineartimealgorithmsdecremental} are based on dynamically maintaining an \emph{expander hierarchy}. This is a bottom-up clustering approach that proceeds by repeatedly (1) finding a variant of an expander decomposition,\footnote{The so-called \emph{boundary-linked} expander decomposition~\cite{Goranci:2021ab}.} (2) contracting the identified expanders, and (3) recursing on the resulting contracted graph. However, similar to other hierarchical bottom-up clustering techniques~\cite{Alon:1995aa,Forster:2019aa,Chechik:2020aa}, these constructions can only guarantee \emph{sub-polynomial} approximation quality. 

In the static setting, a recent result by Li, Rao, and Wang~\cite{Li:2025aa} addressed this limitation by proposing a different bottom-up hierarchical construction that achieves poly-logarithmic approximation. Yet, their method appears inherently sequential and difficult to extend to the dynamic setting.

Finally, we emphasize that our algorithm avoids the use of expander decompositions or their dynamic maintenance. Instead, it relies solely on fast static algorithms for computing spanning trees with low average stretch~\cite{Abraham:2019aa}, which we use as a black-box subroutine.

\paragraph{Applications.} We use Theorem~\ref{th:main} to obtain new dynamic algorithms for fundamental cut-based problems such as all-pairs minimum cuts, sparsest cut, multi-way cut, and multi-cut. At a high level, all of these algorithms exploit the fact that the core of $j$-trees is much smaller compared to the size of the underlying graph and that these optimization problems in the forest of \(j\)-trees admit simpler solutions in the dynamic setting. The exact guarantees for these applications are given in the theorem below.

\begin{theorem} \label{thm:applications}
Let $c \geq 1$ be a constant. Given an $n$-vertex capacitated, undirected graph $G=(V,E,u)$ with capacity ratio \(U = \poly{n}\) undergoing edge insertions and deletions, and any constant $\epsilon \in (0,1)$, there is a data structure that supports updates in $\tilde{O}(n^{\epsilon})$ amortized time and returns a \emph{poly-logarithmic} approximation, with probability at least $1-1/n^{c}$, to the following queries:
\begin{itemize}
    \item[(1)] for any $(s,t) \in V^2$: $s$-$t$ minimum cut value (and thus the  $s$-$t$ maximum flow value), i.e., all-pairs minimum cut; and 
    \item[(2)] sparsest cut value; and
    \item[(3)] multi-way cut, multi-cut values.
\end{itemize}

The first two queries are supported in amortized time $\tilde{O}(n^{\epsilon})$. The last queries are supported in $\tilde{O}(kn^{\epsilon})$ amortized time for multi-way cut, where  $k$ is the number of terminals in the multi-way cut, and $\tilde{O}(kn^{\epsilon} + k^2)$ amortized time for multi-cut, where $k$ is the number of required sets in the multi-cut.\footnote{We can also report the corresponding approximate $s$-$t$ min-cut and sparsest cut in time proportional to the smaller size of the cut.}
\end{theorem}

Previous fully dynamic algorithms for all-pairs minimum cuts achieve either:
(1) a $\tilde{O}(\log n)$ approximation with $\tilde{O}(n^{2/3})$ amortized update and query time~\cite{Chen:2020aa}, or
(2) an $n^{o(1)}$ approximation with $n^{o(1)}$ worst-case amortized update time and $\tilde{O}(1)$ query time~\cite{Goranci:2021ab,brand2024almostlineartimealgorithmsdecremental}. 
Our result is the first to break the $n^{2/3}$ update time barrier while still retaining poly-logarithmic approximation guarantees. 

For sparsest cut, multi-way cut, and multi-cut, all prior fully dynamic algorithms achieved only $n^{o(1)}$ approximation~\cite{Goranci:2021ab,brand2024almostlineartimealgorithmsdecremental}. Our algorithms are thus the first to provide poly-logarithmic approximation guarantees for these problems in the fully dynamic setting.


\section{Technical Overview}

\paragraph{Cut-Based Graph Decompositions.} Many cut-based problems become easier to solve on trees than on general graphs. Building on this insight, Räcke~\cite{Racke:2008aa} showed that, for any $m$-edge undirected graph $G$, there is a polynomial time algorithm to compute a collection of $k := \Otil(m)$ trees $T_1, \ldots, T_k$ along with a probability distribution $\{\lambda_i\}_i$ over such trees such that, for any cut $S$ and a tree $T$ sampled from the distribution, it holds that $U_G(S) \leq U_{T}(S) \leq O(\log n) \cdot U_G(S)$ with constant probability. This result paved the way for simple and efficient approximation algorithms for a broad class of cut-based problems by using the following natural algorithmic approach: (a) sample a small number of trees from the distribution, and (b) solve the problem efficiently on those trees.

In his seminal paper, Madry~\cite{Madry:2010ab} introduced the notion of $j$-trees and showed that there is a natural trade-off between structural and computational complexity. Specifically, Madry showed how to compute a collection of only $\Otil(m/j)$ many $j$-trees, along with a probability distribution over them such that, cuts are preserved up to $\tilde{O}(\log n)$ factor, which is slightly worse than the guarantee in the tree-based setting. Importantly, solving cut-based problems on a $j$-tree typically reduces -- with linear overhead in runtime -- to solving them on their so-called \emph{cores} $C$,  which are subgraphs induced on just $j < n$ vertices, rather than the full vertex set. 

Similar to~\cite{Racke:2008aa}, the algorithm in~\cite{Madry:2010ab} computes the collection $H_1, \ldots, H_k$, $k=\tilde{O}(m/j)$ of $j$-trees using the iterative scheme of \cite{young}, which itself is a variant of the multiplicative weight update (MWU) method \cite{arora2012multiplicative}. In each iteration, a low-stretch spanning tree $T_i$ is computed with respect to some maintained length function on the edges of $G$, using the algorithm of~\cite{AbrahamBN08, Abraham:2019aa}. Madry then employs a rather involved bucketing argument to identify a set $E_i$ of $\Otil(j \log U)$ high-congestion edges in $T_i$ such that, upon removal, the capacity function of the resulting forest $F_i$ has desirable properties that reduce the number of iterations of this iterative scheme to $\Otil(m/j)$. Using the sequence of obtained forests $\{F_i\}$, the algorithm in~\cite{Madry:2010ab} follows a $3$-step construction process on each forest $F_i$: it first builds two different intermediate graphs, which then serve as the foundation for the final construction of the $j$-tree $\{H_i\}$.

In this paper, we present a more modular implementation of Madry's construction. Concretely, we show that with a careful modification of Madry's algorithm,\footnote{We directly bound the \emph{width} of the MWU oracle call.} a much more standard application of the MWU framework suffices. Our approach also simplifies the overall design by avoiding the complicated bucketing scheme of~\cite{Madry:2010ab}; instead, we directly remove the $\Otil(j)$ edges with the highest congestion in $T_i$. Finally, drawing inspiration from \cite{chen2022maximumflowminimumcostflow}, we show that Madry's $3$-step construction can be bypassed entirely: $j$-trees $H_i$ can be computed directly from the forests $F_i$ by leveraging the so-called \emph{branch-free} root sets.

\paragraph{Dynamic Cut-Based Decompositions.} 

Chen et al.~\cite{Chen:2020aa} leveraged the powerful paradigm of vertex sparsification to dynamically maintain Madry's graph decompositions into $j$-trees. A key operation that their algorithm supports is the ability of moving vertices $v$ to the core of $O(j)$ trees 
-- that is, increasing the core $C$ of the $O(j)$-tree to also contain vertex $v$ -- without compromising the cut-approximation guarantees. This operation alone is sufficient to support dynamic edge updates; whenever an edge is updated, its endpoints are simply moved into the core.

 To support the operation of moving vertices to the core, the algorithm in~\cite{Chen:2020aa} relies on two main components: (i) re-routing the capacities of the maintained $O(j)$ tree forests and (ii) expanding the core by adding extra vertices (which can be viewed as computing a hitting set over the forests) so that the leaf to root paths in the forests $\{F_i\}$ are short. 
 Another contribution of our work is that we eliminate the need for (i) and (ii); we avoid re-routing entirely by showing that as the number of components in the forests increases, the embedding quality of the resulting $j$-trees cannot decrease. Furthermore, we employ refined amortization arguments to sidestep the need for augmenting the core with additional vertices. In \Cref{sec:madry}, we prove the following theorem.

\begin{theorem}\label{thm:overview_jtree}
    Let $j$ be a parameter and let $G$ be a dynamic graph that initially contains $m$ edges and undergoes up to $O(j)$ edge insertions and deletions. There exists an algorithm that initializes and maintains a set of $i = 1, 2, \ldots, k = \Otil(m/j)$ many $O(j)$-trees $H_i = C_i \cup F_i$ with cores $C_i$ and forests $F_i$. At all times, it holds that:
    \begin{enumerate}
        \item for all $i$ and all cuts $S$: $U_G(S) \leq U_{H_i}(S)$,
        \item for all cuts $S$: $\frac{1}{k} \sum_{i=1}^k U_{H_i}(S) \leq \Otil(\log n) \cdot U_G(S)$, and
        \item after $t$ edge updates have been performed on $G$, we have $|C_i^{(t)}| \leq |C_i^{(0)}| + O(t)$ for all $i$.
    \end{enumerate}
    The algorithm runs in total time $\Otil(m^2 / j)$.
\end{theorem}

\paragraph{A Two-Level Scheme.} Building on the guarantees of \Cref{thm:overview_jtree}, consider a natural two-level scheme: let $\cH$ be a set consisting of $O(\log n)$ many $j$-trees sampled uniformly at random from~\Cref{thm:overview_jtree}. Then, item $2$ of the theorem together with Markov's inequality implies that the set $\cH$ $\Otil(\log n)$-preserves the cuts of $G$ with probability $1-1/\poly{n}$. By running a dynamic (edge) cut-sparsification algorithm from \cite{7782947} or \cite{bernstein_et_al}, we can furthermore ensure that the cores $C$ of the graphs $H \in \cH$ contain only $\Otil(j)$ many edges. 

This combination of random subsampling and sparsification yields a powerful size reduction. For example, an $s$–$t$ min-cut query in $G$ reduces to solving $s$–$t$ min-cut within the sparse cores of the $O(j)$-trees in $\cH$. Each such core is a graph with $O(j)$ vertices and $\Otil(j)$ edges, and there are only $O(\log n)$ of them. As a result, the query complexity reduces to solving a small number of subproblems on graphs of size $\Otil(j)$. By setting $j = n^{2/3}$, we recover the dynamic algorithm of~\cite{Chen:2020aa}, which achieves $\Otil(\log n)$ approximation with $\Otil(n^{2/3})$ amortized update- and query time.
\begin{theorem}\label{theorem:overview_2level}
    Given a parameter $j \geq 1$, we can dynamically maintain a set $\cH$ such that \begin{itemize}
        \item[(1)] $\cH$ $\Otil(\log n)$-preserves the cuts of $G$ with probability $1 - 1/\poly{n}$, and 
        \item[(2)] $\cH$ consists of $O(\log n)$ many $O(j)$-trees. 
    \end{itemize} 
    Moreover, $\cH$ can be maintained in $\Otil(m^2 / j^2)$ amortized time.
\end{theorem}

\paragraph{A First Attempt at a Hierarchical Scheme.} 
A natural attempt to improve the update time of our algorithm is to apply \Cref{theorem:overview_2level} recursively on the cores of the $O(j)$ trees, as outlined below.  Assume that $G$ has $\Otil(n)$ edges, and let $k = (n/j)^{1/2}$ be a size-reduction parameter. We begin by sampling and maintaining a set $\cH_1$ of $O(\log n)$ many $n/k =: j_1$-trees using \Cref{theorem:overview_2level}. The set $\cH$ $\Otil(\log n)$-preserves the cuts of $G$ with high probability. Next, we apply a dynamic cut-sparsification algorithm on the cores $C_1$ of the graphs $H_1 \in \cH_1$ to maintain sparsified cores $\Tilde{C}_1$ that each contain only $\Otil(j_1)$ many edges. 

Then, for each $H_1 \in \cH_1$, we recursively apply the algorithm from \Cref{theorem:overview_2level} on the (sparsified) cores $\Tilde{C}_1$ of $H_1$ to receive a set $\cH(H_1)$ of $O(\log n)$ many $n/k^2:=j_2$-trees. The set $\cH(H_1)$ $\Otil(\log n)$-preserves the cuts of the sparsified core $\Tilde{C}_1$. We refer to these graphs as the children of $H_1$, and let $\cH_2 := \bigcup_{H_1 \in \cH_1} \cH(H_1)$.

Now consider $H_1 = C_1 \cup F_1 \in \cH_1$ and $H_2 = C_2 \cup F_2 \in \cH(H_1)$, with $\Tilde{C}_2$ as the sparsified core of $H_2$. The graph 
\[H^{H_1, H_2} := \Tilde{C}_2 \cup F_2 \cup F_1\]
is then an $O(j_2)$-tree, whose core is $\Tilde{C}_2$ and whose forest is $F_1 \cup F_2$. If we bring together all such graphs in the set 
\[\cH := \{H^{H_1, H_2}: H_1 \in \cH_1, H_2 \in \cH(H_1)\},\]
one can show that $\cH$ $\Otil(\log^2 n)$ preserves the cuts of $G$ with probability $1 - 1 / \poly{n}$. 

The initialization of $\cH_1$ takes time $\Otil(n^2 / j_1) = \Otil(n \cdot (n/j)^{1/2})$. Similarly, initializing the set $\cH_2$ requires $O(\log n) \cdot \Otil(j_1^2 / j_2) = \Otil(n \cdot (n/j)^{1/2})$ time. Thus, the total time to initialize $\cH$ is $\Otil(n \cdot (n/j)^{1/2})$, which already constitutes a significant improvement over the initialization time of $\Otil(n^2 / j)$ achieved by \Cref{theorem:overview_2level}. 

In general, by instead setting $k = (n/j)^{1/L}$ and repeating this procedure $L$ times, we can construct a set $\cH$ that $\Otil(\log n)^L$-preserves the cuts of $G$ with probability $1 - 1 / \poly{n}$ such that $\cH$ consists of $O(\log^{O(L)})$ many $O(j)$-trees; the algorithm takes initialization time $O(n \cdot (n/j)^{2/L} \cdot \log^{O(L)} n)$.

Unfortunately, this approach breaks down as soon as we move to the dynamic setting, where edge insertions and deletions to $G$ must be handled. Consider again the case $L = 2$ from above. A single update to $G$ can cause the cores  $C_1$ of the $O(j_1)$ trees from \Cref{thm:overview_jtree} to change by up to $\Omega(j_1)$ edge insertions and deletions. Even if we use the cut sparsifiers from \cite{7782947} or \cite{bernstein_et_al} as black-boxes to handle these updates, the sparsified cores $\Tilde{C}_1$ may still change by $\Omega(j_1)$ edge insertions and deletions. These $\Omega(j_1)$ updates must then be propagated to the data structure from \Cref{thm:2level} that runs on $\Tilde{C}_1$ to maintain the set $\cH_2$. This incurs \[\Omega(j_1 \cdot j_1^2 / j_2^2) = \Omega(n \cdot (n/j)^{1/2})\] amortized time, which is worse than the original $\Otil(n^2 / j^2)$ update time from \Cref{theorem:overview_2level}. In fact, for any choice of $j$, this simple recursive scheme results in an update time of $\Omega(n)$, which defeats the purpose of the recursion.

\paragraph{Dynamic Cut-Sparsification under Vertex Splits to The Rescue.} 
We have already seen that running a cut sparsification algorithm on the cores $C$ of the $j$-trees is essential. While the cores are only on $O(j)$ vertices, they might still contain $\Theta(m)$ many edges, preventing us from efficient recursion unless we employ an edge-sparsification procedure. However, as our earlier calculations reveal, merely maintaining a sparse cut-sparsifier dynamically is not sufficient. In order to make the hierarchical construction efficient in the dynamic setting, we also crucially need to develop a cut sparsification algorithm that \emph{sparsifies recourse}. Specifically, while a single update to $G$ might cause $\Omega(j)$ changes in the dense core $C$, we show how to efficiently maintain a sparsified core $\Tilde{C}$ that undergoes only  $O(\polylog(n))$ many edge updates after processing the updates in $C$. Consequently, the recourse of $\Tilde{C}$ is significantly smaller than the number of edges that are moved in the underlying dense graph $C$. This crucially enables us to avoid suffering from a large blow-up in recourse on higher levels in the hierarchy.

In order to achieve this, we exploit the structure of the updates of the cores.  Although $C(H)$ may experience $\Omega(j)$ edge insertions and deletions per update, these changes have a highly structured form. Specifically, they correspond to \emph{vertex splits}, a type of structured update that was also crucially exploited in several hierarchical graph-based data structures \cite{Forster:2019aa, Chechik:2020aa, Goranci:2021ab, chen2022maximumflowminimumcostflow}. 
 
\begin{definition}[Vertex Split]
    Given a $G = (V, E)$ and a set of vertices $U \subseteq N_G(v)$ satisfying $|U| \leq |N_G(v)|/2$. Then we say that the vertex $v$ is split in $G$ if the vertex set of $G$ changes to $V \cup \setof{v'}$, such that $N_{G}(v) = N_G(v) \setminus U$ as well as $N_{G}(v') = U$.
\end{definition}
In other words, a vertex split moves a subset $U$ of $v$'s neighbors to a newly created vertex $v'$. We then establish a key structural property: the minimum spanning forest $F'$ of a graph $G'$ obtained by a vertex split differs from the original minimum spanning forest $F$ by at most a single extra edge.

This crucial observation serves as the main building block of our approach. Leveraging it, we can modify the dynamic cut-sparsifier data structure of \cite{7782947}. The algorithm maintains so-called bundles $B$ of spanning forests packed into the graph $G$. It then alternates between sub-sampling the edges that are not in the maintained spanning forests and again maintaining a packing for the remaining edges.

By exploiting the low-recourse property of spanning forests under vertex splits, we show that it is possible to maintain such a packing with extremely low recourse. This, in turn, leads to a dynamic cut-sparsification algorithm with similarly low recourse guarantees. We state the main theorem below. Due to the versatility of dynamic cut-sparsification algorithms, we believe this result is of independent interest. See~\Cref{sec:cut-sparsifier} for a detailed explanation. 
\begin{theorem}\label{thm:overview_cutsparsify} Let $G$ be a fully dynamic graph on initially $n$ vertices and $m$ edges. Then, for a sequence of up to $m$ edge insertions, $n$ vertex splits and $D$ edge deletions, there is a randomized algorithm that explicitly maintains a $(1 \pm \epsilon)$-cut sparsifier $\Tilde{G}$ of $G$ on $\Otil(n)$ many edges. 

    The algorithm takes $\Otil(m)$  total time and the sparsifier $\Tilde{G}$ undergoes $O(\polylog(n)) \cdot (n+D)$ total recourse (i.e., number of edge deletions and edge insertions) for all the updates. The algorithm is correct with high probability.
\end{theorem}
We again emphasize that while a vertex split causes up to $\Omega(|N(v)|)$ edges to move, the (amortized) recourse in the sparsifier is limited to $O(\polylog(n))$, which is independent of $|N(v)|$. This constitutes a key improvement over directly using \cite{7782947}, which would only guarantee a recourse of $\Otil(|N(v)|)$ -- unacceptably large for our purposes, as noted earlier. 

When applied to the graphs from \Cref{theorem:overview_2level}, we are able to show the following result.
\begin{lemma}\label{thm:overview_sparsifiedcore}
    Let $H$ be an $O(j)$-tree maintained by \Cref{theorem:overview_2level}, and $C = C(H)$ be its core. Then we can initialize and maintain a sparsified core $\Tilde{C}(H)$ over the full sequence of $O(j)$ updates to $G$ in total time $\Otil(m)$ and total recourse (i.e., edge insertions/deletions in $\Tilde{C}(H)$) of $\Otil(j)$. 
\end{lemma}
We can now re-visit the analysis of the case $L=2$ from before. A single update in $G$ first triggers an update to the data structure from \Cref{theorem:overview_2level}, which maintains the set $\cH_1$. Processing the update takes amortized time $\Otil(n^2 / j_1^2) = \Otil(n/j)$. The cores $C_1$ of the graphs $H_1$ then change by $O(1)$ vertex splits. These in turn are then processed by \Cref{thm:overview_sparsifiedcore} to maintain sparsified cores $\Tilde{C}_1$ in amortized time $\Otil(n/j_1)$. After processing the update, $\Tilde{C}_1$ will have changed by $O(\polylog(n))$ edge insertions and deletions. Finally, these update are then forwarded to the data structure of \Cref{theorem:overview_2level} to maintain the $O(j_2)$-trees in $\cH_2$. This is done in amortized time $\Otil(j_1^2 / j_2^2) = \Otil(n/j)$. In total, the amortized update time of the whole hierarchy can be bounded as $\Otil(n/j)$ -- already significantly improving the anortized update time of $\Otil(n^2 / j^2)$ from the two-level scheme.

With some care, the above analysis extends to a hierarchy consisting of $L \geq 3$ levels. Similar to before, the key insight is that, rather than incurring a recourse of $\Omega(j_L j_{L_-1} \ldots j_1)$ as would arise by a naive cut-sparsification approach, the use of our low-recourse sparsifier from \Cref{thm:overview_cutsparsify} allows us to bound the additional recourse at each level by only $O(\log^{O(L)}n)$. Additionally, we periodically rebuild parts of the hierarchy to ensure the cores remain within the desired size. We perform the full analysis of the hierarchical construction in \Cref{sec:hierarchy}.

\paragraph{Applications.}
To answer \(s\)-\(t\)-min-cut, multi-way-cut, and multi-cut queries, it suffices to move the relevant vertices into the cores \(C(H)\) of all \(O(j)\)-trees \(H \in \mathcal{H}\), and then apply static algorithms to solve these problems within the cores. In contrast, maintaining query access to an approximate sparsest cut value is technically more challenging. This difficulty arises because the sparsest cut is a \emph{global property} of the entire \(j\)-tree \(H\), and thus cannot be resolved by restricting computation to the core alone.

We address this challenge by observing that any sparsest cut in \(H\) must either (a) cut only core edges or (b) cut only forest edges. Based on this observation, we separately maintain the sparsest core-edges-only cut and forest-edges-only cut. Upon query, we can use the standard sparsest cut approximation algorithm to compute an approximate sparsest cut inside the core. At the same time, we must also track the sparsest cut value in the dynamically changing forest \(F\), which spans \(\Omega(n)\) vertices. To this end, we extend the $2$-level scheme from \Cref{theorem:overview_2level} in a careful manner to maintain this value efficiently. Finally, upon query we return the minimum of these two candidate values. See \Cref{sec:applications}, \Cref{appendix:multicut} and \Cref{appendix:sparest-cut-proof} for details.
\subsection{Related Works}
\label{sec:related-works}

\paragraph{Related Technical Works.} Our dynamic algorithm closely resembles the fully dynamic approximate min-ratio cycle data structure from \cite{chen2022maximumflowminimumcostflow}. Conceptually, one can view our approach as an analogue for dynamic $\ell_{\infty}$ approximation problems of their dynamic $\ell_1$ optimization framework, albeit while avoiding many of the technical difficulties. Whereas their data structure required an algorithm for dynamically maintaining a low-recourse spanner under vertex splits, we instead need to maintain a low-recourse cut-sparsifier under vertex splits. Because they also need to explicitly maintain short path embeddings, their sparsification procedure is much more involved than ours. Our algorithm also is structurally similar to the fully dynamic APSP data structure from \cite{KyngMG24}, although again we avoid most of the difficulties as our cut-sparsifier does not require the input graphs to be of low degree.

\paragraph{Partially Dynamic Maximum Flow.} For incremental graphs, it is known how to obtain $(1+\epsilon$)-approximation for undirected $p$-norm flows~\cite{BrandCKLPGSS24} and directed min-cost flow~\cite{chen2024almost} in $O(m^{1+o(1)}\epsilon^{-1})$ total update time. A previous work of Goranci and Henzinger~\cite{GoranciH23} have obtained $(1+\epsilon)$-approximation in $O(m^{3/2 + o(1)} \epsilon^{-1/2})$ total update time for directed unweighted maximum flow. A variant of their algorithm can also be extended to maintain an exact solution in $O(n^{5/2 + o(1)})$ total time, albeit only for undirected graphs. Similarly, for the exact case, Gupta and Khan~\cite{GuptaK21} have shown how to maintain a directed maximum flow in $\tilde{O}(mF^*)$ total time, where $F^*$ stands for the max-flow value at the end of the insertion sequence. Recently, Goranci, Henzinger, R{\"{a}}cke, and Sricharan~\cite{GoranciHRS25} has shown an algorithm for $(1+\epsilon)$-approximate $s$-$t$ maximum flow in undirected, uncapacitated graphs in $O(m+nF^* \epsilon^{-1})$ total time which is the first algorithm for the problem with poly-logarithmic amortized update time for dense graphs. For decremental graphs, van den Brand, Chen, Kyng, Liu, Meierhans, Gutenberg, and Sachdeva~\cite{brand2024almostlineartimealgorithmsdecremental} have shown a $(1+\epsilon)$-approximate directed min-cost flow algorithm with $O(m^{1+o(1)} \epsilon^{-1})$ total update time.

\paragraph{Dynamic Algorithms for Global Minimum Cut.} For incremental graphs Goranci, Henzinger and Thorup \cite{GoranciHT16} have shown a dynamic algorithm exactly maintaining the size of a global minimum-cut in $\tilde{O}(1)$ amortized update time. Known results for the fully dynamic setting obtain weaker guarantees. Goranci, Henzinger, Nanongkai, Saranurak, Thorup, and Wulff{-}Nilsen~\cite{GoranciHNSTW23} have obtained an exact algorithm for the problem with $\tilde{O}(n)$ worst-case and $\tilde{O}(m^{30/31})$ amortized update time with and without randomization respectively. In the approximate regime the state of the art result was shown by El-Hayek, Henzinger and Li~\cite{El-HayekH025} which achieves $(1+o(1))$-approximation in $O(n^{o(1)})$ amortized update time.

\paragraph{Fast Static Algorithms for Cut-Based Problems.} The sparsest cut problem has received significant attention in the static setting \cite{AlonM85, AndersenP09, LeightonR99,AroraHK10, KhandekarRV09, AroraK16, OrecchiaSVV08}. The classical work of Arora, Rao, and Vazirani~\cite{Arora:2009aa} presented the first algorithm for the problem which achieves $O(\sqrt{\log n})$-approximation for the problem in polynomial time. Sherman \cite{Sherman:2009aa} has shown that in order to obtain $O(\sqrt{\epsilon^{-1}\log n })$-approximation for the problem its sufficient to solve $\tilde{O}(n^{\epsilon})$ maximum flow instances. An influential sparsifier related to cut-based problems is the Gomory-Hu tree \cite{gomory1961multi}, which is a weighted tree on the nodes of of the graph preserving any minimum $s-t$ cut values. It has been shown to be a near-optimal data structure for answering $s-t$ cut queries by Abboud, Krauthgamer and Trabelsi \cite{AbboudKT21}. For more than 50 years no faster algorithm was known for finding such trees than the initial paper of Gomory and Hu computing the structure through $n-1$ max-flow computations. A recent line of work has reduced its construction time to $O(n^{2+o(1)})$ \cite{AbboudKT21b, AbboudKT21a, 0006PS21, Zhang22}, also showing that the all pairs min-cut problem can be solved in the same amount of time. Finally, Abboud, Panigrahi, and Saranurak~\cite{Abboud:2023aa} achieved a construction time of \(m^{1+\littleo{1}}\) for Gomory-Hu trees by employing the almost-linear time max-flow algorithm from \cite{chen2022maximumflowminimumcostflow}.


\section{Preliminaries}
\paragraph{Graph Theory.}
Throughout this paper, all the graphs $G = (V, E_G)$ we study will be undirected multi-graphs. We denote by $n = |V|$ its number of vertices, and $m = |E_G|$ its number of edges. Most of the times, the graphs are capacitated, i.e., equipped with a function $u_G \in \R_{>0}^{E}$ on their edges. For a capacitated graph, we usually denote by $U$ the capacity ratio of $G$, defined as $U = \max_{e, e' \in E_G} u_G(e) / u_G(e')$. In our main results, we will assume that $U$ is polynomially bounded, i.e., that $U = O(\poly{n})$. For a vertex $v \in V$ we denote by $N_G(v)$ the neighborhood of $v$ in $G$, and treat it as a multi-set. Unless otherwise mentioned, the basic data structure we will use to access and modify graphs will be an adjacency list representation of the graph.

For a cut $(S, V \setminus S)$ of a graph $G$, we define $E_G(S, V \setminus S)$ as the set of edges of $G$ with exactly one endpoint in the set $S$, and also define $U_G(S) := \sum_{e \in E_G(S, V \setminus S)} u_G(e)$ as the total capacity crossing the cut.

Given two graphs $G_1 = (V, E_1, u_1), G_2 = (V, E_2, u_2)$ on the same vertex set, scalars $\alpha, \beta > 0$, we define the linear combination graph $G := \alpha G_1 + \beta G_2$ as a graph on vertex set $V$ with edge set $E_G = E_1 \cup E_2$ and capacities $u_G(e) = \alpha u_{1}(e) + \beta u_{2}(e)$. 

Let $G_1 = (V_1, E_1, u_1), G_2 = (V_2, E_2, u_2)$ be two graphs such that $V_2 \subseteq V_1$, i.e., $G_2$ is defined on a subset of the vertices of $G_1$. Then we define the union $G := G_1 \cup G_2$ as a (multi-)graph on vertex set $V_1$ and edge set $E_1 \cup E_2$, where the capacities of the multi-edges correspond to their capacities in $G_1$ or $G_2$ respectively.

\paragraph{Adversary Model.} All our results assume an \emph{oblivious adversary}, i.e., that the full update sequence is independent of the outputs of the algorithm as well as the internal randomness used inside the algorithm.

\paragraph{Dynamic Graphs.} 
We call a sequence of graphs $G^{(t)}$ indexed by time $t$ a \emph{dynamic graph}. We will be interested in algorithms that run on a dynamic input graph $G^{(t)}$ such that for all times $t$, $G^{(t+1)}$ differs from $G^{(t)}$ by a single edge insertion or edge deletion. In this paper, whenever we used use superscript $t$, it always denotes the time after all data structures have processed the $t$'th update to $G$. We remark that while the input graph $G$ changes by a single edge insertion / deletion, we will also encounter auxiliary graphs such that $H^{(t+1)}$ differs from $H^{(t)}$ by multiple edges. We will denote the total number of edge insertions / deletions required to perform in order to transform $H^{(t)}$ to $H^{(t+1)}$ as the \emph{recourse} of $H$.

\begin{definition}[Spanning Forest]
    Let $G$ be a graph. Then we call a sub-graph $F$ of $G$ a spanning forest of $G$ if $F$ is a forest and every connected component of $G$ is also a connected component in $F$. If the forest $F$ consists only of a single connected component, we call it a spanning tree.
\end{definition}

\begin{definition}[Rooted Forest]
    We call a graph $F = (V, E_F, u_F)$ a rooted forest with root set $R \subseteq V$ if every connected component $T$ of $F$ is a tree satisfying $|V(T) \cap R| = 1$.

    Given a vertex $v \in V$ in a connected component $T_v$, we let $\root_F(v)$ denote the unique root $r_v \in R \cap V(T_v)$.
\end{definition}
For a rooted forest $F = (V, E_F, u_F)$ and a vertex $v \in V$, we define $v^{\downarrow} \subseteq V$ as the set of vertices who are descendants of $v$ in $F$ (where we agree that $v$ is a descendant of itself).
\begin{definition}[Decremental Rooted Forest]
    We call a sequence of rooted forests $\big(F^{(t)} = (V, E_F^{(t)}, u_F^{(t)}) \big)_{t \in \mathbb{N}}$ with root set $R^{(t)}$ decremental, if for all $t \geq 1$ we have $E_F^{(t+1)} \subseteq E_F^{(t)}$, $R^{(t+1)} \supseteq R^{(t)}$ and $u_F^{(t+1)} \leq u_F^{(t)}$.
\end{definition}
Note that  two consecutive graphs can differ by more than 1 edge.

\paragraph{Graph Embeddings.} We start with the definition of multi-commodity flows. 
\begin{definition}
    Given a graph $G = (V, E_G, u_G)$, we call a collection $\boldsymbol{f} = \setof{f_i}_{1 \leq i \leq k}$ of $s_i$-$t_i$ flows in $G$ a multi-commodity flow. For an edge $e \in E$ we define $|\boldsymbol{f}(e)| := \sum_{i=1}^k |f_i(e)|$ to be the total amount of flow crossing edge $e$. We say the multi-commodity flow $\boldsymbol{f}$ is $\alpha$-feasible if, for every $e \in E_G$, we have $|\boldsymbol{f}(e)| \leq \alpha \cdot u_G(e)$. 
\end{definition}
Given this definition, we can now introduce the concept of embeddability between two graphs. 
\begin{definition}\label{def:embeddability}
    Let $G = (V, E_G, u_G)$ and $H = (V, E_H, u_H)$ be two graphs on the same vertex set and $\boldsymbol{f} = \setof{f_e}_{e \in E_H}$ be a multi-commodity flow in $G$. Then, we call $\boldsymbol{f}$ an embedding of $H$ into $G$ if, for every $e = (s, t) \in E_H$, the flow $f_e$ routes $u_G(e)$ amount of flow between $s$ and $t$.

    If $\boldsymbol{f}$ is $\alpha$-feasible in $G$, then we say that $H$ is $\alpha$-embeddable into $G$ (denoted as $H \preceq_{\alpha} G$). 
\end{definition}
Whenever $H$ is $1$-embeddable into $G$, we simply say that $H$ is embeddable into $G$ and write $H \preceq G$. This notion of embedding allows us to approximate the cut structure of one graph by that of a different, hopefully simpler graph. 
\begin{lemma}\label{lma:cutpreserve}
    Let $G, H$ be two graphs on the same vertex set and suppose that $H \preceq_{\alpha} G$. Then for any cut $(S, V \setminus S)$ we have that $U_H(S) \leq \alpha \cdot U_G(S)$. 
\end{lemma}
\begin{proof}
    Let $\f$ be the embedding from $H$ into $G$. The total amount of flow crossing the cut $(S, V \setminus S)$ is at least the capacity of the cut in $H$, i.e., $U_H(S)$. Furthermore by assumption we know that for each edge $e \in E_G$, we have $|\f(e)| \leq \alpha u_G(e)$. Hence
    \[
    \alpha U_G(S) =  \sum_{e \in E_G(S, V \setminus S)} \alpha u_G(e) \geq \sum_{e \in E_G(S, V \setminus S)} |\f(e)| \geq \sum_{e \in E_H(S, V \setminus S)} u_H(e) = U_H(S).
    \]
\end{proof}
\paragraph{$\alpha$-Cut-Preserving Set.}

\begin{definition}\label{def:cut-preserving-collection}
    Let $G$ be a graph and $\cH$ be a collection of graphs on the same vertex set as $G$. We say the collection $\alpha$-preserves the cuts of $G$ with 
    probability $p$, if 
    \begin{enumerate}
        \item for every cut $(S, V \setminus S)$, we have $U_G(S) \leq \min_{H \in \cH} U_H(S)$, and
        \item for any cut $S$, $\min_{H \in \cH} U_H(S) \leq \alpha \cdot U_G(S)$ with probability at least $p$.
    \end{enumerate}
\end{definition}
    Note that if $\cH$ preserves the cuts of $G$ with probability $1 - 1/n^c$, this does not imply that all exponentially many cuts of $G$ are preserved. It only guarantees that for any fixed cut it is preserved with probability $1-1/n^c$. However, if one considers any set of  at most $n^{c-1}$ cuts, then all of them are preserved simultaneously with probability at least $1-1/n$, which we will call \emph{with high probability (whp)}. In particular, if $c\ge 3$ then the set of $s-t$ min-cuts is preserved simultaneously with high probability .

\paragraph{Cut Sparsification.} We define what it means for a graph $H \subseteq G$ to be a cut sparsifier of $G$.
\begin{definition}
    Let $G = (V, E_G, u_G)$ be a graph and $H = (V, E_H, u_H)$ be a sub-graph of $G$. Let $\epsilon > 0$. Then we say that $H$ is a $(1 \pm \epsilon)$-cut sparsifier of $G$ if for every cut $(S, V \setminus S)$ of $G$, it holds that
    \[
    (1-\epsilon) \cdot U_G(S) \leq U_H(S) \leq (1+\epsilon) \cdot U_G(S).
    \]
\end{definition}
Note that this definition requires the upper and lower bound to hold for all cuts, not just to hold for any cut with high probability.

\paragraph{Standard Scalar Chernoff Bound.} In order to prove guarantees of the dynamic cut-sparsification algorithm in \Cref{sec:cut-sparsifier}, we need the following standard Chernoff Bound.
\begin{theorem}\label{thm:chernoff}
    Let $X_1, \ldots, X_n$ be indepedent random variables taking values in $\setof{0,1}$, and let $X = \sum_{i=1}^n X_i$ be their sum. Let $\mu \geq 0$ be such that $\E{X} \leq \mu$. Then for any $\delta \geq 0$,
    \[
    \Pr[X \geq (1 + \delta) \mu] \leq \exp \rbrack{-\delta^2 \mu / (2+ \delta)}
    \]
\end{theorem}

\paragraph{Dynamic Tree Data Structures.} In this paper, we will need to efficiently maintain quantities over the vertices and edges of dynamically changing trees. The following data structures will allow us to do so.
\begin{lemma}[Link-Cut Tree~\cite{sleator1981data}]\label{thm:treeds}
    Given a rooted forest $F$ with capacities $u_F$, there is a deterministic data structure $\cD_1$ that supports the following operations in amortized time $O(\log n)$:
    \begin{enumerate}
        \item $\textsc{FindRoot}(u)$: Returns $\root_F(u)$, the root of $u$ in $F$.
        \item $\textsc{MakeRoot}(u)$: Makes $u$ the root of the tree in $F$ containinig $u$.
        \item $\textsc{PathMin}(u)$: Returns the edge $e_{\min}^F(u)$ with minimum capacity on the $u$-$\root_F(u)$ path in $F$.
        \item $\textsc{Cut}(e)$: Removes edge $e$ from $F$.
    \end{enumerate}
\end{lemma}

\begin{lemma}[Euler-Tour Tree~\cite{henzinger1999randomized}]\label{thm:ett}
    Given a rooted forest $F = (V, E)$ and a vector $z \in \mathbb{Z}^{V}$, there is a deterministic data structure $\cD_2$ that supports the following operations in worst-case time $O(\log n)$:
    \begin{enumerate}
        \item $\textsc{FindRoot}(u)$: Returns $\root_F(u)$, the root of $u$ in $F$.
        \item $\textsc{MakeRoot}(u)$: Makes $u$ the root of the tree in $F$ containinig $u$.
        \item $\textsc{Cut}(e)$: Removes edge $e$ from $F$.
        \item $\textsc{Sum}(v)$: Returns $\sum_{u \in v^{\downarrow}} z(u)$.
        \item $\textsc{Add}(v, \Delta)$: Updates $z(v) \gets z(v) + \Delta$.
    \end{enumerate}
\end{lemma}

\paragraph{Vertex Splits.}
The construction of our hierarchy requires to change the graph also by \emph{vertex splits}, which are defined as follows.
\begin{definition}[Vertex Split]\label{def:vsplit}
    Given a multi-graph $G = (V, E)$ and a (multi-)set of vertices $U \subseteq N_G(v)$ satisfying $|U| \leq |N_G(v)|/2$. Then we say that the vertex $v$ is split in $G$ if the vertex set of $G$ changes to $V \setminus {v} \cup \setof{v', v''}$, such that $N_{G}(v') = N_G(v) \setminus U$ as well as $N_{G}(v'') = U$.
\end{definition}
The following lemma shows that because the vertex split is encoded by the set $U$, which can be much smaller than the full neighborhood $N_G(v)$, a sequence of $n$ vertex splits can be simulated by $O(m \log n)$ many edge insertions and deletions. That is, we simulate a vertex split by re-naming the vertex $v$ to $v'$, removing the neighbors $U$ from $v'$, inserting a new vertex $v''$ and adding the neighbors in $U$ as neighbors of $v''$.
\begin{restatable}{lemma}{splitefficient}[See Lemma $7.9$ in \cite{chen2024almost} and \Cref{appendix:splitefficient}]\label{lma:splitefficient}
    Let $G = (V, E)$ be a dynamic graph on initially $n$ vertices and $m$ edges, undergoing $m$ edge updates and $n$ vertex splits over time. Then, this sequence of vertex splits can be simulated by $O(m \log n)$ edge insertions and deletions.
\end{restatable}
\section{Maintaining Cut-Preserving j-Trees} \label{sec:madry}
The goal of this section will be to introduce $j$-trees, motivate their utility, and then show how to construct and maintain them. As initiated in the work by Chen, Goranci, Henzinger, Peng, and Saranurak~\cite{Chen:2020aa}, a $j$-tree will serve as a ``vertex sparsifier''. That is, a $j$-tree allows us to reduce the problem of solving certain cut problems on an initial graph G on $n$ vertices to a smaller graph on only $j$ vertices.

\subsection{Basic Definitions}
\Cref{lma:cutpreserve} tells us that if $G \preceq H \preceq_{\alpha} G$, then also for any cut $S$ we have $U_G(S) \leq U_H(S) \leq \alpha \cdot U_G(S)$. So if we could find a structurally simple graph $H$ on the same vertex set as $G$, such that $G \preceq H \preceq_{\alpha} G$, we can use $H$ as a proxy for $G$ in cut-based computations, while losing only an approximation factor of at most $\alpha$. In~\cite{Madry:2010ab}, it is shown how to find a collection of $i = 1, \ldots, k$ structurally simpler graphs $H_i$ such that their average satisfies $G \preceq k^{-1} \sum_i H_i \preceq_{\alpha} G$. We will heavily build on his approach. 

Given a graph $G$ and a rooted forest $F \subseteq G$\footnote{We emphasize that here $F$ is generally not a spannning forest of $G$.} with root set $R$, we denote by $\cross_F^G$ the set of edges $e = (u, v) \in E(G)$ where the endpoints are not in the same tree, i.e., where $\root_F(u) \neq \root_F(v)$. Given a crossing-edge $e = (u, v) \in \cross_F^G$, we define the projected edge $\proj_F^G(e) := (\root_F(u), \root_F(v))$. Moreover, we call the edges in $E_G \setminus \cross_F^G$ \emph{internal edges}. These are the edges where both endpoints are in the same connected component of $F$. Finally, given two vertices $u, v$ in the same connected component of $F$, we denote by $F[u, v]$ the unique $u$ to $v$ path in $F$.

We can now introduce the central definition of this section, first introduced in~\cite{Madry:2010ab}.
\begin{restatable}[\(j\)-tree]{definition}{jtree} \label{def:jtree}
    A graph $H = (V, E_H, u_H)$ is called a $j$-tree if it is the union of a rooted forest $F \subseteq H$ with root set $R$ satisfying $|R| \leq j$ and a core $C(H) = H[R]$ which is the induced subgraph on the set of roots $R$. 
\end{restatable}
We will often write $H = C(H) \cup F(H)$ to denote the core and forest of $H$, respectively.

Given a graph $G$ and a rooted forest $F$ with root set $R$, we now define the contracted graph $G/F$, which is a graph on vertex set $R$ instead of $V$, in which we contract each forest component into a single vertex.
\begin{definition}[Contracted Graph] Given a graph $G = (V, E_G, u_G)$ and a rooted forest $F \subseteq G$ with root set $R$, we define the contracted graph $G / F$ to be a multi-graph on vertex set $V(G / F) = R$ and edge set $E(G/F) = \{\proj_F^G(e): e \in \cross_F^G\}$ with capacities $u_{G / F}(\proj_F^G(e)) = u_G(e)$.
\end{definition}
\paragraph{Induced $j$-tree.} Let us now discuss how we will construct $j$-trees. Given a graph $G$ and a forest $F \subseteq G$ with capacities $u_F$ and root set $R$, there is a natural way to construct the $|R|$-tree $H = H(G, F)$ \emph{induced by $F$ and $R$} as follows. We use the root set $R$ of $F$ as the vertex set of the core. Then, for each edge $e \in \cross_F^G$, we add a multi-edge $\proj_F^G(e)$ to the core with capacity $u_G(e)$. Additionally, for each forest edge $e \in E_F$, we keep the forest capacities $u_F(e)$ in $H$. Written concisely, we can write $H(G, F) = G / F \cup F$.

With an induced $j$-tree $H$, we can always associate a \emph{canonical embedding} $\f$ of $G$ into $H$. Each internal edge $e = (u, v) \not \in \cross_F^G$ we embed into $H$ by letting $f_e$ send $u_G(e)$ units of flow along the tree path $T[u, v]$, where $T$ is the connected component of $F$ containing both $u$ and $v$. Each crossing-edge $e = (u, v) \in \cross_F^G$ we embed by letting $f_e$ be the flow sending $u_G(e)$ units of flow along the path $F[u, \root_F(u)] \oplus \proj_F^G(e) \oplus F[\root_F(v), v]$, where $\oplus$ denotes the concatenation of paths. 

Remember that by \Cref{def:embeddability} such an embedding is only valid if $|\f(e)| \leq u_H(e)$ for every edge $e \in E_H$. In order for this to hold for the above embedding, it is immediate that we only need to ensure that $|\f(e)| \leq u_H(e)$ for \emph{forest edges} $e \in E_F$, i.e., we need to choose the forest capacities $u_F(e)$ such that $|\f(e)| \leq u_F(e)$. 

In order to verify embeddability in the other direction, i.e., embeddability of an induced $j$-tree $H$ into $G$ it suffices to consider the capacities on the forest edges. This can be seen in the following claim.

\begin{claim}\label{clm:forestcong}
    Let $G = (V, E_G, u_G)$ be a graph and $F = (V, E_F, u_F) \subseteq G$ be a rooted forest with root set $R$. Let $H = H(G, F)$ be the induced $|R|$-tree. Suppose that $G \preceq H$ by the canonical embedding, and let $\alpha := \max_{e \in E_F} u_F(e) / u_G(e)$. Then $H \preceq_{2\alpha} G$. Moreover, the embedding routes only $u_G(e)$ flow on edges $e \not\in F$ and at most $2u_{F}(e)$ flow otherwise.
\end{claim}
\begin{proof}
    We specify an embedding $\setof{f_e}_{e \in E_H}$ from $H$ into $G$. First, for every $e \in E_F$ we also have $e \in E_G$, so we simply let $f_e$ be the flow routing $u_F(e)$ units of flow through the edge $e$ in $G$. Second, for each (multi-)edge $e = (u, v) \in C(H)$ with capacity $u_H(e)$, there exists a unique edge $e' = (u', v') \in E_G$ with capacity $u_G(e') = u_H(e)$ such that $e = \proj_F^G(e')$ and $u = \root_F(u'), v = \root_F(v')$. Thus, we let $f_e$ send $u_G(e')$ units of flow along the path $F[u, u'] \oplus e' \oplus F[v, v']$ in $G$.
    
    This concludes the specification of the embedding. It is immediate that any edge $e \in E_G \setminus E_F$ gets sent at most $u_G(e)$ units of flow through it. It thus suffices to consider the congestion of forest edges $e \in E_F$. But notice that by design, the total amount of flow that gets sent through any edge $e \in E_F$ is given by $u_F(e) + |\f(e)| \leq 2 u_F(e) \leq 2 \alpha u_G(e)$, where $\f$ is the canonical embedding of $G$ into $H$, and we used that by assumption $|\f(e)| \leq u_F(e)$. 
\end{proof}
\paragraph{Dynamic Induced $|R|$-trees \& Vertex Splits}\label{pg:dyntreesplits} Given a dynamic graph $G$ and a decremental rooted forest $F$ with root set $R$, we will want to maintain the induced $|R|$-tree $H(G, F) = G / F \cup F$. As the root set $R$ (and hence the number of components) of the forest grows over time, the vertex and edge sets of the contracted graph $G / F$ changes significantly. Suppose that forest $F'$ differs from forest $F$ by the removal of an edge $e$ from the connected component $T$ of $F$ with root $r$ and the root set grows as $R' = R \cup \{u\}$. Then, the connected component $T$ transforms into two connected components $T_r, T_u$ in $F'$, where $T_u$ has root $u$ and $T_r$ has root $r$. Consequently, the updated contracted graph $G / F'$ now has vertex set $V(G / F') = R' = R \cup \{u\}$. Its edge set changes as follows. On one hand, all edges $e = (x, y)$ with $x \in T_u$ and $y \in T_r$ now are cross edges, so that for each such edge there is a corresponding projected edge $\proj_F^G(e)$ inserted into $G / F'$. 

On the other hand, an edge $e = (x, y)$ with $x \in T_u$ and $y \not \in T_u \cup T_r$ previously led to a projected edge $\proj_F^G(e) = (r, \root_F(y))$. However, in the new forest $F'$, we have $\root_{F'}(x) = u$, so that the projected edge is $\proj_{F'}^G(e) = (u, \root_F(y))$. We thus need to update the edge set by moving the endpoint of all such projected edges from $r$ to $u$. This is precisely captured by the definition of a \emph{vertex split} in \Cref{def:vsplit}, and will be crucial in the subsequent sections. To summarize, when a rooted forest incurs an edge deletion, the corresponding $|R|$-tree can be updated through one edge deletion, one vertex split, and potentially multiple edge insertions.

\subsection{Maintaining $j$-trees}
Ideally, we would like to find a rooted forest $F$ with capacities $u_F$ and root set $R$, such that $|R| \approx j$, where $j \ll n$ and $G \preceq H(G, F) \preceq_{\alpha } G$ for some small $\alpha$, say $\alpha = \Otil(\log n)$\footnote{In this paper, we write $g(n) = \Otil(f(n))$ if there exists some constant $C >0$ with $g(n) = O(f(n) \log^C f(n))$, and write $g(n) = \polylog(f(n))$ if there exists a constant $C > 0$ such that $g(n) = O(\log^C f(n))$}. By \Cref{clm:forestcong}, in order to achieve this we would need the capacity of each forest edge $e \in E_F$ to be within a factor of at most $\Otil(\log n)$. On the other hand, if we embed $G$ into $H$, many edges of $G$ might embed through the forest edges, which forces us to choose large forest capacities $u_F$. Inherently, there is a tension between embedding $G$ into $H$ and embedding $H$ into $G$. Thus, satisfying that $G \preceq H \preceq_{\alpha} G$ for a single $j$-tree $H$ might be too ambitious a goal.

To resolve this situation, Madry~\cite{Madry:2010ab} showed that in static graphs, if, instead of trying to find a single $j$-tree, we take an average of multiple ones, we can achieve a much better embedding quality. That is, he shows how to find a sequence of rooted forests $F_i$ with relatively small root sets $R_i$ such that the induced $|R_i|$-trees on average embed into $G$ with small $\alpha$. Similarly, we will prove the following theorem.  In \Cref{appendix:madry-almost-tight} we provide a well-known argument that shows -- at least when considering \emph{induced} $j$-trees -- that the number of $j$-trees required to guarantee $\alpha$-embeddability is almost tight in the theorem below.

\begin{restatable}{theorem}{MadryJtreeCollection}\label{thm:jtree}
    Given a parameter $j$ and a dynamic graph $G = (V, E_G, u_G)$ with initially $m$ edges and capacity ratio $U$ undergoing up to $O(j)$ edge insertions and deletions. Then, there exists an algorithm that initializes and maintains a set of $i = 1, 2, \ldots, k = \Otil(m/j)$ many decremental rooted forests $F_i = (V, E_{F_i}, u_{F_i})$ with root sets $R_i$, as well as the corresponding induced $|R_i|$-trees $H_i = H(G, F_i)$. 
    
    Initially, we have that $|R_i| \leq O(j)$ for all $i$. Moreover, at all times, it holds that:
    \begin{enumerate}
        \item $G \preceq H_i$ for all $i = 1, 2, \ldots, k$,
        \item $\frac{1}{k} \sum_{i=1}^k H_i \preceq_{\alpha} G$ for some $\alpha = \Otil(\log n)$, 
        \item The capacity ratio of each $H_i$ is $O(m U)$,
        \item After $t$ edge updates have been performed on $G$, we have $|R_i^{(t)}| \leq |R_i^{(0)}| + O(t)$ for all $i$.
    \end{enumerate}
    The algorithm runs in total time $\Otil(m^2 / j)$. Moreover, the total changes of each contracted graph $G / F_i$ (and thus of the core $C(H_i)$) consist of $O(j)$ vertex splits, $O(j)$ edge deletions and $O(m)$ edge insertions.
\end{restatable}
Note that the above theorem tells us that after initializing a set of (induced) $O(j)$-trees, we can maintain it over a sequence of $O(j)$ updates by slightly growing the root sets $R$ -- and thus the cores $C(H_i)$. Then, as discussed in the previous section, growing the root sets $R_i$ corresponds to edge insertions and, in particular, vertex splits in the contracted graphs $G / F_i$. We provid a complete proof of the theorem in the next section.

For now note that while the graphs $C(H_i) = G / F_i$ are only on $|R_i| = O(j)$ vertices, they might contain $\Omega(m)$ (multi-)edges. Thus, there is no significant size reduction at this point in the number of edges. To achieve this, we additionally need to run an edge sparsification procedure on the core. We define a sparsified core as follows.
\begin{restatable}[Sparsified Core]{definition}{SparsifiedCore} \label{def:sparsified-core}
Given a $j$-tree $H$ with core $C(H)$ and capacity ratio $U$, we call a graph $\Tilde{C}(H) \subseteq C(H)$ a sparsified core of quality $\beta$ if
\begin{enumerate}
    \item $\Tilde{C}(H)$ is a graph with at most $\Otil(j \log U)$ many edges,
    \item The capacity ratio of $\Tilde{C}(H)$ is $\poly{n} U$, and
    \item for all cuts $(S, V \setminus S)$ in $C(H)$, we have
        \[
        U_{C(H)}(S) \leq U_{\Tilde{C}(H)}(S) \leq \beta \cdot U_{C(H)}(S).
        \]
\end{enumerate}
\end{restatable}
Often, we will also call $\Tilde{C}(H)$ a $\beta$-approximate sparsified core of $H$.
\begin{claim}\label{clm:sparsifiedjtree}
    Let $H = C \cup F$ be a $j$-tree with core $C = C(H)$ and forest $F = F(H)$, and let $\Tilde{C}$ be a sparsified core of quality $\beta$. Then $\Tilde{H} := \Tilde{C} \cup F$ is a $j$-tree that satisfies, for every cut $(S, V \setminus S)$,
    \[
    U_H(S) \leq U_{\Tilde{H}}(S) \leq \beta \cdot U_H(S)
    \]
\end{claim}
\begin{proof}
    That $\Tilde{H}$ is a $j$-tree is immediate. Consider now a cut $S$ and first note that since forest capacities are the same in $H$ as in $\Tilde{H}$, the total capacity of cut forest edges is the same in $H$ as in $\Tilde{H}$. We thus only need to consider core edges, but by the guarantees of $\Tilde{C}$ being a sparsified core of quality $\beta$, we also have that $U_C(S \cap V(C)) \leq U_{\Tilde{C}}(S \cap V(C)) \leq \beta \cdot U_C(S \cap V(C))$.
\end{proof}
By running the cut sparsifier from \Cref{thm:fullydynamicsparsifier}, we can maintain the sparsified cores under low recourse. Remember that the graphs $G / F_i$ grow in size over time by vertex splits. Moreover, while the data structure processes only $O(j)$ updates on the initial graph $G$, the core $C_i(H)$ can change by up to $\Omega(m)$ additional edge insertions. Hence it is crucial that (a) the sparsifier is able to handle vertex splits with low recourse and (b) has recourse independent of the number of edge insertions.
\begin{restatable}{corollary}{sparsifiedcoremaintain}\label{thm:sparsifiedcore}
    Let $H = H(G, F)$ be an $O(j)$-tree maintained by \Cref{thm:jtree}, and $C = C(H)$ be its core. Then, by running \Cref{thm:fullydynamicsparsifier}, we can initialize and maintain a sparsified core $\Tilde{C}(H)$ of quality 
    $2$ over the full sequence of $O(j)$ updates to $G$ in total time $\Otil(m \log U)$ and total recourse (i.e., edge insertions/deletions in $\Tilde{C}(H)$) of $\Otil(j \log U)$. The algorithm is correct with probability $1-1/n^c$ for any constant $c > 0$ specified before the algorithm starts. 
\end{restatable}

\subsection{Proof of \Cref{thm:jtree}}
Madry showed in~\cite{Madry:2010ab} how to initialize $j$-trees satisfying the properties of \Cref{thm:jtree} from scratch. Using his initialization as a complete black-box, however, will be inconvenient for the analysis of the dynamic algorithm. For this reason, we will only use his theorem on initializing a certain set of spanning trees $T_i$ of $G$ with a small sub-set of special high-congestion edges $S_i \subseteq E_{T_i}$. Using these trees one can construct certain corresponding $j$-trees straightforwardly, and their maintenance will also be simple. We first describe an easy way to choose capacities $u^T$ for a spanning tree $T$ of $G$ such that $G \preceq T$. 

Given a graph $G = (V, E_G, u_G)$, a spanning tree $T$ of $G$ and an edge $e = (u, v) \in E_G$, there is a unique way to route $u_G(e)$ units of flow between $u$ and $v$ in $T$. We simply route flow along the unique tree path $T[u, v]$. This constitutes a canonical embedding from $G$ into $T$. Doing this for every edge of $G$ induces a multi-commodity flow in $T$, denoted by $\f_T$. Then we define the induced tree capacities $u^T$\footnote{We follow the notation $u^T$ from \cite{Madry:2010ab} to emphasize that the capacities are induced by the tree $T$} for an edge $e \in E_T$ as $u^T(e) := |\f_T(e)| = \sum_{e' = (u,v) \in E_G: e \in T[u, v]} u_G(e')$. Note that by construction, $G \preceq T$, and that the capacity ratio of $T$ is $O(m U)$, where $U$ is the capacity ratio of $G$. We next give a measure of how well $T$ can be embedded back into $G$.

\begin{definition}
Given a spanning tree $T$ of $G$, let us define the congestion of an edge $e \in E_G$ as follows:
\[
\fcong_G^T(e) = \begin{cases}
    u^T(e)/ u_G(e) & \tn{if } e \in E_T \\
    1 & \tn{otherwise.}
\end{cases}
\]
\end{definition}

We can now state the main building block of the construction by Madry. We provide a self-contained proof in \Cref{appendix:madryproof}. 
\begin{restatable}{theorem}{madry}\label{thm:madry}
    Let $G = (V, E_G, u_G)$ be a graph and $j$ be a parameter. Then, in time $\Otil(m^2/j)$, one can construct a sequence of $i = 1, 2, \ldots, k = \Otil(m/j)$ spanning trees $T_1, \ldots, T_k$ of $G$, together with edge-sets $S_i \subseteq E_{T_i}$ satisfying $|S_i| \leq j$, such that for every edge $e \in E_G$
    \[
    \frac{1}{k}\sum_{i: e \not \in S_i} \fcong_G^{T_i}(e) \leq \Otil(\log n).
    \]
\end{restatable}
Of course, by construction we have that $G \preceq T_i$ for all $i$. The above theorem essentially shows that if we are allowed to remove a small set of edges $S_i$ from each tree $T_i$, then the average of the trees also embeds back well into $G$ by the identity embedding. Similar to the construction in~\cite{Madry:2010ab}, this set of edges $S_i$ contains edges $e \in T_i$ whose congestion $\fcong_G^{T_i}(e)$ is large, although we construct it in a somewhat different and arguably simpler way. 

Naturally, once we remove the edges $S_i$, the graph $G$ won't embed into the resulting forests $F_i$; the insight is that with some additional modifications one can show that it will however embed back into the $O(j)$-trees induced by the forests $F_i$. This is what we will do next, i.e., we now show how using these trees $T_i$ and sets $S_i$, we can quickly construct a set of $O(j)$-trees satisfying the properties of \Cref{thm:jtree}.

We start by introducing the definition of branch-free sets, as used in~\cite{Chen:2023aa}. They are useful for making the algorithm dynamic.
\begin{definition}
    Given a rooted tree $T$ and a set of vertices $R \subseteq V(T)$, we call the set $R$ branch-free if for any $u, v \in R$, the lowest common ancestor of $u$ and $v$ in $T$ is also in $R$.
\end{definition}
\begin{remark}
    For any set $R \subseteq V(T)$ that is not branch-free, one can in time $O(n)$ create a new set $R' \supseteq R$ that is branch-free and satisfies $|R'| \leq 2|R|-1$.
\end{remark}
We can now outline the construction of the $O(j)$-trees given the trees $T_i$ from \Cref{thm:madry}. First, we will describe the construction without arguing for efficiency, which we will do immediately afterwards. We start by running \Cref{thm:madry} with parameter $j$, which gives us a sequence of $\Otil(m/j)$ many trees $T_i$ and associated edge sets $S_i$ satisfying $|S_i| \leq j$. We then proceed as follows.

\underline{Specifying the root sets:} We arbitrarily root each tree $T_i$ and construct the forests $F_i$ and their root sets $R_i$ as follows. For each edge $e \in S_i$, we add both endpoints $u, v$ to $R_i$. By the guarantees of \Cref{thm:madry}, we have at this stage that $|R_i| = O(j)$. Then, we also add up to $O(j)$ additional vertices to $R_i$ in order to make $R_i$ a branch-free set with respect to $T_i$. This specifies the initial set of roots $R_i$. 

\underline{Specifying the forests:} From $T_i$ and $R_i$, we construct the forests $F_i$ as follows. Call two vertices $u, v \in R_i$ adjacent if the tree-path $T_i[u, v]$ of $u$ and $v$ does not contain any other vertex in $R_i$. Let $\cP_i$ be the set of paths $T_i[u, v]$ for adjacent vertices $u, v \in R_i$. Note that $|\cP_i| = O(|R_i|) = O(j)$\footnote{This can be seen by noting that on each leaf to root path in $T_i$ there are as many paths in $\cP_i$ as there are vertices $r \in R_i$ on the given leaf to root path.}. Now, given a path $p \in \cP_i$, let $e_{\min}^{p, T_i}$ be the edge $e \in p$ with minimum capacity $u^{T_i}(e)$ on the path $p$ (ties are broken arbitrarily but consistently.) Then $F_i$ is received from $T_i$ by removing all edges $e_{\min}^{p, T_i}$, i.e., we set $E_{F_i} = E_{T_i} \setminus \{e_{\min}^{p, T_i}: p \in \cP_i\}$. Note that because for any edge $e = (u, v) \in S_i$, we have that both $u, v \in R_i$, we in particular have that $E_{F_i} \subseteq E_{T_i} \setminus S_i$. Finally, we set the forest capacities as $u_{F_i}(e) = 2 u^{T_i}(e)$, and set $H_i = H(G, F_i) = G / F_i \cup F_i$. We stress that the forest capacities will stay fixed throughout the algorithm. This concludes the initialization of the graphs $H_i$.

\underline{Maintenance:} We now describe how to maintain the forests $F_i$ and root sets $R_i$ under edge updates to $G$. Suppose that an edge $e = (u, v)$ either gets inserted into $G$ or deleted from it. Then we update the root sets $R_i$ by adding both endpoints $u, v$ to $R_i$, and possibly an additional $2$ extra vertices in order to keep $R_i$ as a branch-free set with respect to $T_i$. This will introduce $O(1)$ new paths $p \in \cP_i$. For each one of these new paths $p$, we again find the edges $e_{\min}^{p, T_i}$ and remove them from $F_i$. Then, we correspondingly update the induced $|R_i|$-trees $H_i = H(G, F_i)$. Remember that $H_i = G / F_i \cup F_i$, so in order to update $H_i$, we need to update the contracted graphs $G / F_i$. Note that once this is done, both endpoints $u, v$ of the edge are inside the core $C_i = C(H_i)$. We can thus simply insert/remove the edge $(u, v)$ from the updated core $C(H_i)$. We will later show that trivially handling such edge updates inside the core does does not decrease the embedding quality.

\underline{Efficiency:} While we have now described the ideas on how to maintain the forests $F$ and graphs $H_i = H(G, F_i)$, we still need to discuss how to implement everything efficiently. First, in order to find the new edges $e_{\min}^{p, T_i}$, we will use the Link-Cut-Tree data structure $\cD_1$ from \Cref{thm:treeds}. It will allow us to identify each such edge in $O(\log n)$ amortized time\footnote{Consequently, the set $\cP_i$ is not explicitly maintained by our algorithm. We only need to maintain the correct edge set $E_{F_i} = E_{T_i} \setminus \{e_{\min}^{p, T_i}: p \in \cP_i\}$.}. 

Next, we need to show that we can maintain the graphs $H_i = H(G, F_i)$ efficiently, which boils down to maintaining the contracted graphs $G / F_i$ efficiently. We implement this in the $\textsc{AddTerminal}(w)$ operation in \Cref{algo:maintain} as follows. As discussed at the end of \Cref{pg:dyntreesplits}, if forest $F'$ differs from forest $F$ by the removal of an edge $e$ and the root set grows as $R' = R \cup \{w\}$, then the connected component $T$ of $F_i$ with root $r$ now consists of two connected components $T_r, T_w$ in $F_i'$, where $T_w$ has root $w$ and $T_r$ has root $r$.
So the graph $G / F_i'$ changes (a) by the insertion of edges $\proj_F^G(e)$ for $e = (x, y)$ with $x \in T_u$ and $T_r$ and (b) by moving the endpoint of edges $(r, y)$ to $(w, y)$ if the edge $(r,y)$ corresponded to a projected edge $\proj_F^G((x, y))$, where $\root_F(x) = r$, but now $\root_{F'}^G(x) = w$. Moving edges in this way correspond precisely to vertex splits. See also \Cref{fig:update} for a graphical depiction on how the forest and core can change.

In order to efficiently perform these vertex splits, we want to simulate them correspondingly and then invoke \Cref{lma:splitefficient}. To do so, we maintain for each connected component $T$ of $F_i$ and vertex $v \in T$, the number of crossing-edges $e \in \cross_F^G$ adjacent to the sub-tree $v^{\downarrow}$ of $T$ rooted at $v$, i.e., the crossing-edges $e$ such that $e \cap v^{\downarrow} \neq \emptyset$. We can do this using the data structure $\cD_2$ from \Cref{thm:ett} in $O(\log n)$ worst-case time. Then, if an edge-deletion splits a tree $T$ into trees $T_r$ and $T_w$, by calling $\cD_2.\textsc{Sum}(w)$ and $\cD_2.\textsc{Sum}(r)$, we can find $\min\{|E_G(T_r, V \setminus (T_r \cup T_w))|, |E_G(T_w, V \setminus (T_r \cup T_w))|\}$ and only need to move the edges of the smaller half, which is important to invoke \Cref{lma:splitefficient} and from which it follows that the moving of edges required to maintain the $O(j)$-trees induced by $F_i$ can be performed in $O(m \log n)$ total time. For pseudo-code, see \Cref{algo:maintain}. 

Before we start with the formal proof of \Cref{thm:jtree}, we show two crucial lemmas. The first one shows that after initialization, $G \preceq H_i$, and that growing the root sets $R_i$ (and thus the cores $G / F_i$) preserves embeddability. 

\begin{lemma}\label{lma:struct1}
    Let $G$ be a graph and $T_i$, $S_i$ be a tree and corresponding edge set from \Cref{thm:madry}. Let $R_i$ be a branch-free root-set satisfying $R_i \supseteq \{v: $v$ \tn{ is endpoint of an edge } e \in S_i \}$, and let $F_i$ be the corresponding forest with edge set $E_{F_i}  = E_{T_i} \setminus \{e_{\min}^{p, T_i}: p \in \cP_i\}$. Let $H_i = H(G, F_i)$ be the induced $|R_i|$-tree with forest capacities set as $u_{F_i}(e) = 2 u^{T_i}(e)$. Then $G \preceq H_i$ by the canonical embedding.
\end{lemma}

\begin{remark}
    Note that in our dynamic algorithm, the root sets $R_i$ and forests $F_i$ satisfy the conditions of the theorem at all times. If we let $G^{(0)}$ denote the initial graph, and $F_i^{(t)}$ denote forest $F_i$ after our data structure has processed $t$ updates, then Lemma~\ref{lma:struct1} implies that  $G^{(0)} \preceq H(G^{(0)}, F_i^{(t)})$. This is a crucial property that makes the construction amenable to the dynamic setting, as from it we will be able to infer that also $G^{(t)} \preceq H(G^{(t)}, F_i^{(t)})$. 
\end{remark}
\begin{proof}
    First remember that the canonical embedding is defined as follows. For each connected component $T \in F_i$ and internal edge $e = (u, v)$ with $u, v \in T$, we let $f_e$ be the flow sending $u_G(e)$ units of flow along the tree path $T[u, v] = T_i[u, v]$. For all the other edges $e = (u, v) \in \cross_{F_i}^G$, we let $f_e$ be the flow sending $u_G(e)$ units of flow along the path $F_i[u, \root_{F_i}(u)] \oplus \proj_{F_i}^G(e) \oplus F_i[\root_{F_i}(v), v]$.

     Now notice that any edge $e = (x, y) \in E_G$ for which $x, y$ are in the same connected component $T$ of $F_i$ still gets routed through the tree path $T_i[x, y]$. Hence, the only edges whose embedding path in $H_i$ differs from the tree path are the edges $e \in \cross_{F_i}^G$, i.e., the edges $e = (x, y)$ for which there exists a path $p \in \cP_i$ with $e_{\min}^{p, T_i} \in T_i[x, y]$. 

     Now let $p \in \cP_i$ be a path that intersects the component $T$ of $F_i$, i.e., $p \cap T \neq \emptyset$. We know by definition that $\sum_{e = (x,y): e \in T_i[x, y]} u_G(e) = u^{T_i}(e_{\min}^{p, T_i})$, i.e., that the total amount of flow that routes through $e_{\min}^{p, T_i}$ in the tree $T_i$ is precisely the capacity of the edge $e_{\min}^{p, T_i}$. Now note that for any edge $e' \in p \cap T$, the only potential edges $e = (x,y)$ for which $e' \in f_e$ but $e' \not \in T_i[x, y]$ are precisely the edges for which $e_{\min}^{p, T_i} \in T_i[x, y]$. But this means that the embedding $\f$ sends at most $u^{T_i}(e') + u^{T_i}(e_{\min}^{p, T_i}) \leq 2 u^{T_i}(e')$ units of flow through $e'$. In the last inequality, we used that by definition of $e_{\min}^{p, T_i}$ we have $u^{T_i}(e_{\min}^{p, T_i}) \leq u^{T_i}(e')$.
     


\end{proof}
The second observation is that the insertion/deletion of edges $(u, v)$ where both endpoints $u, v$ are already in the core of the $j$-tree does not invalidate the canonical embedding. 
\begin{lemma}\label{lma:struct2}
    Let $G$ be a graph, $F \subseteq G$ be a rooted forest with root set $R$, and let $H = H(G, F)$ be the $|R|$-tree induced by $F$ and $R$. Suppose further that $G \preceq H$ by the canonical embedding. Then if $G'$ is a graph differing from $G$ by the insertion/deletion of edge $e = (u,v)$ with $u, v \in R$, the induced $|R|$-tree $H' = H(F, G')$ of $F$ and $G'$ still satisfies $G' \preceq H'$ by the canonical embedding.
\end{lemma}
\begin{proof}
    If $\setof{f_{e}}_{e \in E_G}$ is the canonical embedding of $G$ into $H$, and $e = (u, v)$ is an edge with $u, v \in R$, then only the flow $f_{e}$ routes through this edge, and no other flow even uses the edge. Hence after removing the edge from both $G$ and $H$, the remaining flows still serve as a valid embedding. If an edge $e = (u,v)$ with $u, v \in R$ gets inserted into $G$ and $C(H)$, we can extend the embedding by adding the flow $f_e$ which maps $e \in G$ to the same edge $e \in C(H)$ with capacity $u_G(e) = u_H(e)$. 
\end{proof}

We can now prove the main theorem, with pseudo-code in \Cref{algo:maintain}.
\begin{proof}[Proof of \Cref{thm:jtree}]
\underline{Correctness:} As the capacity ratio of a tree $T_i$ is at most $mU$, we in particular have that the capacity ratio of the decremental forest $F_i \subseteq T_i$ with capacities $2 u^{T_i}(e)$ is at most $2mU$. Consequently, the capacity ratio of $H_i = G / F_i \cup F_i$ is at most $2mU$. 

We have also already seen that the root sets $R_i$ initially satisfy $|R_i| \leq O(j)$. After each update, the root set grows by at most $O(1)$, also as promised. 

We next show that at all times $G \preceq H_i$ by the canonical embedding for all $i$. Suppose that $G$ has received $t$ edge updates. Our goal is to show that $G^{(t)} \preceq H_i^{(t)} = H(G^{(t)}, F_i^{(t)})$. Note first that by \Cref{lma:struct1}, it holds that $G^{(0)} \preceq H(G^{(0)}, F_i^{(t)})$. By construction of the root set $R_i^{(t)}$, we know that if $G^{(t)}$ is obtained from $G^{(0)}$ by the insertion/deletion of edges $e_1 = (u_1, v_1), \ldots, e_t = (u_t, v_t)$, it holds that $u_l, v_l \in R_i^{(t)}$ for all $l = 1, \ldots, t$. But then repeated application of \Cref{lma:struct2} shows that $G^{(t)} \preceq H(G^{(t)}, F_i^{(t)})$, as desired. 

It remains to show that for some $\alpha = \Otil(\log n)$, we have $k^{-1} \sum_{i=1}^k H_i \preceq_{\alpha} G$ at all times. Consider the embeddings $\f_i$ from $H_i$ into $G_i$ specified in \Cref{clm:forestcong}. This claim shows that for each edge $e \in E_G$, we have that $|\f_i(e)| = u_G(e)$ if $e \not \in F_i$, and $|\f_i(e)| \leq 2 u_{F_i}(e)$  otherwise. We let $k^{-1} \sum_{i=1}^k \f_i$ be the embedding from $k^{-1} \sum_{i=1}^k H_i$ into $G$ and observe that, again by \Cref{clm:forestcong}, for every $e \in E_G$
\begin{align*}
     \frac{|\f(e)|}{u_G(e)} &= \frac{1}{k} \sum_{i=1}^k \frac{|\f_i(e)|}{u_G(e)} \leq 1 +  \frac{1}{k} \sum_{F_i \ni e} \frac{|\f_i(e)|}{u_G(e)} \leq 1 +  \frac{1}{k} \sum_{F_i \ni e} 2 \frac{u_{F_i}(e)}{u_G(e)} = 1 +  \frac{1}{k} \sum_{F_i \ni e} 4 \frac{u^{T_i}(e)}{u_G(e)} \\
    & = 1 +  \frac{1}{k} \sum_{F_i \ni e} 4 \fcong^{T_i}_G(e) \leq 1 + \frac{1}{k} \sum_{i: e \not \in S_i} 4 \fcong^{T_i}_G(e) \leq \Otil(\log n),
\end{align*}
where the penultimate inequality is due to the fact that $\{F_i: e \in F_i\} \subseteq \{i: e \not \in S_i\}$, and the last inequality is due to the initialization guarantees of \Cref{thm:madry}. 

\underline{Running time \& Recourse:} Initialization is done in time $\Otil(m^2 / j)$ by \Cref{thm:madry}. It remains to show that updating all graphs $H_i$ and forests $F_i$ over a sequence of $O(j)$ updates can be done in total time $\Otil(m^2/j)$. We show that maintenance of a single graph $H_i$ and its forest $F_i$ can be done in total time $\Otil(m)$, from which the statement follows as there are $\Otil(m/j)$ many such graphs. 

The function $\textsc{AddTerminal}(w)$ first updates the forest $F_i$ to contain vertex $w$ in its root set by finding and removing the minimum capacity edge $e_{\min}^{p, T_i}$ on the $w$ to $\root_{F_i}(w)$ path $p$ in $F_i$. Then, it updates the contracted graph $G / F_i$ accordingly, so that the induced $|R_i|$-tree $H(G, F_i) = G / F_i \cup F_i$ remains valid. We show that for a given forest $F_i$, the total running time required to perform all $\textsc{AddTerminal}(\cdot)$ calls is $\Otil(m)$. We also show that the contracted graphs $G / F_i$ (or equivalently the cores $C_i$), over the full sequence of $O(j)$ updates made to $G$, change by $O(j)$ vertex splits, $O(j)$ edge deletions and $O(m)$ edge insertions. 

We first observe that the total amount of $\textsc{AddTerminal}(\cdot)$ calls is bounded by $O(j)$. In each such call, we find the edge $e_{\min}^{p, T_i}$ in $O(\log n)$ time using the data structure $\cD_1$ from \Cref{thm:treeds}, yielding a total time of $O(j \log n)$ for updating the forest $F_i$. We also note that throughout the algorithm, we maintain the vector $z(v) = |\{e \in \cross_{F_i}^G: v \in e\}|$ using the data structure $\cD_2$ from \Cref{thm:ett}. We show below that this takes $\Otil(m)$ total time. For now notice that as a consequence, at all times we have for a component $T_r \in F_i$ with root $r$, that $\cD_2.\textsc{Sum}(r) = |E_G(T_r, V \setminus T_r)|$.

Next, consider the second for-loop in Line~\ref{algo:maintain:insert} of \Cref{algo:maintain}. An edge $e = (x, y)$ is only considered if in all previous times, both endpoints $x, y$ were in the same connected component of $F_i$, but after the deletion of $e_{\min}^F(u)$, we now have that $x \in T_u$ and $y \in T_r$. Note that because the forest $F_i$ is decremental, $x$ and $y$ will never be in the same connected component again, and thus never be considered in any of the subsequent loops performed in \Cref{algo:maintain:insert}. Hence, every edge gets added to the core at most once, so the total number of insertions into $H_i$ is $O(m)$ and the time spent to perform them and the calls to increment the counts in the data structure $\cD_2$ in the first for-loop is $O(m \log n)$ by \Cref{thm:ett}. 

It remains to show that the moving of non-tree edges performed in Line~\ref{algo:maintain:split} of \Cref{algo:maintain} can be done in total time $O(m \log n)$.  Now note that the original tree $T$ containing both $u$ and $r$ gets split into two sub-trees $T_r$ and $T_u$. Note that after cutting the edge $e_{\min}$ (and before inserting the crossing edges), we have $\cD_2.\textsc{Sum}(r) = |E(T_r, V \setminus (T_r \cup T_w))|$  and $\cD_2.\textsc{Sum}(w) = |E(T_w, V \setminus (T_r \cup T_w)|$. By calling $\textsc{Sum}(r)$ and $\textsc{Sum}(w)$ we can thus in time $O(\log n)$ find out if $|E(T_w, V \setminus (T_r \cup T_w))| \leq |E(T_r, V \setminus (T_r \cup T_w))|$ or not. Assume without loss of generality that $|E(T_w, V \setminus (T_r \cup T_w))| \leq |E(T_r, V \setminus (T_r \cup T_w))|$, as otherwise we can simply re-name vertices. Then the edge insertions/deletions in $G / F_i$ correspond precisely to a vertex split. Hence $G / F_i$ can be updated by performing, additionally to $O(m)$ edge insertions and $O(j)$ edge deletions, only $O(j)$ vertex splits, and, by \Cref{lma:splitefficient}, the total amount of edges we need to move by simulating all these vertex splits is bounded by $O(m \log n)$.
\end{proof}
\begin{algorithm2e}[!ht]
\caption{$\textsc{Maintain-j-Tree}()$}
\label{algo:maintain}
\SetKwProg{Globals}{global variables}{}{}
\SetKwProg{Proc}{procedure}{}{}
\Proc{$\textsc{Initialize}(G, F_i)$}{
    Construct $H_i = H(G, F_i)$. \\
    $C_i \gets$ core of $H_i$. \\
    Let $z(v) = |\{e \in \cross_{F_i}^G: v \in e\}|$ \\
    Initialize the data-structure $\cD_1$ from \Cref{thm:treeds} on $F_i$. \\
    Initialize the data-structure $\cD_2$ from \Cref{thm:ett} on $F_i$ with vector $z$.
}
\Proc{$\textsc{AddTerminal}(w)$}{
    $r \gets \cD_1.\textsc{FindRoot}(w)$ \\
    $e_{\min} \gets \cD_1.\textsc{PathMin}(w)$ \\
    Insert vertex $w$ into $C_i$ \\
    $\cD_1.\textsc{Cut}(e_{\min})$; $\cD_2.\textsc{Cut}(e_{\min})$ \\
    $\cD_1.\textsc{MakeRoot}(w)$; $\cD_2.\textsc{MakeRoot}(w)$ \\ 
    Let $T_w, T_r$ be the two components of $F_i$ with roots $w,r$ respectively. \\

    \CommentSty{ Assume w.l.o.g that $\cD_2.\textsc{Sum}(w) \leq \cD_2.\textsc{Sum}(r)$} \\
    \ForEach{$e = (x, y) \in E_G \setminus F_i$ with $x \in T_w$ and $y \in V \setminus (T_w \cup T_r)$}{\label{algo:maintain:split}
        Remove corresponding multi-edge $(r, y)$ from $C(H)$ \\
        Insert edge $(w, y)$ into $C(H)$ with capacity $u_G(e)$ \\
    }

    \ForEach{$e = (x, y) \in E_G $ with $x \in T_w$, $y \in T_r$ }{\label{algo:maintain:insert}
        Insert multi-edge $(w, r)$ into $C(H)$ with capacity $u_G(e)$ \\
        $\cD_2.\textsc{Add}(x, 1)$; $\cD_2.\textsc{Add}(y, 1)$
    }
}
\Proc{$\textsc{InsertEdge}(e = (u, v))$}{
    $G \gets G \cup \{e\}$ \\
    Let $R'$ be the updated branch-free root set containing $u$ and $v$ \\
    \ForEach{$w \in R' \setminus R$}{
        $\textsc{AddTerminal}(w)$
    }
    Insert edge $(u,v)$ into $C_i$ \\
    $\cD_2.\textsc{Add}(u, 1)$; $\cD_2.\textsc{Add}(v, 1)$
}

\Proc{$\textsc{DeleteEdge}(e = (u, v))$}{
    $G \gets G \setminus \{e\}$ \\
    Let $R'$ be the updated branch-free root set containing $u$ and $v$ \\
    \ForEach{$w \in R' \setminus R$}{
        $\textsc{AddTerminal}(w)$
    }
    Remove edge $(u,v)$ from $C_i$ \\
    $\cD_2.\textsc{Add}(u, -1)$; $\cD_2.\textsc{Add}(v, -1)$
}
\end{algorithm2e}

\begin{figure}[t]
\centering
\setkeys{Gin}{width=.8\linewidth}
    \begin{subfigure}{0.32\textwidth}
        \centering
        \includegraphics{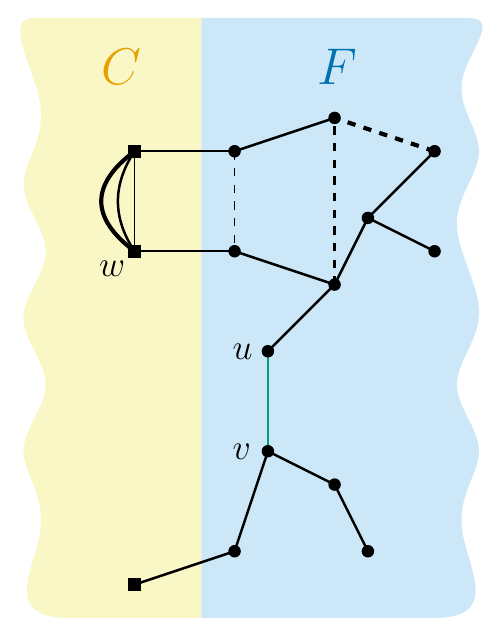}
        \caption{}
        \label{fig:update-insertion}
    \end{subfigure}
    \hfil
    \begin{subfigure}{0.32\textwidth}
        \centering
        \includegraphics{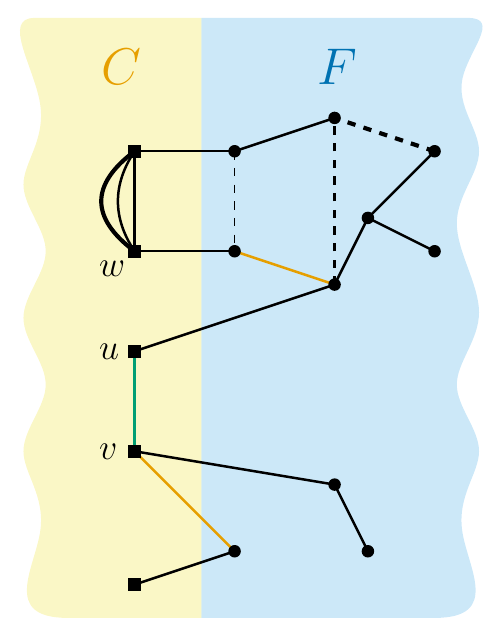}
        \caption{}
        \label{fig:update-move-to-core}
    \end{subfigure}
    \begin{subfigure}{0.32\textwidth}
        \centering
        \includegraphics{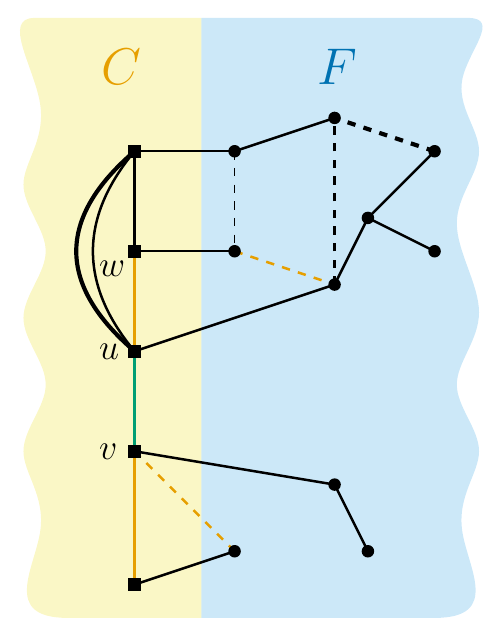}
        \caption{}
        \label{fig:update-embeddings}
    \end{subfigure}
\caption{Graphical depiction of how the core $C$ and forest $F$ might change after an update. Dotted edges denote edges that are in $G$ but not in $H$. The green edge denotes an edge newly inserted into $G$ and has to be processed by the data structure. The orange edges denote the edges $e_{\min}^F(u)$ and $e_{\min}^F(v)$. In (b), the vertices $u,v$ are added as roots to the set $R$. In (c), the minimum capacity edges have been removed from the forest and been added as projected edges into the core. Moreover, two projected edges in the core have their endpoints moved from $w$ to $u$.} 
\label{fig:update}
\end{figure}

\section{Fully Dynamic Cut Sparsifier Supporting Vertex Splits}\label{sec:cut-sparsifier}

In this section we will show how to adapt the dynamic cut-sparsifier from \cite{7782947} to also handle vertex splits. Our goal will be to show that a vertex split on the underyling graph $G$ causes only $O(\polylog(n))$ recourse in the cut-sparsifier $H$ we maintain. This is essential in order to handle the recourse propagating over the whole hierarchy of $j$-trees we will be building. Furthermore, we also notice that this sparsifier also has the property that it can handle a large number of edge-insertions with extremely low recourse. That is, if the number of edge insertions is much larger than the number of edge deletions and vertex splits, the recourse will be sub-linear in the total number of updates, i.e., edge insertions, edge deletions, and vertex splits. This will be crucial for our application. We prove the following theorem.

\begin{restatable}{theorem}{CutSparsifier}\label{thm:fullydynamicsparsifier}
    Consider a parameter $\epsilon > 0$ and a fully dynamic graph $G = (V, E, u_G)$ with $|V(G^{(0)})| = n, |E(G^{(0)})| = m \leq O(n^2)$ and capacity ratio $U$. Then, for a sequence of up to $m$ edge insertions, $n$ vertex splits and $D$ edge deletions, there is a randomized algorithm that explicitly maintains a $(1 \pm \epsilon)$-cut sparsifier $H$ of $G$, such that, at all times
    \begin{enumerate}
        \item $|E(H)| \leq O(\polylog(n) \cdot \epsilon^{-2} \cdot \log U \cdot n)$,
        \item the capacity ratio of $H$ is $O(n U)$.
    \end{enumerate}
    The algorithm takes $O(\polylog(n) \cdot \epsilon^{-4} \cdot \log U \cdot m)$  total time and the sparsifier $H$ undergoes $O(\polylog(n) \cdot \epsilon^{-2} \cdot \log U) \cdot (n+D)$ total recourse (i.e., number of edge deletions and edge insertions) for all the updates. The algorithm is correct with probability $1-1/n^c$ for any constant $c > 0$ specified before the algorithm starts.
\end{restatable}

We emphasize that if $m$ is much larger than $n$ and $D$, this in fact yields recourse sublinear in the total number of updates. In \cite{7782947}, a similar theorem was proven. Their data structure can not handle vertex splits, however, and they don't explicitly give a bound on total recourse that is independent of the number of edge insertions.
We restate the theorem here for comparison.
\begin{theorem}[Theorem $1.2$ in \cite{7782947}]\label{thm:edgedynamicCS}
    Consider $\epsilon >0$ and a dynamic graph $G = (V, E)$ on $n$ vertices undergoing edge insertions and deletions. Then, there is a randomized algorithm that explicitly maintains w.h.p a $(1 \pm \epsilon)$ cut sparsifier $H$ of $G$ of size $\Otil(n \cdot \epsilon^{-2})$ with amortized $O(\polylog(n) \cdot \epsilon^{-2})$ time and recourse  per update. 
\end{theorem}
One naive approach to adapting the cut-sparsifier of \Cref{thm:edgedynamicCS} to allow for vertex splits is by simulating the vertex split as a sequence of edge insertions and deletions. If the vertex $v$ is split into vertices $v'$ and $v''$, such that $N(v') = N(v) \setminus U$ and $N(v'') = U$, where $|U| \leq |N_G(v)|/2$. Then, we can simulate the operation by first inserting a new vertex $v''$, then removing all neighors in $U$ from $v$ and adding them as neighbors of $v''$. By \Cref{lma:splitefficient}, by simulating the vertex splits accordingly, we need to move only $O(m \log n)$ many edges in total, so feeding these updates into the data structure of \Cref{thm:edgedynamicCS} yields a total update time of $\Otil(m)$, as desired.

A priori however, a single vertex split simulated like this might lead to $\Omega(|N(v)|)$ recourse, and over a sequence of $n$ vertex splits the total recourse can be $\Omega(m \log n)$ 
which is unacceptably large for our purposes. Our main task will thus be to show that by internally modifying the algorithm, we can guarantee an amortized recourse of $O(\polylog(n))$, so that after $n$ vertex splits the total recourse crucially is only $\Otil(n)$. The main tool to achieve this will be the maintenance of low-recourse spanning forests of a graph undergoing vertex splits.

In order to prove \Cref{thm:fullydynamicsparsifier}, we first recall the central definition of an $l$-SF-bundle $B$ of a graph $G$, as defined in \cite{7782947}.
\begin{definition}[$l$-SF-bundle]
    Let $F_1$ be a spanning forest of $G$, $F_2$ be a spanning forest of $G \setminus F_1$ and, inductively, $F_k$ be a spanning forest of $G \setminus \bigcup_{i=1}^{k-1} F_i$. Then we call $B = \bigcup_{i=1}^{l} F_i$ an $l$-SF-bundle of $G$.
\end{definition}
\begin{definition}[Edge-Connectivity] Given a (multi-)graph $G = (V, E)$ and an edge $e = (u, v) \in E$, the edge connectivity $\lambda_G(e)$ of $e$ is defined as the $u$-$v$ maximum flow value in $G$.
\end{definition}
Notice that if $G$ is a graph and $B$ an $l$-SF-bundle of $G$, then any edge in $G \setminus B$ must have edge-connectivity at least $l$ as we can route one unit of flow through each of the edge-disjoint tree paths. So the bundle $B$ serves as a \emph{certificate} of edge connectivity, which is crucial for the sub-sampling based cut-sparsification procedures of \cite{kargersparse} and \cite{generalsparse}. We remark that we see momentarily that we only need to sparsify uncapacitated graphs, so every spanning forest of our graph will also be a minimum spanning forest.

\paragraph{A Standard Reduction to Uncapacitated Graphs}
Here, we will quickly show how simple techniques can be used to reduce proving \Cref{thm:fullydynamicsparsifier} for the case of uncapacitated graphs. We start by showing the following structural property.
\begin{lemma}\label{lma:cutunion}
    Let $G = (V, E_G, u_G)$ be a graph and $\{G_i\}_{i \in \mathcal{I}}$ a collection of graphs and $c_i \geq 0$ such that $G = \bigcup_{i \in \mathcal{I}} c_i \cdot G_i$. Suppose that for all $i \in \mathcal{I}$, we have that $H_i$ is a $(1 \pm \epsilon)$-cut sparsifier of $G_i$. Then $H:=\bigcup_{i \in \mathcal{I}} c_i \cdot H_i$ is a $(1 \pm \epsilon)$-cut sparsifier of $G$.
\end{lemma}
\begin{proof}
    Let $(S, V \setminus S)$ be a cut. Then for each $i$, we have $U_{H_i}(S) \in (1 \pm \epsilon)\cdot U_{G_i}(S)$, so that
    \begin{align*}
        U_{H}(S) = \sum_{i \in \mathcal{I}} c_i U_{H_i}(S) \leq \sum_{i \in \mathcal{I}} c_i (1 + \epsilon) U_{G_i}(S) = (1 + \epsilon) U_G(S),
    \end{align*}
    and the lower bound follows analogously.
\end{proof}
Now let $G$ be a graph with integer capacities in the set $\setof{0, 1, \ldots, U-1}$. Then by a binary expansion, we can write $G = \bigcup_{i=0}^{\log_2 U} 2^i G_i$ for uncapacitated graphs $G_i$, where $G_i$ contains edge $e \in E_G$ if the $(i+1)$-st digit of its binary expansion is $1$. Also note that edges moving by vertex splits do not change their capacities, so a vertex split in $G$ causes exactly one vertex split in each of the graphs $G_i$. Consequently, assuming a version of \Cref{thm:fullydynamicsparsifier} that works for uncapacitated graphs, we can run it on the graphs $G_i$ separately to receive sparsifiers $H_i$, and then return $H = \bigcup_{i=0}^{\log_2 U} 2^i H_i$ as the final sparsifier by an invokation of the previous lemma. In the subsequent discussions, we will thus always assume that the graph $G$ is uncapacitated.
\subsection{Fully Dynamic $l$-SF-bundles}
Our first step towards proving \Cref{thm:fullydynamicsparsifier} will be to show that we can maintain $l$-SF-bundles also under vertex splits. We prove the following.

\begin{theorem}\label{thm:dynamicbundle}
    Let $G = (V, E)$ be a fully dynamic, initially empty graph with $|V(G^{(0)})| = n$. Then, for a sequence of up to $m$ edge insertions/deletions and $n$ vertex splits, there is a deterministic algorithm that maintains an $l$-SF-bundle $B$ of $G$ with the following recourse properties:
    \begin{enumerate}
        \item If $G$ changes by an edge insertion, then $B$ changes at most by the insertion of the same edge.
        \item If $G$ changes by an edge deletion, then $B$ changes at most by the deletion of the same edge and possibly one additional edge insertion.
        \item If $G$ changes by a vertex split, then $B$ changes by performing the vertex split and at most $l$ additional edge insertions.
    \end{enumerate}
    The algorithm takes total time $\Otil(l \cdot m + l^2 \cdot n)$.
\end{theorem}
To do that, we need to use the data structure of \cite{10.1145/502090.502095} that allows us to dynamically maintain a spanning forest $F$ of a graph $G$ receiving only edge updates. The provided recourse property of the algorithm was crucially exploited in \cite{7782947} to control recourse in $l$-SF-bundles for graphs that only receive edge updates. It shows that the spanning forest maintained by the data structure changes only very slowly under edge deletions and insertions. 

\begin{theorem}[\cite{10.1145/502090.502095}]\label{thm:edgedynamicMST}
    Given a dynamic graph $G$ on $n$ vertices undergoing edge insertions / deletions, there exists a deterministic algorithm that maintains a minimum spanning forest $F$ of $G$. Furthermore, the following recourse properties hold:
    \begin{enumerate}
        \item If $G$ changes by an edge insertion, then $F$ changes at most by the insertion of the same edge.
        \item If $G$ changes by an edge deletion, then $F$ changes at most by the deletion of the same edge and possibly one additional edge insertion.
    \end{enumerate}
    The algorithm runs in amortized update time $O(\log^4 n)$. 
\end{theorem}

 In the next lemma, we show that even when allowing for vertex splits, there always exists a new spanning forest which differs only slightly from the previous one. 
\begin{lemma}\label{lma:splitrecourse}
    Let $F$ be a minimum spanning forest of $G$. Suppose vertex $v \in V$ of $G$ gets split into vertices $v', v''$ to receive $G'$. Then there exists a minimum spanning forest $F'$ of $G'$ that satisfies $E(F') \supseteq E(F)$. In particular, $F'$ can be obtained from $F$ by performing the vertex split of $v$ in $F$ and at most one extra edge insertion.
\end{lemma}
\begin{proof}
For every edge $e = (x, y) \in F$, there exists an $x$-$y$ cut $(S, V \setminus S)$ in $G$ such that $e$ is an edge of minimum capacity in the cut. Suppose w.l.o.g that $v \in S$. Let $S' := (S \setminus \setof{v}) \cup \setof{v', v''}$. If $x \neq v$, then it is evident that $e = (x, y)$ is still an edge of minimum capacity in the cut $(S', V' \setminus S')$ in $G'$, so it remains an edge in $F'$. If $x = v$, then we know that either $e' = (v', y)$ or $e'' = (v'', y)$ is an edge in $G'$ of capacity $u(e)$, and hence an edge of minimum capacity on the cut $(S', V' \setminus S')$. This proves that $E(F') \supseteq E(F)$. 

Now let $T$ be the tree of $F$ containing vertex $v$. After performing the vertex split, the tree $T$ will be partitioned into disjoint trees trees $T', T''$ containing vertices $v', v''$ respectively, such that $E(T) = E(T') \cup E(T'')$. If $v'$ and $v''$ share a connected component in $G'$, then we can receive $F'$ from $F$ by adding the minimum capacity edge that connects $T'$ and $T''$. If $v'$ and $v''$ do not share a connected component, we do not need to insert a new edge. 
\end{proof}

By assigning $G$ unique (and artificial) capacities \emph{only internally} when running the algorithm of \Cref{thm:edgedynamicMST}, we can make the the spanning forests unique, and consequently the previous lemma implies that the algorithm maintains a spanning forest under low-recourse. By now simulating a vertex split of $G$ as a sequence of a vertex insertion, edge deletions and edge insertions, the following corollary is immediate from \Cref{lma:splitrecourse} and \Cref{lma:splitefficient}.
\begin{corollary}\label{corollary:fullydynamicMST}
    Given a fully dynamic, initially empty graph $G = (V, E)$ with $|V(G^{(0)})| = n$ undergoing up to $m$ edge insertions/deletions and $n$ vertex splits, there exists a deterministic algorithm that maintains a spanning forest $F$ of $G$. The following recourse properties hold:
    \begin{enumerate}
        \item If $G$ changes by an edge insertion, then $F$ changes at most by the insertion of the same edge.
        \item If $G$ changes by an edge deletion, then $F$ changes at most by the deletion of the same edge and possibly one additional edge insertion.
        \item If $G$ changes by a vertex split, then $F$ changes by performing the vertex split and at most one additional edge insertion.
    \end{enumerate}
    The algorithm takes total time $\Otil(m)$. 
\end{corollary}

We can now show how to efficiently maintain an $l$-SF-bundle $B$ of a fully dynamic graph $G$ with low recourse. We simply maintain for each $k = 1, 2, \ldots, l$ the forest $F_k$ by calling the algorithm of \Cref{corollary:fullydynamicMST} on the fully dynamic graph $H_{k} := G\setminus \bigcup_{i=1}^{k-1} F_i$. The crucial insight is that while a vertex split causes up to $l$ additional updates in our bundle, these updates are only in the form of edge insertions, which themselves do not cause additional recourse.

\begin{proof}[Proof of \Cref{thm:dynamicbundle}]
The fact that $B$ at all time constitutes a valid $l$-SF-bundle of $G$ follows immediately from the correctness of \Cref{corollary:fullydynamicMST}. It remains to show the claims on recourse and running time. We start with the former.

Consider first property $1$ of \Cref{thm:dynamicbundle}. We want to show that when $G$ receives an edge insertion, then $B$ changes at most by the same edge insertion. We show the claim by induction on $l$. If $l = 1$, this is what \Cref{corollary:fullydynamicMST} guarantees us. Now suppose the claim has been shown for $l-1$. Then, when $G$ receives an edge insertion, $\cup_{i=1}^{l-1} F_i$ changes either by the insertion of edge $e$ or not at all. In the former case, we know that $H_l = G \setminus \cup_{i=1}^{l-1} F_i$ does not change at all, and consequently $F_l$ does not change at all. In the latter case, $H_l$ receives the insertion of edge $e$, in which case $F_l$ might change by the insertion of the same edge. In total, $\cup_{i=1}^{l} F_i$ changes at most by the insertion of edge $e$.

Consider now property $2$ of \Cref{thm:dynamicbundle}. We now want to show that when $G$ changes by an edge deletion, then $B$ changes at most by the same deletion and one additional edge insertion. We again show the claim by induction on $l$. When $l = 1$, this is guaranteed by \Cref{corollary:fullydynamicMST}. So suppose the claim has been shown for $l-1$, and that $G$ changes by the deletion of edge $e$. If the edge $e$ was not in $B$, there is nothing to show. If the edge $e$ was in $F_l$, the claim also follows from \Cref{corollary:fullydynamicMST}. If the edge $e$ was in $\cup_{i=1}^{l-1} F_i$, we know by the induction hypothesis that $\cup_{i=1}^{l-1} F_i$ changes by the deletion of edge $e$ and possibly by the insertion of an edge $e'$. This, on the other hand, implies that $H_l$ changes at most by the deletion of edge $e'$. If $e'$ was not in $F_l$, there is nothing to show. If, however, $e'$ was in $F_l$, then by \cref{corollary:fullydynamicMST}, we know that $F_l$ will change by the deletion of $e'$ and possibly the insertion of an additional edge $e''$. In total, however, this implies that $\cup_{i=1}^{l} F_i$ will have changed exactly by the deletion of edge $e$ and the insertion of edge $e''$, as desired.

We now deal with property $3$ of \Cref{thm:dynamicbundle}. We again show the claim by induction on $l$. If $l = 1$, this follows from \Cref{corollary:fullydynamicMST}. So suppose the claim has been shown for $l-1$. Then we know that $\bigcup_{i=1}^{l-1} F_i$ changed by performing the vertex split and at most $l-1$ additional edge insertions. Each edge that was inserted into $\bigcup_{i=1}^{l-1} F_i$ might constitute an edge deletion in $F_l$. We process the update by first feeding the vertex split into \Cref{thm:edgedynamicMST}, and then the edge deletions. The vertex split causes $F_l$ to change by at most one edge insertion. Each edge deletion to $H_l$ causes $F_l$ to change by at most the given deletion and one additional insertion. By the same argument as before, if an edge $e$ was deleted from $F_l$, and consequently a new edge $e''$ inserted into $F_l$, the union $\bigcup_{i=1}^{l} F_i$ changes only by the insertion of $e''$. In total, this implies that $\bigcup_{i=1}^{l} F_i$ changes at most by $l$ additional edge insertions.

Let us first consider a fixed $i$. By \Cref{lma:splitefficient}, the total time for moving the edges required to simulate all vertex splits is $\Otil(m)$ (we can move an edge in time $O(1)$ in the adjacency list data structure we use to store the graph). Each vertex split causes the data structure maintaining the $F_i$ at most $l$ additional updates. Over a sequence of $n$ vertex splits, this causes a total of $n \cdot l$ additional edge updates. Now the total time for maintaining all edge updates can thus be bounded by $\Otil(m+nl)$ per data structure, leading to a total running time of $\Otil(ml + nl^2)$.
\end{proof}

\subsection{Initializing a Cut Sparsifier}
In this section, we present the initialization procedure of the cut sparsifier, which is based on the algorithm provided in \cite{7782947}. Then, in the next section, we will show that we can make the algorithm dynamic with the help of the low-recourse dynamic $l$-SF-bundles from \Cref{thm:dynamicbundle}.

Before we start describing the algorithm, we need to state the following lemma from \cite{7782947}, which in turn was adapted from \cite{generalsparse}. It shows that by sub-sampling the edges of a graph $G$ based on their edge-connectivities yields a cut-sparsifier with high probability.
\begin{lemma}[See Lemma $5.6$ in \cite{7782947}]\label{lma:cutsample}
Let $G$ be an uncapacitated graph, and let $H$ be a sub-graph of $G$ obtained by sampling each edge $e$ with probability $p_e \geq \min \{1, \rho / \lambda_G(e)\}$, where $\rho = C_{\xi} c \log^2 n / \epsilon^2$ and $C_{\xi}$ is an explicitly known constant. If an edge $e$ is selected, it is assigned a capacity of $1/p_e$ in $H$. Then $H$ is a $(1 \pm \epsilon)$-cut sparsifier of $G$ with probability at least $1 - 1/n^c$.
\end{lemma}
The utility of an $l$-SF-bundle $B$ of $G$ now becomes apparent: for every edge $e \in G \setminus B$, it holds that $\lambda_G(e) \geq l$, i.e., the bundle $B$ gives us a lower bound on the edge connectivity of edges not contained in $B$. The consequence is collected in the following lemma.
\begin{lemma}
    Let $G$ be a graph and $l > 0$. Suppose that $B$ is an $l$-SF-bundle of $G$. Let $G'$ be a graph obtained by sub-sampling each edge $e \in G \setminus B$ with probability $p \geq C_{\xi}c \log^2 n / (l\epsilon^2)$ and assigning it a capacity of $1/p$ if sampled. Then $G' \cup B$ is a $(1 \pm \epsilon)$-cut sparsifier of $G$ with probability at least $1 - 1/n^c$.
\end{lemma}
\begin{proof}
This follows from \Cref{lma:cutsample}. Every edge $e \in B$ is included in $G' \cup B$ with probability $p_e = 1$, and every edge $e \in G \setminus B$ has $\lambda_G(e) \geq l$, so is included in $G' \cup B$ with probability $p_e \geq C_{\xi} c \log^2 n /(\epsilon^2 \lambda_G(e))$.
\end{proof}
This, in turn, yields the following corollary, which is heavily in exploited in the algorithm.
\begin{corollary}\label{lma:mainsparsifylemma}
    Let $G$ be a graph and $l \geq 2 C_{\xi} c \log^2 n / \epsilon^2$. Suppose that $B$ is an $l$-SF-bundle of $G$. Let $G'$ be a graph obtained by sub-sampling each edge $e \in G \setminus B$ with probability $1/2$ and assigning it a capacity of $2$ if sampled. Then $G' \cup B$ is a $(1 \pm \epsilon)$-cut sparsifier of $G$ with probability at least $1 - 1/n^c$.
\end{corollary}
\begin{remark}
Let $G$ be a capacitated graph with uniform capacities $u_G(e) = U$ for all $e \in E_G$ and let $G'$ be as defined in \Cref{lma:mainsparsifylemma}. Then by \Cref{lma:cutunion} and the above result, by assigning each edge $e \in B$ a capacity of $u_H(e) = U$ and each edge $e \in G'$ a capacity of $u_H(e) = 2 U$, the graph $H = G' \cup B$ is a $(1 \pm \epsilon)$-cut sparsifier of $G$ with probability at least $1 - 1/n^c$.
\end{remark}
\paragraph{The Algorithm}Using the above corollary, we can now describe the algorithm. We let $G_0 = G$ and compute $B_1$ as an $l_s := \lceil 2 \cdot C_{\xi} \cdot 40 \cdot \log^2 (2n) \cdot \epsilon^{-2} \rceil$-SF-bundle of $G_0$. Then, we receive $G_1$ by sub-sampling each edge $e \in G_0 \setminus B_1$ with probability $1/2$ and, if sampled, assigning it capacity $2$. We then iteratively compute $B_i$ as an $l_s$-SF-bundle of $G_{i-1}$ \emph{without capacities}. That is, we run the data structure from \Cref{thm:dynamicbundle} on $G_{i-1}$ without considering its capacities (which are uniformly $2^{i-1}$.) Each edge $e \in B_i$ is afterwards assigned a capacity of $u_{B_i}(e) = u_{G_{i-1}}(e) = 2^{i-1}$.

Then, we construct $G_i$ by independently sub-sampling each edge in $G_{i-1} \setminus B_i$ with probability $1/2$ and assigning it capacity $2 u_{G_{i-1}}(e) = 2^i$ if selected. We perform this for $\rho = \lceil \log (2n) \rceil$ many steps, and return $H = G_{\rho} \cup \bigcup_{i=1}^{\rho} B_i$ as the final sparsifier. 

\begin{theorem}[See also section $5.3$ in \cite{7782947}]\label{thm:initcorrect}
    Let $G$ be a (multi-)graph on at most $2n$ vertices, $m \leq O(n^2)$ edges, and let $\epsilon > 0$ be a parameter. Then, the output $H$ of $\textsc{InitCutSparsifier}(G, \epsilon)$ is a $(1 \pm \epsilon)$ cut sparsifier of $G$ that satisfies $|E(H)| = \Otil(n \cdot \epsilon^{-2})$ with probability $1- 1/n^c$ for any constant $c > 0$ specified before the algorithm starts. 
\end{theorem}
\begin{proof}
    We start by showing that $|E(H)| = \Otil(n \cdot \epsilon^{-2})$. Note first that the number of edges in $\bigcup_{i=1}^{\rho} B_i$ can be bounded by $\rho \cdot l_s \cdot  n = \Otil(n \epsilon^{-2})$, so we only need to bound the number of edges in $G_{\rho}$. 
    
     To do so, consider the indicator random variables $\mathbf{1}\{e \in G_{\rho}\}$. We have $\E{\mathbf{1}\{e \in G_{\rho}\}} = \Pr[e \in G_{\rho}] \leq 2^{-\rho} \leq n^{-1}$, so that $\E{|G_{\rho}|} = \sum_{e \in E} \E{\mathbf{1}\{e \in G_{\rho}\}} \leq m/n = O(n)$. We would now like to apply the Chernoff-Bound from \Cref{thm:chernoff} in order to show that also with probability at least $1-n^{-40}$, we have that $|G_{\rho}| \leq \Otil(n)$. Note however that, in general, $\mathbf{1}\{e \in G_{\rho}\}$ is not independent of $\mathbf{1}\{e' \in G_{\rho}\}$ for two edges $e, e' \in E_G$, as the status of edge $e'$ could affect whether or not $e$ gets included in a bundle $B_i$ or not. 
    
    To resolve this issue, we consider a coupling as follows: for each edge $e \in G$, we let $b_{e, i} \sim \textnormal{Ber}(1/2)$ be another indicator random variable representing the independent coin flips that determine if, in the case that $e \in G_{i-1} \setminus B_i$, the edge $e$ also gets sampled to $G_i$ (i.e., the coin flips used in \Cref{algo:cutSparsify}, Line~\ref{alg:static:sample}). Since for an edge $e \in G_{\rho}$ it is necessary that it was sampled into every $G_i$, we have that, almost surely, $\mathbf{1}\{e \in G_{\rho}\} \leq \prod_{i=1}^{\rho} b_{e, i}=: b_e$. In particular, we have that $|E(G_{\rho})| \leq \sum_{e \in E} b_e$ almost    surely. Let $X = \sum_{e \in E} b_e$, and notice that $\E{X} = m 2^{-\rho} = O(n)$. Hence, by \Cref{thm:chernoff}, we have that $\Pr[X \geq \Omega(n \log n)] \leq \exp \rbrack{-\Omega(n \log n)} \leq n^{-40}$. It follows that $|E(G_{\rho})| \leq O(n \log n)$ with probability at least $1-n^{-40}$.
    
    We now show that $H$ is a $(1 \pm \epsilon)$-cut sparsifier of $G$. First note that by \Cref{lma:mainsparsifylemma} and the remark that followed, independent of the randomness in $G_{i-1}$, we have that $G_i \cup B_i$ is a $(1 \pm \epsilon)$-cut sparsifier of $G_{i-1}$ with probability at least $1 - n^{-40}$. Now let $\cE_i$ be the event that $G_i \cup B_i$ is a $(1 \pm \epsilon)$-cut sparsifier of $G_{i-1}$, and let $\cE = \bigcup_{i=1}^{\rho} \cE_i$. We next show that conditioned on the event $\cE$ we have that $H$ is a $(1 \pm \epsilon)$-cut sparsifier of $G$. To do so, let $H_j := G_j \cup \bigcup_{i=1}^j B_j$. We prove by induction on $j$ that on the event $\cE$, $H_j$ is a $(1 \pm \epsilon)$-cut sparsifier of $G$. That $G_1 \cup B_1$ is a $(1 \pm \epsilon)$-cut sparsifier of $G$ follows from \Cref{lma:mainsparsifylemma}. Now assume that $G_j \cup \bigcup_{i=1}^j B_j$ is a $(1 \pm \epsilon)$-cut sparsifier of $G$. Then since $G_{j+1} \cup B_{j+1}$ is a $(1 \pm \epsilon)$-cut sparsifier of $G_j$ conditioned on the event $\cE$, by \Cref{lma:cutunion} we have that $G_{j+1} \cup \bigcup_{i=1}^{j+1} B_j$ is a $(1 \pm \epsilon)$-cut sparsifier of $G$, as desired. 

    Using a simple union bound we can now finish by noting that $\Pr[\neg \cE] \leq \sum_{i=1}^{\rho} \Pr[\neg \cE_i] \leq \rho n^{-40} \leq n^{-30}$ by the success probability guarantee from \Cref{lma:mainsparsifylemma}. Finally, we remark that the failure probability $n^{-30}$ was arbitrary and can straightforwardly be boosted to $1/n^c$ for any $c > 0$ by setting $l_s = \lceil 2 C_{\xi} c \log^2(2n) \epsilon^{-2} \rceil$ instead.
\end{proof}

\begin{algorithm2e}
$G_0 \gets G = (V, E)$ \\
$n \gets |V|$ \\
$\rho \gets \lceil \log(2n) \rceil$ \\
$l_s \gets \lceil 2 \cdot C_{\xi} \cdot 40 \cdot \log^2 (2n) \cdot \epsilon^{-2} \rceil$\label{alg:static:li} \CommentSty{\hspace{1em} \textbackslash\textbackslash \hspace{0.2em} Let $C_{\xi}$ be the constant from \Cref{lma:cutsample}} \\
\For{$i=1$ to $\rho$}{
    \label{algo:cutSparsify:li}
    Initialize $B_{i}$ as a dynamic $l_s$-SF-bundle of $G_{i-1}$ using \Cref{thm:dynamicbundle} \\
    $G_i \gets (V, \emptyset)$ \\
    \For{$e \in E(G_{i-1}) \setminus E(B_i)$}{
        with probability $1/2$, add $e$ to $G_i$ with capacity $u_{G_i}(e) \gets 2^{i}$ \label{alg:static:sample}
    }
}
$H \gets G_{\rho} \cup \bigcup_{i=1}^{\rho} B_i$
\caption{$\textsc{InitCutSparsifier}(G, \epsilon)$}
\label{algo:cutSparsify}
\end{algorithm2e}
\subsection{Fully Dynamic Cut Sparsifier}
We will now show how to use fully dynamic $l$-SF-bundles to prove \Cref{thm:fullydynamicsparsifier}. In \cite{7782947} the static \Cref{algo:cutSparsify} was made dynamic to handle edge updates by using a dynamic MST data structure as given by \Cref{thm:edgedynamicMST}. We will show that by using the dynamic $l$-SF-bundle data structure from \Cref{thm:dynamicbundle} that can also handle vertex splits, we can also allow for vertex splits as an update for the cut sparsifier. The correctness of this approach will follow straightforwardly from \Cref{thm:initcorrect}; the main task is to show the recourse and running time bounds through careful analysis of how updates made by vertex splits are handled in the hierarchy of fully dynamic graphs in \Cref{algo:cutSparsify}. 

\paragraph{The Dynamic Algorithm} First, let us notice that in order to make \Cref{algo:cutSparsify} dynamic, we only need to ensure that at all times $B_i$ is an $l_s$-SF-bundle of $G_{i-1}$ (where $l_s = \lceil 2 \cdot C_{\xi} \cdot 40 \cdot \log^2 (2n) \cdot \epsilon^{-2} \rceil)$, and that $G_i$ is a sample from $G_{i-1} \setminus B_i$ as specified in Line~\ref{alg:static:sample} of \Cref{algo:cutSparsify}, i.e., that an edge $e \in G_{i-1} \setminus B_i$ satisfies $\Pr[e \in G_i] = 1/2$. If this is satisfied, then by \Cref{thm:initcorrect}, we have that $G_{\rho} \cup \bigcup_{i=1}^{\rho} B_i$ is a $(1 \pm \epsilon)$-cut sparsifier of $G$ with high probability. As we are assuming an oblivious adversary, we can union bound over all updates to get the final high probability claim. It thus  remains to show how to efficiently maintain such $B_i$ and $G_i$ with low recourse. In the following, we denote by $G_i^{(t)}, B_i^{(t)}$ the graph $G_i, B_i$ respectively after the $t$'th update to $G$ has been processed by the data structure.

We can maintain $B_i$ straightforwardly, as summarized in \Cref{algo:updatecutSparsify}: let $U_{i-1}^{(t)}$ be the batch of updates performed on $G_{i-1}^{(t)}$ to receive $G_{i-1}^{(t+1)}$. Then we receive $B_i^{(t+1)}$ from $B_i^{(t)}$ by feeding the batch of updates $U_{i-1}^{(t)}$ to the data structure of \Cref{thm:dynamicbundle}. Then, we receive $G_{i}^{(t+1)}$ from $G_i^{(t)}$ by the following update batch $U_i^{(t)}$: 
\begin{itemize}
    \item If a vertex $v$ was split in the input graph $G$ such that $N_{G}(v') = N_G(v) \setminus U$ and $N_G(v'') = U$, split the vertex $v$ in $G_i$ accordingly by setting $N_{G_i}(v') =  N_{G_i}(v) \cap N_G(v) \setminus U$ and $N_{G_i}(v'') =N_{G_i}(v) \cap U$.
    \item For every edge $e$ that was added to $G_{i-1}^{(t+1)}$ but not included in its $l_s$-SF-bundle $B_i^{(t+1)}$, i.e., for $e \in (E(G_{i-1}^{(t+1)}) \setminus E(B_i^{(t+1)})) \setminus (E(G_{i-1}^{(t)}) \setminus E(B_i^{(t)}))$, add it to $G_i^{(t+1)}$ with probability $1/2$ and assign it capacity $2^i$ if sampled. 
    \item For every edge $e$ that was deleted from $E(G_{i-1}^{(t+1)}) \setminus E(B_i^{(t+1)})$, i.e., for $e \in (E(G_{i-1}^{(t)}) \setminus E(B_i^{(t)})) \setminus (E(G_{i-1}^{(t+1)}) \setminus E(B_i^{(t+1)}))$, remove the edge also from $G_i^{(t)}$ (if it had existed in the graph). 
\end{itemize}

\paragraph{Analysis of \Cref{algo:updatecutSparsify}.}
We start by showing that at all times, the distribution of $G_{\rho} \cup \bigcup_{i=1}^{\rho} B_i$ remains correct, which shows that indeed we maintain a cut-sparsifier with high probability. 
\begin{claim}\label{clm:correctness}
    At all times, $B_i^{(t)}$ is an $l_s$-SF-bundle of $G_{i-1}^{(t)}$ and, for each edge $e \in G_{i-1}^{(t)} \setminus B_i^{(t)}$, we have that $\Pr[e \in G_{i}^{(t)}] = 1/2$. In particular, $H^{(t)}$ satisfies $|E(H^{(t)})| \leq \Otil(n \cdot \epsilon^{-2})$ and is a $(1 \pm \epsilon)$-cut sparsifier of $G^{(t)}$ with probability $1-1/n^c$.
\end{claim}
\begin{proof}
  We prove the claim by induction on time $t$. The case $t=0$ follows from initialization. Supposing the statement is true for time $t$, we show that it still holds for time $t+1$. 
  
  First, that $B_i^{(t+1)}$ is an $l_s$-SF-bundle of $G_{i-1}^{(t+1)}$ follows immediately from the correctness of \Cref{thm:dynamicbundle}. It remains to show the part about $G_i^{(t+1)}$. So let $e \in G_{i-1}^{(t+1)} \setminus B_i^{(t+1)}$. We want to verify that $\Pr[e \in G_{i}^{(t+1)}] = 1/2$. There are two cases. If $e$ was already in $G_{i}^{(t)} \setminus B_i^{(t)}$, then $\Pr[e \in G_{i}^{(t+1)}] = \Pr[e \in G_{i}^{(t)}] = 1/2$ by the induction hypothesis. Otherwise, by construction of the update batch $U_i^{(t)}$, the edge $e$ is sampled as specified in Line~\ref{alg:update:sample} of \Cref{algo:updatecutSparsify} to directly satisfy $\Pr[e \in G_{i}^{(t+1)}] = 1/2$. Since the number of vertices of $G^{(t+1)}$ is still bounded by $2n$ (as $G$ in total undergoes at most $n$ vertex splits), the guarantees of \Cref{thm:initcorrect} then also show that $H^{(t+1)}$ satisfies $|E(H^{(t+1)})| \leq \Otil(n \cdot \epsilon^{-2})$ and is a $(1 \pm \epsilon)$-cut sparsifier of $G^{(t+1)}$ with probability $1-1/n^c$.
\end{proof}
We now show the following claim on recourse.
\begin{claim}\label{clm:sparsifyrecourse}
    Let $H$ be the cut-sparsifier of $G$ as maintained by \Cref{algo:updatecutSparsify}, and $l_s = \lceil 2 \cdot C_{\xi} \cdot 40 \cdot \log^2 (2n) \cdot \epsilon^{-2} \rceil$. Then the following recourse properties hold:
    \begin{enumerate}
        \item If $G$ changes by an edge insertion, then $H$ changes at most by the insertion of the same edge.
        \item If $G$ changes by an edge deletion, then $H$ changes at most by the deletion of the same edge and possibly one extra edge insertion.
        \item If $G$ changes by a vertex split, then $H$ changes by at most $\rho \cdot l_s$ additional edge insertions.
    \end{enumerate}
\end{claim}

\begin{proof}
    Consider first property $1$. We want to show that when $G$ receives an edge insertion, then $H$ changes at most by the insertion of the same edge (albeit possibly with a different capacity). This follows immediately from the fact that the $B_i$'s can only change by the insertion of the given edge, as guaranteed by \Cref{thm:dynamicbundle}, and noting that the extra sub-sampling  in Line~\ref{alg:update:sample} of Algorithm~\ref{algo:updatecutSparsify} can only have the effect of the $B_i$'s not changing at all, which is even better. 

    Suppose now that $G$ changes by the deletion of edge $e$. If $e \not \in H$, the sparsifier does not change at all and we are done. So suppose that $e \in H$. In the case that $e \in G_{\rho}$, the sparsifier $H$ changes only by the deletion of edge $e$. So it remains to show the case for an edge $e$ being deleted from some $B_j$. We prove by induction on $j$ that in that case, the graph 
    $\bigcup_{i=j}^{\rho} B_i$ changes by at most one extra edge insertion. So consider first the case that $e \in B_{\rho}$. Then by \Cref{thm:dynamicbundle}, we know that $B_{\rho}$ changes at most by one extra edge insertion, as desired. Suppose the statement has been proven for all $\rho \geq i \geq j$, and that an edge from $B_{j-1}$ was deleted. Then from \Cref{thm:dynamicbundle}, we know that $B_{j-1}$ changes at most by the insertion of one extra edge $e''$. If $e''$ was not in $H$ beforehand, no further changes will be performed, and we are done. Otherwise, this corresponds to an edge deletion either from $G_{\rho}$ or some $B_i$ with $i > j-1$. In the former case, nothing remains to be shown. In the latter case, by the induction hypothesis, we know that $\bigcup_{i=j}^{\rho} B_j$ changes at most by the insertion of one extra edge, so that $\bigcup_{i=j-1}^{\rho} B_i$ changes at most by the insertion of one extra edge. 

    It remains to analyse vertex splits. We prove that the graph $\bigcup_{i=1}^j B_i$ changes by at most $j \cdot l_s$ additional edge insertions by induction on $j$. For the case $j = 1$, this follows from \Cref{thm:dynamicbundle}. So suppose the statement has been proven for $j-1$. Then to update $B_j$, we first perform the vertex split, which causes an additional $l_s$ edge insertions to $B_j$ by \Cref{thm:dynamicbundle}. By the induction hypothesis, $\bigcup_{i=1}^{j-1} B_i$ changed by up to $(j-1) \cdot l_s$ additional edge insertions. Each of those insertions can possibly correspond to an edge deletion in $B_j$. We process the edge deletions one by one in $B_j$, and by \Cref{thm:dynamicbundle}, we know that for each such deletion, $B_j$ might change by the insertion of one additional edge. In total, we have that $\bigcup_{i=1}^{j} B_i$ changes by up to $j \cdot l_s$ edge insertions, as we wanted to show. 
\end{proof}
We can now conclude with a proof of the full theorem.
\begin{proof}[Proof of \Cref{thm:fullydynamicsparsifier}]
    As mentioned in the beginning, we can assume without loss of generality that $G$ is uncapacitated. Then, correctness follows from \Cref{thm:initcorrect}, \Cref{clm:correctness}, union bounding over all $O(m)$ failure events and by noting that the capacities in $H$ are bounded by $2^\rho = O(n)$. 
    
    To prove the statement about recourse, first note that by \Cref{clm:sparsifyrecourse}, edge insertions and vertex splits performed on $G$ cause at most $1$ respectively $\rho l_s = O(\polylog(n))$ additional edge insertions in $H$ and, crucially, no edge deletions to $H$. Thus, since at all times $|E(H)| \leq \Otil(n \cdot \epsilon^{-2})$ and an edge leaves $H$ only when it gets deleted from $G$, we have that the total number of edge deletions and edge insertions by which $H$ changes must be bounded by $O(\polylog(n) \cdot \epsilon^{-2} \cdot (n+D))$, where $D$ is the total number of edge deletions that $G$ receives. Note, however, that $H$ additionally changes by $n$ vertex splits. Fortunately, \Cref{lma:splitefficient} shows that simulating the vertex splits by edge deletions and edge insertions causes at most additional $O(\polylog(n) \cdot \epsilon^{-2} \cdot (n+D))$ many updates.
    
    To show the running time guarantees, note first that each bundle $B_i$ receives at most $n$ vertex splits, so by \Cref{thm:dynamicbundle}, these can be handled in total time $\Otil(\rho l_s \cdot m + \rho \cdot l_s^2 \cdot n) \leq \Otil(m \rho l_{s}^2) = \Otil(m \cdot \epsilon^{-4})$. Furthermore, by \Cref{clm:sparsifyrecourse}, the total amount of edge updates each data structure maintaining the $B_i$ has to handle can be bounded by $O(m + \rho l_s \cdot n) = O(m + \rho l_{s} n) = \Otil(m \cdot \epsilon^{-2})$. Since each edge update can be handled in amortized time $O(\polylog(n))$, the total amount of time spent for maintaining the $B_i$'s is $\Otil(\rho \cdot m \cdot \epsilon^{-2}) = \Otil(m \cdot \epsilon^{-2})$. 
\end{proof}

\begin{algorithm2e}
\CommentSty{\textbackslash\textbackslash \hspace{1em}Here $u$ either encodes an edge insertion/deletion or a vertex split.} \\
$U_0^{(t)} \gets \setof{u}$ \\
\For{$i=1$ to $\rho$}{
    Forward the update batch $U_{i-1}^{(t)}$ of $G_{i-1}^{(t)}$ to the data structure of \Cref{thm:dynamicbundle}. \\
    Update $G_i^{(t)}$ by the update batch $U_i^{(t)}$ constructed as follows \\
    \If{$U_0^{(t)}$ \textnormal{contains vertex splits}}{
        add all vertex splits to $U_i^{(t)}$
    }
    \For{$e \in (E(G_{i-1}^{(t+1)}) \setminus E(B_i^{(t+1)})) \setminus (E(G_{i-1}^{(t)}) \setminus E(B_i^{(t)}))$}
    {
    With probability $1/2$, add insertion of edge $e$ with capacity $2^i$ to $U_i^{(t)}$\label{alg:update:sample}
    }
    \For{$e \in (E(G_{i-1}^{(t)}) \setminus E(B_i^{(t)})) \setminus (E(G_{i-1}^{(t+1)}) \setminus E(B_i^{(t+1)}))$}
    {
    Add removal of edge $e$ to $U_i^{(t)}$
}
}
\caption{$\textsc{UpdateCutSparsifier}(u)$}
\label{algo:updatecutSparsify}
\end{algorithm2e}

\section{Hierarchical \atmost{j}-Tree}
\label{sec:hierarchy}


In this section, we prove \Cref{th:main} by leveraging the tools we developed in \Cref{sec:madry,sec:cut-sparsifier} within an \((L,j)\)-hierarchy, the central concept of this section.
We start by restating \Cref{th:main} and discussing some remarks about it.

\main*


\begin{remark}
We note the trade-off in the choice of \(L\) in \Cref{th:main}.
Choosing \(L\) as a constant results in a cut quality of \poly{\log n} and an amortized update time of \atmosttilde{(n/j)^{2/L}}.
By selecting a larger constant for \(L\), the update time can be made arbitrarily small, yet still polynomial.
On the other hand, setting \(L = \Theta({\sqrt{\log n} / \log \log n})\) improves the amortized update time to \(n^{\littleo{1}}\), at the cost of a worse cut quality of \(n^{\littleo{1}}\).
\end{remark}

\begin{remark} \label{rmk:L}
The upper bound \(L \leq \littleo{ \sqrt{ \log (n/j) / \log \log n }} \) in \Cref{th:main} is to ensure that the term \(\log ^{\atmost{L}}n\) does not dominate the update time for larger values of \(L\).
This is due to the fact that
\begin{align*}
\log ^{\atmost{L}} n &= \exp\left( \atmost{L \log \log n} \right)
\\
&= \exp\left( \frac{1}{L} \atmost{L^2} \log \log n \right)
\\
&= \exp\left( \frac{1}{L} \littleo{ \frac{\log (n/j)}{\log \log n} } \log \log n \right)
&&\text{since \(L \leq \littleo{ \sqrt{ \log (n/j) / \log \log n }} \)}
\\
&= \exp\left( \littleo{\frac{1}{L} } \log (n/j) \right) 
\\
&= \left( \frac{n}{j} \right)^{\littleo{1/L}}.
\end{align*}
\end{remark}

The section is organized as follows.
We begin by defining the \((L,j)\)-hierarchy and proving some basic properties in \Cref{subsec:hierarchy}. 
Importantly, we show that the hierarchy naturally results in a set \(\mathcal H\) of $O(j)$-trees. 
In \Cref{subsec:single-level}, we describe how to initialize and maintain a single level of the hierarchy, which will be used as a crucial sub-routine in the full construction.
We then proceed to \Cref{subsec:data-structure}  where we explain the initialization and maintenance of the data structure of \Cref{th:main}.
Finally, we prove \Cref{th:main} in \Cref{subsec:proof-main-thm}.

\subsection{Hierarchy Definition and Basic Properties} \label{subsec:hierarchy}

As mentioned before, to maintain the set \(\mathcal H\) of \atmost{j}-trees, 
we use a hierarchical construction called an \((L,j)\)-hierarchy of \(G\), which is the central concept of this section.
Before formally defining the hierarchy, we first review a definition from \Cref{sec:madry}.




\SparsifiedCore*



We now proceed to define the \((L,j)\)-hierarchy.

\begin{definition}[\((L,j)\)-hierarchy of \(G\)] \label{def:hierarchy}
Given an \(n\)-vertex graph \(G=(V,E,u)\) and positive integers \(j,L \leq n\), let \(n = j_0 > j_1 > \dots > j_{L-1} > j_L = j\) be a strictly decreasing sequence of integers. Then, an \((L,j)\)-hierarchy of \(G\) is a sequence \(\mathcal H_0, \dots, \mathcal H_L\) of sets such that
\begin{enumerate}
\item
The set \(\mathcal H_0\) is a singleton containing 
\(G\),
i.e., \(\mathcal H_0 = \{ G \}\).
The graph \(G\) is treated as the trivial $n$-tree\footnote{The trivial $n$-tree consists of the core defined on the whole vertex set and the empty forest.} of \(G\) and $\Tilde{C}(G)$ denotes a $2$-approximate cut-sparsifier of $G$ (\Cref{def:sparsified-core}). 
\item For every \(1 \leq i \leq L\), and \(H_i \in \mathcal H_i\), there exists a unique parent graph $H_{i-1} \in \cH_{i-1}$ such that $H_i$ is an $O(j_i)$-tree on the vertex set of $\Tilde{C}(H_{i-1})$, a \(2\)-approximate sparsified core of $H_{i-1}$. Moreover, each graph $H_i$ has an associated $2$-approximate sparsified core $\Tilde{C}(H_i)$. 
\end{enumerate}
For every $H_{i} \in \mathcal H_{i}$, we denote by $\cH(H_{i})$ the set of graphs $H_{i+1} \in \cH_{i+1}$ that have $H_{i}$ as their parent.
\end{definition}

In an $(L, j)$-hierarchy, the set \(\mathcal H_L\) is defined as a set of \atmost{j}-trees on the vertex sets of the sparsified cores of \(\mathcal H_{L-1}\), whereas the set \(\mathcal H\) of \atmost{j}-trees in \Cref{th:main} is defined on the full vertex set of the input graph \(G\). 
We now describe how we will build the set $\cH$ using the \((L,j)\)-hierarchy. 
We begin by defining the notion of a chain in the \((L,j)\)-hierarchy.

\begin{definition}[Chain in an \((L,j)\)-hierarchy] \label{def:chain}
Consider an \((L,j)\)-hierarchy \(\mathcal H_0, \dots, \mathcal H_L\).
The sequence \(H_0, \dots, H_L\) of graphs is a chain in the hierarchy if, for every \(1 \leq i \leq L\), the graph $H_{i-1}$ is the parent of \(H_i\) in the hierarchy.
i.e., $H_i \in \cH(H_{i-1})$.
\end{definition}
Now, given a chain $H_0, H_1, \dots, H_L$ in the hierarchy, let $F_i = F(H_i)$ and $\Tilde{C}_i = \Tilde{C}(H_i)$ be the corresponding forest and the sparsified core of the \atmost{j_i}-tree \(H_i\), respectively. 
Then, we associate with the chain the graph 
\begin{equation} \label{eq:chain-graph}
H^{H_0, \dots, H_L} := \Tilde{C}_L \cup F_L \cup \dots \cup F_1.
\end{equation}
We now show that the associated graph of any chain is an $O(j)$-tree on the vertex set $V$ of the input graph $G$.
\begin{lemma}\label{lma:chaintree}
    For any chain $H_0, \dots, H_L$ in the hierarchy \(\mathcal H_0, \dots, \mathcal H_L\), the graph $H^{H_0, \dots, H_L}$ is an $O(j)$-tree on the vertex set $V(G)$.
\end{lemma}
\begin{proof}
By induction on $1 \leq i \leq L$, we prove that the graph \(\widetilde C_i \cup F_i \cup \dots \cup F_{1}\) is an \atmost{j_i}-tree.
The claim then follows from the fact that \(j_L = j\). 

\begin{itemize}
\item \underline{Base case:} 
by definition, $H_1$ is an $O(j_1)$-tree with core $C_1$ and forest $F_1$.
Hence $\Tilde{C}_1 \cup F_1$ is also an $O(j_1)$ tree by \Cref{clm:sparsifiedjtree}.

\item \underline{Induction step:}
assume that, for some level \(i\), the forest \(F_{i} \cup \dots \cup F_1\) along with the sparsified core \(\Tilde{C}_i\) forms an \atmost{j_i}-tree. 
By definition, \(H_{i+1} = C_{i+1} \cup F_{i+1}\) is an \atmost{j_{i+1}}-tree on vertex set $V(\Tilde{C}_i)$, which means that $\Tilde{C}_{i+1} \cup F_{i+1}$ is also an \atmost{j_{i+1}}-tree on the same vertex set.
Thus, the root of each tree in $F_{i} \cup \dots \cup F_1$ belongs to exactly one tree of $F_{i+1}$, so that $F_{i+1} \cup F_{i} \cup \dots F_{1}$ is indeed a rooted forest with root set $V(\Tilde{C}_{i+1})$. 
\end{itemize}
\end{proof}

Consequently, with each $(L, j)$-hierarchy, we can associate a corresponding set $\cH$ of $O(j)$-trees, as described below. 
\begin{definition}[Corresponding set \(\mathcal H\) of an \((L,j)\)-hierarchy]
    Given an $(L, j)$-hierarchy \(\mathcal H_0,\dots,\mathcal H_L\), we define the associated set $\cH$ of $O(j)$-trees as
    \[
    \cH := \setof{H^{H_0, \dots, H_L}: H_0, \dots, H_L \tn{ is a chain in the hierarchy}}.
    \]    
\end{definition}



 In addition to $\cH$ consisting of $O(j)$-trees, we  need to guarantee that each cut is approximately preserved by $\cH$ with high probability (\Cref{def:cut-preserving-collection}). 
As stated in \Cref{th:main}, the hierarchy we construct and maintain will satisfy this extra condition.

In the remainder of this section, we will briefly sketch how we initialize and maintain such an $(L, j)$-hierarchy. Then, in \Cref{subsec:single-level,subsec:data-structure} we will give a formal description, and in \Cref{subsec:proof-main-thm}, we will prove that the maintained associated set $\cH$ of $O(j)$-trees satisfies the conditions of \Cref{th:main}. 

First notice that, for any \(1 \leq i \leq L\), we can write \(\mathcal H_i = \bigcup_{H_{i-1}\in \mathcal H_{i-1}} \mathcal H(H_{i-1})\). This suggests that, to initialize and maintain an \((L,j)\)-hierarchy \(\mathcal H_0, \dots, \mathcal H_L\), we can focus on initializing and maintaining each set \(\mathcal H(H_{i-1})\), along with the sparsified cores \(\widetilde C_{i-1}\) on whose vertex set the graphs in \(\mathcal H(H_{i-1})\) are defined. To achieve this, we will employ two data structures from the previous sections:
\begin{itemize}
\item To initialize and maintain the set \(\mathcal H(H_{i-1})\), we employ \Cref{thm:jtree} to obtain a collection of \atmost{j_i}-trees on the sparsified core \(\widetilde C_{i-1}\). 
Because the collection of graphs obtained from \Cref{thm:jtree} is too large for our purposes, we drastically reduce its size by sampling \atmost{\log n} elements uniformly at random, as detailed in \Cref{subsec:single-level}.
\item The sparsified core \(\widetilde C_{i}\) of the (dense) core $C_{i}$ is maintained
using the cut-sparsifier of \Cref{thm:fullydynamicsparsifier}, which can handle vertex splits and a large number of edge insertions in $C_i$.
We re-state the guarantees of the theorem in the corollary below.
\end{itemize}

\sparsifiedcoremaintain*


%


To handle updates in the hierarchy (and thus maintain \(\mathcal H\)), 
  the data structure must efficiently handle the interactions between \(\mathcal H_{i+1}\) and \(\mathcal H_i\) for every \(1 \leq i \leq L\). We will discuss this in full detail in \Cref{subsec:hierarchy}.
Below, we briefly discuss the main points in maintaining the hierarchy. 

\paragraph{Handling Recourse Propagation.}
Given a chain $H_0, H_1, \dots, H_L$ in the hierarchy, a single update to $H_0$ might cause multiple updates to $H_1$, each of which causes multiple updates to $H_2$, and so on.
By \Cref{thm:jtree}, the core $C_{i-1}$ of $H_{i-1}$ undergoes $O(1)$ vertex splits after a single update in $H_{i-1}$.
Consequently, simulating each vertex split as a sequence of edge insertions and deletions results in $\Omega(j_{i-1})$ recourse per update,
and since each \atmost{j_i}-tree \(H_i \in \mathcal H_i\) has a core of size $\Theta(j_i)$, an update in \(G\) could result in $\Omega({j_{1} j_{2} \dots j_L})$ updates in graphs \(H_L \in \mathcal H_L\), which we cannot afford.
Thus, using \Cref{thm:jtree} directly on $C_{i-1}$ to initialize and maintain the set $\cH(H_{i-1})$ is computationally infeasible for our purpose.
 
In order to resolve this issue, we  initialize the set $\cH(H_{i-1})$ by calling \Cref{thm:jtree} on a sparsified core $\widetilde C_{i-1}$ of $C_{i-1}$, which is maintained using \Cref{thm:fullydynamicsparsifier}. 
That is, if $\gamma_{\tn{rec}}$ denotes an upper bound on the (amortized) number of updates any $H_{i}$ undergoes in response to a single update to its parent $H_{i-1}$, then a single update to the initial graph $G$ can cause up to $\gamma_{\tn{rec}}^i$ updates to each graph $H_i \in \cH_i$. 
By \Cref{thm:sparsifiedcore}, we can ensure that $\gamma_{\tn{rec}} = \poly{\log n} \log U$. 
Assuming $U = \poly{n}$, this results in $\gamma_{\tn{rec}}^i = \log^{O(i)} n$.\footnote{As  \Cref{thm:sparsifiedcore} only guarantees a recourse of $\gamma_{\tn{rec}} = \polylog(n)$ amortized over the whole sequence of updates, the final analysis in \Cref{subsec:proof-main-thm} is more involved.}

%


\paragraph{Rebuilding of Levels.}
Another issue arising during the maintenance of the hierarchy is that any $O(j_i)$-tree $H_i \in \mathcal H_i$ can only handle at most $O(j_i)$ many updates:
we have seen in \Cref{thm:jtree} that the size of the core $C_i$ grows proportionally to the number of updates the underlying graph receives. 
Hence, as we want to guarantee that $H_i$ stays an $O(j_i)$-tree throughout, we cannot afford to handle more than $O(j_i)$ many updates. 

To handle this issue, as we will discuss in \Cref{subsec:data-structure}, we re-compute the sets $\cH_i, \cH_{i+1}, \dots, \cH_L$ once there is a graph $H_i \in \cH_i$ whose core has doubled in size since its initialization.

\paragraph{Preserving the Cuts in \(G\).}
To show that $\mathcal{H}$ approximately preserves the cuts in \(G\), we use an inductive argument: 
by \Cref{lma:jtreesample}, for any $1 \leq i \leq L$, the set $\mathcal{H}(H_i) \subseteq \mathcal{H}_{i+1}$ of $O(j_{i+1})$-trees approximately preserves the cuts of its denser counterpart $\widetilde C_{i}$.
By \Cref{thm:fullydynamicsparsifier},  any sparsified core $\widetilde C_{i-1}$ itself approximately preserves the cuts in their denser counterpart $C_{i-1}$.

\bigbreak

To achieve these guarantees, 
  our data structure incorporates some elements from the data structure of \cite{Chen:2020aa},
which maintains a single-level hierarchy of \atmost{j}-trees (i.e., when \(L = 1\)).
In the next section, we explain the single-level case, which will be used as a subroutine in our hierarchy.

\subsection{Warm-Up: Single-Level Scheme} \label{subsec:single-level}


In this section, we summarize how to employ \Cref{thm:jtree} to maintain a set \(\mathcal H\) of $O(\log n)$ many \atmost{j}-trees in a single-level hierarchy (i.e., we prove \Cref{th:main} for the case \(L=1\)). This single-level scheme will later be used as a subroutine in our hierarchy, as detailed in \Cref{subsec:hierarchy}. 

For clarity, we begin by re-stating \Cref{thm:jtree}. 

\MadryJtreeCollection*

Now let $S$ be an arbitrary cut in $G$ and let $\cH' = \setof{H_i}_{i=1}^k$ be the collection of $O(j)$-trees obtained from \Cref{thm:jtree}.
Then, by \Cref{lma:cutpreserve}, we have
\begin{align}\label{eq:bound-for-sample}
U_G(S) \leq k^{-1} \sum_{i=1}^k U_{H_i}(S) \leq \alpha \cdot U_G(S).
\end{align}
In particular, this implies that there must exist a graph $H_i$ in the collection such that $U_{H_i}(S) \leq \alpha U_G(S)$. 
Consequently, the set $\cH'$ \(\alpha\)-preserves \emph{all} cuts of $G$ \emph{deterministically}, where \(\alpha = \atmosttilde{\log n}\). 

Unfortunately, the collection $\cH'$ has size $|\cH'| = \Omega(m/j)$, and thus is too large for our recursive application in \Cref{sec:hierarchy}.
As a solution to this issue, we sample $O(\log n)$ many graphs from $\cH'$ uniformly at random to initialize the set $\cH$.
This drastically reduces the size of the set, but at the cost of preserving the cuts with high probability. 

\begin{lemma}\label{lma:jtreesample}
    Let $H_1, \dots, H_k$ be the set \(\mathcal H'\) of $O(j)$-trees from \Cref{thm:jtree}, and let $(S, V \setminus S)$ be a cut in $G$. 
Then, if we sample $l = O(\log n)$ many $O(j)$-trees $H_{i_1}, \dots, H_{i_l}$ uniformly at random,
\begin{itemize}
\item $U_G(S) \leq \min_{1 \leq j \leq l} U_{H_{i_j}}(S)$, and
\item with probability at least $1-1/n^c$, $\min_{1 \leq j \leq l} U_{H_{i_j}}(S) \leq 2 \alpha \cdot U_G(S)$.
\end{itemize}
\end{lemma}
\begin{proof}
By using Markov's inequality and \Cref{eq:bound-for-sample}, for any randomly sampled $O(j)$-tree $H_{i_j}$, 
\[\Pr\left[U_{H_{i_j}}(S) \geq 2 \alpha U_G(S)\right] \leq 1/2.\] 
Hence, sampling $l = O(c \log n)$ many $O(j)$-trees uniformly at random suffices.
\end{proof}
Consequently, as we are satisfied with a Monte-Carlo algorithm that works again an \emph{oblivious adversary}, by first sampling $O(\log n)$ many $O(j)$-trees, we can perform all subsequent computations on a small number of $O(j)$-trees instead of the full sequence. 
The guarantees of our single-level algorithm are summarized in the lemma below, followed by a description of the algorithm and the proof of the theorem. 

\begin{lemma}\label{thm:2level}\label{th:single-level}
Given an integer \(j \geq 1\), and an $n$-vertex graph $G = (V, E, u)$ on initially $m$ edges, there is a data structure that maintains a collection \(\mathcal H\) of $O(j)$-trees that $\Otil(\log n)$-preserves the cuts of $G$ with probability at least $1-1/n^c$, for any constant $c > 0$ specified before the algorithm starts.

The data structure guarantees that at all times $|\cH| = O(\log n)$, and that the capacity ratio of each $O(j)$-tree $H \in \cH$ is $O(m U)$, where $U = \poly{n}$ is a bound on the capacity ratio of $G$.
\end{lemma}

\paragraph{The Algorithm.}
We employ the algorithm of \Cref{thm:jtree} with parameter $j$ to initialize and maintain, over $O(j)$ updates, a set $\mathcal H'$ of $O(j)$-trees 
that \atmosttilde{\log n}-preserves the cuts in \(G\).
The set \(\mathcal H\) is obtained by sub-sampling $O(\log n)$ many $O(j)$-trees from \(\mathcal H'\) uniformly at random.  After $O(j)$ updates in \(G\), we re-start the algorithm from scratch in order to ensure that, at all times, the collection $\cH'$ consists of $O(j)$-trees (i.e., has cores consisting of $O(j)$ vertices).

\begin{proof}[Proof of \Cref{thm:2level}]
We prove the bound on the update time and the high probability claim.

\underline{Update time:}
directly follows from \Cref{thm:jtree} as maintaining the graphs $H_i$ over $O(j)$ updates is done in $\Otil(m^2 /j)$ total time, and thus in amortized time $\Otil(m^2 /j^2)$.

\underline{High probability claim:}
for any fixed cut $(S, V \setminus S)$ and at any time, by \Cref{lma:jtreesample}, we have
\[\min_{H \in \mathcal H} U_{H}(S) \leq 2 \alpha \cdot U_G(S)\]
with probability at least \(1 - 1/n^c\), and \(\alpha = \atmosttilde{\log n}\). 
This means that the set \(\mathcal H\) is a collection of $O(j)$-trees that $\Otil(\log n)$-preserves the cuts of $G$ with probability at least \(1 - 1/n^c\).  
Moreover, by \Cref{thm:sparsifiedcore}, the graph $\Tilde{C}_i$ remains a \(2\)-approximate sparsified core of \(C_i\) with probability at least \(1 - 1/n^c\). 
%
%
\end{proof}

\subsection{The Data Structure of \Cref{th:main}} \label{subsec:data-structure}

The main motivation behind our data structure of \Cref{th:main}
  is to improve the update time of \Cref{thm:2level}, albeit at the cost of a larger approximation. 
We achieve this by maintaining a set \(\mathcal H\) of \atmost{j}-trees
  using an \((L,j)\)-hierarchy \(\mathcal H_{0}, \dots, \mathcal H_L\), as defined in \Cref{def:hierarchy}.
Note that if we set \(L=1\) in \Cref{th:main}, 
  we obtain the same guarantees as in \Cref{thm:2level}.


The data structure is summarized in \Cref{alg:j-tree-hierarchy}. 
We proceed by explaining the initialization and maintenance in \Cref{subsubsec:intialization,subsubsec:updates}, respectively.

\begin{algorithm2e}[]
\KwIn{An \(n\)-vertex graph \(G = (V, E, u)\), number of levels \(L\), and positive integer \(j\)}
\Maintain{a set \(\mathcal H\) of \atmost{j}-trees of \(G\) with \(|\mathcal H| = \atmost{\log^L n}\)}
\Procedure{Initialize}{
    \(C_{0} \gets G\) and \(F_{0} \gets (V, \emptyset, \vect{0})\) \;
    \(\mathcal H_{0} \gets \{C_{0} \cup F_{0}\} \) \;
    \ForEach{\(0 \leq i \leq L\)}{
        \(j_i \gets n \left( j/n \right)^{i/L}\)
        \tcc*[f]{note that \(j_0 = n\) and \(j_{L} = j\)}
    }
    initialize \(\widetilde C_{0}\) on \(C_{0} = G\) using cut sparsifier of \Cref{thm:fullydynamicsparsifier} \;
    \textsc{Compute-Hierarchy(\(L\))}
    \tcc*[f]{initialize \(H_L,\dots,H_1\)}
}

\Procedure{Update(\(U\))}{
    \tcc*[f]{\(u \in U\) is either edge insertion or deletion} \;
    pass \(U\) to data structure maintaining \(\widetilde C_{0}\) \;
    \ForEach{graph \(H_1 \in \mathcal H_1\)}{
        \(U(H_1) \gets\) changes in \(\widetilde C_{0}\)
        }
    
    \For{\(i = 1\) to \(L\)}{
        \ForEach{graph \(H_i \in \mathcal H_i\) with core \(C_i\)}{
            pass \(U(H_i)\) to data structure maintaining \(H_i\)\;
            \tcc*[f]{update \atmost{j_i}-tree \(H_{i} = C_{i} \cup F_{i}\) using \Cref{thm:2level}}
            
            pass changes in \(C_i\) to data structure maintaining \(\widetilde C_{i}\) \;
            \tcc*[f]{update underlying cut sparsifier \(\widetilde C_{i}\) using \Cref{thm:fullydynamicsparsifier}}
            
            \ForEach{graph \(H_{i+1} \in \mathcal H_{i+1}(\widetilde C_i)\)}{
                \(U(H_{i+1}) \gets\) changes in \(\widetilde C_{i}\)
                }
            }
        \If{there is \(H_i \in \mathcal H_i\) with core size \(|V(C_i)| > cj_i\)}{
            \tcc*[f]{\(c\) is a sufficiently large constant}

            \textsc{Compute-Hierarchy(\(i\))}\; 
            \Break
            }
        }
}

\Procedure(\tcc*[f]{compute \(\mathcal H_{i},\dots,\mathcal H_L\)}){Compute-Hierarchy(\(i\))}{
    \For(){\(k = i\) to \(L\)}{
        \(\mathcal H_k \gets \emptyset\) \;
        \ForEach{\(H_{k-1} \in \mathcal H_{k-1}\) with sparsified core \(\widetilde C_{k-1}\)}{
            compute \(\mathcal H_k(\widetilde C_{k-1})\) as \atmost{j_i}-trees of \(\widetilde C_{k-1}\) using \Cref{thm:2level} \; 
            add \(\mathcal H_k(\widetilde C_{k-1})\) to \(\mathcal H_k\) \;
        }
        \ForEach{\(H_{k} \in \mathcal H_{k}\) with core \(C_k\)}{
            (re-)initialize \(\widetilde C_{k}\) using cut sparsifier of \Cref{thm:fullydynamicsparsifier}\;
            \(U(H_k) \gets \emptyset\)
        }
    }
    
}
\caption{Dynamic-J-Tree(\(G, L, j\))}
\label{alg:j-tree-hierarchy}
\end{algorithm2e}

\subsubsection{Initialization} \label{subsubsec:intialization}
As mentioned in \Cref{def:hierarchy}, each level \(i\) of the data structure
  initializes the set \(\mathcal H_i\), where each graph \(H_i \in \mathcal H_i\) has a core \(C_i\) of size \atmost{j_i}.
The decreasing nature of the series \(j_0>\dots>j_L\)
  ensures that each level has cores smaller in size than those in the previous level.
We begin by explaining our choices for every \(j_i\).

For each level \(i\),
  we choose the integer \(j_i\) such that
  the cores in \(\mathcal H_i\) contain a fraction
  \(1/k = (j/n)^{1/L}\) of the vertices 
  from the cores in \(\mathcal H_{i-1}\).
i.e., starting with \(j_{0} = n\) vertices at level \(0\), the cores in 
  \(\mathcal H_1\) contains \atmost{n/k} vertices, the cores in
  \(\mathcal H_{2}\) contains \atmost{n/k^2} vertices,
  and so on,
  until the cores in \(\mathcal H_L\) contain \( n /k^L =  n (j/n)^{L/L} = j\) vertices.
Therefore, we set
\begin{equation} \label{eq:js}
j_i = n/k^i = n \left( j/n \right)^{i/L}
\end{equation}
for every \(0 \leq i \leq L\).

The initialization of the data structure is as follows.
It starts with \(\mathcal H_0 = \{ H_0 \}\), where \(H_0\) is the trivial \(n\)-tree of \(G\) with core \(C_{0} = G\) and the empty forest \(F_0 = (V, \emptyset, 0)\).
It then initializes the data structure of \Cref{thm:fullydynamicsparsifier} on \(C_{0}\)
  to maintain the sparsified graph \(\widetilde C_{0}\).

To initialize \(\mathcal H_1\),
  the data structure employs \Cref{thm:2level} on the sparsified core \(\widetilde C_0\) to maintain the set \(\mathcal H(H_{0})\) of \atmost{j_1}-trees on the vertex set of \(V(\widetilde C_0)\).
Since \(\mathcal H_0\) only consists of \(H_0\),
the \atmost{j_1}-trees of \(\mathcal H_1\) are defined as \(\mathcal H_1 = \mathcal H(H_{0}) \).

To obtain the graphs for the next level, the data structure repeats the procedure until it reaches level \(L\).
e.g., level \(2\) is initialized on the cores of the \atmost{j_1}-trees in \(\mathcal H_1\):
 each \(H_1 \in \mathcal H_1\) consists of a core \(C_{1}\),
  on which the data structure 
  initializes \Cref{thm:fullydynamicsparsifier} to maintain the sparsified core \(\widetilde C_1\).
It then initializes \Cref{thm:2level} on \(\widetilde C_1\) to 
  maintain the set \(\mathcal H(H_1)\) of \atmost{j_1}-trees on the vertex set of \(\widetilde C_1\).
The \atmost{j_2}-trees of \(\mathcal H_2\) is then defined as \(\mathcal H_2 = \bigcup_{H_1 \in \mathcal H_1} \mathcal H(H_1) \).


\Cref{fig:summary-initialization} provides a summary of the initialization for levels \(i\) and \(i+1\), where \(1 \leq i \leq L-1\).

\begin{figure}[th]
\begin{center}
\input{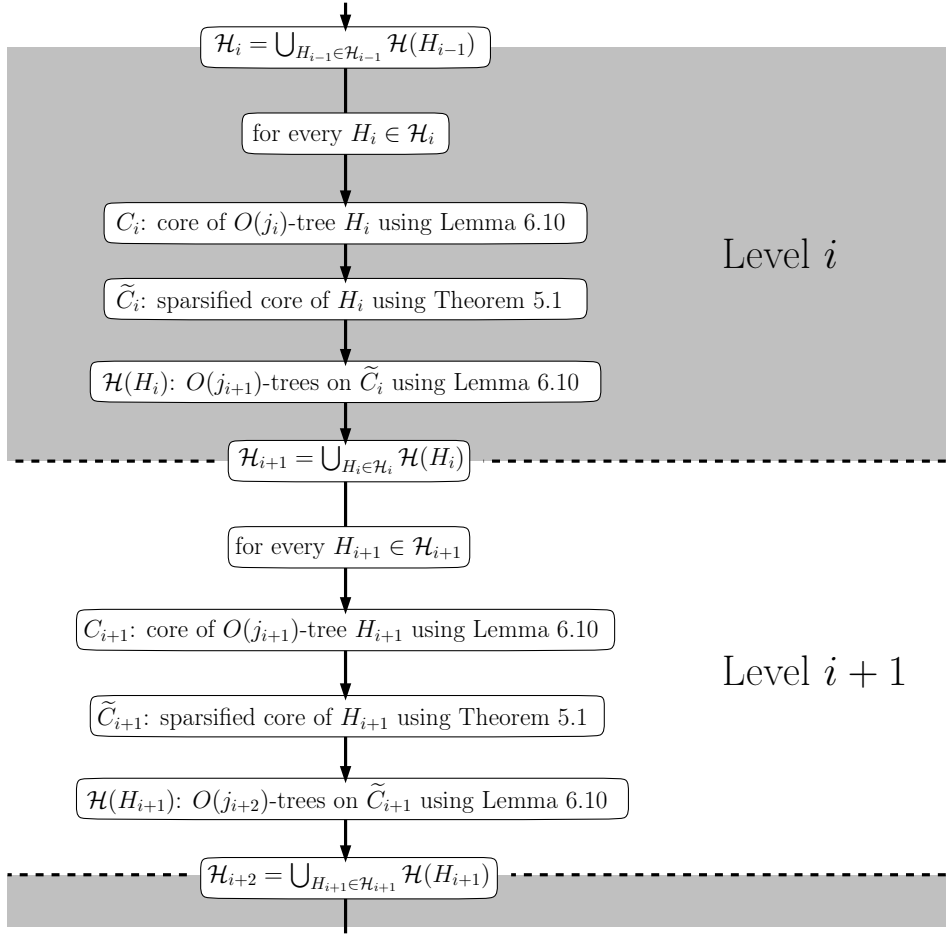}
\end{center}
\caption{Summary of initialization.}
\label{fig:summary-initialization}
\end{figure}

\subsubsection{Handling the Updates} \label{subsubsec:updates}
The data structure can handle a batch of updates \(U\) 
  where each update \(u \in U\) consists of the insertion or deletion
  of an edge into \(G\).
For simplicity, we assume that \(U\) is a singleton,
  as this can naturally be extended to the general case. 

To dynamically maintain the hierarchy \(\mathcal H_{0}, \dots, \mathcal H_L\), 
  we use a batching technique (explained below) to guarantee the size of each core in \(\mathcal H_i\) stays \atmost{j_i} at any time.
In a nutshell, the update procedure follows a top-down approach as follows. 
For any chain \(H_0,\dots,H_L\) in the hierarchy and every \(0 \leq i \leq L-1\),
  the data structure obtains a set \(U(H_{i})\), which it passes to \(H_{i}\). 
The set \(U(H_{i})\) consists of the updates that must to be passed to \(H_i\) as a result of the update \(U\) in \(G\).
While \(H_{i}\) is undergoing the updates in \(U(H_{i})\),
  the data structure obtains a set \(U(H_{i+1})\).
The set \(U(H_{i+1})\) is then passed to \(H_{i+1}\), and so on.

Using this terminology,
for any chain \(H_0,\dots,H_L\) in the hierarchy, the update starts from \(H_0\) (i.e., \(G\)) by setting \(U(H_0) = U\).
The update \(U(H_0)\) is passed to the data structure of \Cref{thm:fullydynamicsparsifier} maintaining the sparsified core \(\widetilde C_0\).
Since every graph \(H_1 \in \mathcal H_1\) is defined on \(\widetilde C_0\), the algorithm sets \(U(H_1)\) to be the changes in \(\widetilde C_0\) due to the update.
For every \(H_1 \in \mathcal H_1\), the data structure passes \(U(H_1)\) to the data structure  of \Cref{thm:2level} maintaining \(H_1\)\footnote{As the cores $\Tilde{C}_i$ grow over time, the data structure of \Cref{thm:2level} also needs to be able to handle isolated vertex insertions. However, we can straightforwardly handle them by including a set of $O(j_{i+1})$ dummy vertices upon initialization and re-naming them on demand.}.
This results in some changes in the core \(C_1\) of \(H_1\), 
which is passed to the data structure of \Cref{thm:fullydynamicsparsifier} maintaining the sparsified core \(\widetilde C_1\).
The changes in \(\widetilde C_1\) is then set as \(U(H_2)\) for every \(H_2 \in \mathcal H(H_1)\).
This procedure continues until we update \(H_L\), \textit{or} the size of a core at some level \(i\) has doubled since its initialization (i.e., for a fixed sufficiently large constant \(c\), once we 
have 
\(|V(C_i)| > c j_i \) 
after the update).
In the latter case, the data structure stops passing the updates to higher levels, and instead, re-initializes the sets \(\mathcal H_i, \dots, \mathcal H_L\) of the hierarchy from level \(i\) to \(L\).


\Cref{fig:summary-update} provides a summary of the update for levels \(i\) and \(i+1\), where \(1 \leq i \leq L-1\).

\begin{figure}[th]
\begin{center}
\input{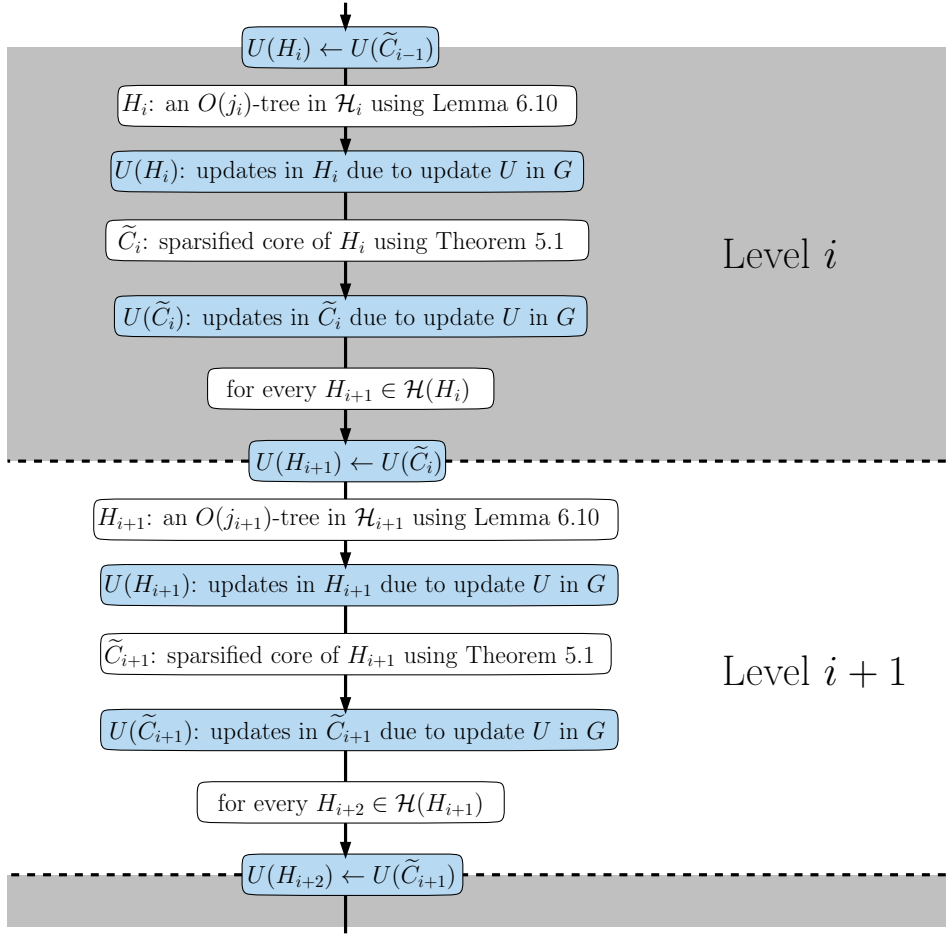}
\end{center}
\caption{Summary of update.}
\label{fig:summary-update}
\end{figure}

%

\subsection{Proof of \Cref{th:main}} \label{subsec:proof-main-thm}
In this section, we finalize the proof of \Cref{th:main}.
We begin by proving some auxiliary claims in \Cref{subsub:auxiliary}, which will be used in the update time analysis in \Cref{subsub:update-time} and in the proof of \Cref{th:main} in \Cref{subsub:proof}.

\subsubsection{Auxiliary Claims} \label{subsub:auxiliary}
Given the \((L,j)\)-hierarchy \(\mathcal H_0, \dots, \mathcal H_L\) from \Cref{alg:j-tree-hierarchy}, we first obtain an upper bound on the size of each \(\mathcal H_i\).
We then proceed to bound the total recourse in every \atmost{j_i}-tree \(H_i \in \mathcal H_i\), which will be used to obtain the total recourse of \(\mathcal H_i\).
Along the way, we will see how the assumption \(U = \poly{n}\) in \Cref{th:main} plays a role in the analysis.

\begin{claim} \label{clm:size-Hi}
For every \(0 \leq i \leq L\), \(|\mathcal H_i| = \log ^{\atmost{i}} n\). 
\end{claim}
\begin{proof}
We prove by induction on \(i\). 
\begin{itemize}
\item \underline{Base case:}
for $i=0$, this is trivial since $\cH_0 = \setof{G}$ by \Cref{def:hierarchy}.

\item \underline{Induction step:}
suppose that \(|\mathcal H_i| = \log ^{\atmost{i}} n\).
\Cref{thm:2level} guarantees that for every $H_i \in \cH_i$, we have that $|\cH(H_i)| = O(\log n)$. Hence, we have \[|\cH_{i+1}| \leq O(\log n) |\cH_i| = \log ^{\atmost{i+1}} n.\]
\end{itemize}
\end{proof}

\begin{claim} \label{clm:recourse-chain}
Let \(H_0,\dots,H_L\) be an arbitrary chain in the hierarchy and
let \rec be an upper bound on the amortized recourse of the sparsified core \(\widetilde C_i\) after a single update in \(H_i\).
Then, following \(l = \Omega(j_1)\) total updates in \(G\), if \(H_i\) has not been re-initialized by \Cref{alg:j-tree-hierarchy},  
the update in \(G\) results in \atmost{l \rec ^{i+1}} total recourse in \(\widetilde C_i\). 

\end{claim}
\begin{proof}
Recall that, after an update in \(G\),
  the data structure passes the updates from one level to the next,
  until the size of a core \(C_k\) of an \atmost{j_k}-tree \(H_k \in \mathcal H_k\) exceeds the value \(c j_k\), where \(c\) is a constant specified by the algorithm before the procedure begins.
Therefore, if \(H_i \in \mathcal H_i\) is not re-initialized after the update, then neither are \(H_0, \dots, H_{i-1}\).

We prove by induction on \(i\).
\begin{itemize}
\item \underline{Base case:}
since \(\mathcal H_0 = \{ G \}\) by \Cref{def:hierarchy}, \(l\) updates in \(C_0 = G\) results in \(l\) total recourse in \(C_0\), and so by \Cref{thm:fullydynamicsparsifier}, \atmost{l \rec} is serves as an upper bound on the total recourse in \(\widetilde C_0\). 

\item \underline{Induction step:}
suppose that the total recourse in \(\widetilde C_{i}\) is \atmost{l \rec^{i+1}}.
Since \(H_{i+1}\) is defined on the vertex set of \(\widetilde C_{i}\), the total number of updates passed to \(H_{i+1}\) cannot exceed the total recourse in \(\widetilde C_{i}\). 
Thus, \(\atmost{l \rec^{i} \rec} = \atmost{l \rec^{i+1}}\) serves as an upper bound on the total recourse in \(\widetilde C_{i+1}\).

\end{itemize}
\end{proof}

\begin{claim} \label{clm:recourse-hierarchy}
Given an \(n\)-vertex graph \(G=(V,E,u)\) with an upper bound \(U = \poly{n}\) on its capacity ratio. 
Then, following \(l = \Omega(j_1)\) total updates in \(G\), if \(\mathcal H_i\) has not been re-initialized by \Cref{alg:j-tree-hierarchy},  
the update in \(G\) results in \(l \log ^{\atmost{i}} n\) total recourse in the sparsified cores of \(\mathcal H_i\).
\end{claim}
\begin{proof}
From \Cref{clm:recourse-chain}, each sparsified core \(\widetilde  C_i\) of \(H_i \in \mathcal H_i\) undergoes \atmost{l \rec ^{i+1}} recourse.
By \Cref{clm:size-Hi}, \(|\mathcal H_i| = \log ^{\atmost{i}}n\).
Thus, the total recourse of the sparsified cores in \(\mathcal H_i\) is
\[
\atmost{ \sum_{H_i \in \mathcal H_i} l \rec ^{i+1}} = l \rec ^{i+1} \log ^{\atmost{i}}n = l \log ^{\atmost{i}} n,
\]
where the last equality follows from the fact that \(\rec = \poly{\log n} \log U = \log ^{\atmost{1}} n\) by \Cref{thm:sparsifiedcore} and the assumption \(U = \poly{n}\).
\end{proof}

\begin{claim} \label{clm:recomputation-number}
Given an \(n\)-vertex graph \(G=(V,E,u)\) with an upper bound \(U = \poly{n}\) on its capacity ratio. 
Then, following \(l = \Omega(j_1)\) total updates in \(G\), the number of times the set \(\mathcal H_i\) has been re-initialized by \Cref{alg:j-tree-hierarchy} is \atmost{l \log^{\atmost{i}} n/j_i}. 
\end{claim}
\begin{proof}
Recall from \Cref{sec:madry} that each update in a \(j\)-tree could result in the moving of at most $O(1)$ vertices into its core. 
Thus, an \atmost{j_i}-tree in \Cref{alg:j-tree-hierarchy} is re-initialized if either (1) the size of its core exceeds \(cj_i\) or (2) there exists an \atmost{j_k}-tree in \(\mathcal H_k\), \(k>i\),  whose core size exceeds \(cj_k\).

We prove by induction on \(i\).
Let \(N_i\) denote the number of times \(\mathcal H_i\) has been re-initialized.

\begin{itemize}
\item \underline{Base case:}
since \(\mathcal H_0 = \{ G \}\) by \Cref{def:hierarchy}, the set \(\mathcal H_0\) is re-initialized after \atmost{j_0} updates in \(G\), which means that \(N_0 = \atmost{l/j_0}\). 

\item \underline{Induction step:}
suppose that \(\mathcal H_i\) has been re-initialized  \atmost{l \log^{\atmost{i}} n/j_i} times.
By item (2) above, each such re-initialization also triggers re-initialization of \(\mathcal H_{i+1}\).
Therefore, we only need to bound the number of re-initializations due to item (1). i.e., we find an upper bound \(k\) on the number of times \(\mathcal H_{i+1}\) has been re-initialized while \(\mathcal H_{i}\) has not.

To this end, let \(l_i\) be the number of updates in \(G\) that result in the \textit{first} re-initialization of \(\mathcal H_{i}\).
By \Cref{clm:recourse-hierarchy}, this happens when \(l_i \log ^{\atmost{i}}n = cj_{i}\) as the number of recourses in \(\widetilde C_i\) upper bounds the number of vertices moved to \(C_i\) due to the updates.
Using a straightforward computation, we have \(k = \atmost{l_i/l_{i+1}} = j_i\log ^{\atmost{1}}n / j_{i+1}\).
Therefore, we have
\[
N_{i+1} = \atmost{N_i + kN_i} = \frac{j_i}{j_{i+1}}\log ^{\atmost{1}}n  \left( \frac{l}{j_i} \log^{\atmost{i}} n \right) = \frac{l}{j_{i+1}} \log^{\atmost{i+1}} n. 
\]
\end{itemize}
\end{proof}

\subsubsection{Total Update Time} \label{subsub:update-time}
Given an \((L,j)\)-hierarchy \(\mathcal H_0, \dots, \mathcal H_L\) in \Cref{alg:j-tree-hierarchy}, we first analyze the total update time of an \atmost{j_i}-tree \(H_i \in \mathcal H_i\), and then use this to obtain the total update time of \(\mathcal H_i\).
Along the way, we will see how our choice of the values for \(j_i\)'s results in the desired update times.
Specifically, we will use the crucial fact that  \(j_{i-1}/j_i = (n/j)^{1/L} \), as given by \Cref{eq:js}.

\begin{lemma} \label{lem:total-time-single}
Given an \(n\)-vertex graph \(G=(V,E,u)\) with an upper bound \(U = \poly{n}\) on its capacity ratio.
Then, following \(l = \Omega(j_1)\) total updates in \(G\), the total update time for maintaining an \atmost{j_i}-tree \(H_i \in \mathcal H_i\) and its sparsified core \(\widetilde C_i\) is
\[
l \left( \frac{j_{i-1}^2}{j_i^2} + \frac{j_{i-1}}{j_i} \right) \log^{\atmost{i}}n
\]
for every \(1 \leq i \leq L\).
\end{lemma}
\begin{proof}

We discuss the total time for the \atmost{j_i}-tree \(H_i \in \mathcal H_i\) and its sparsified core \(\widetilde C_i\) separately. 

\begin{itemize}
\item \underline{For \(H_i\):}
by \Cref{th:single-level}, the total time spent to initialize and maintain \(H_i\) before its re-initialization is \atmosttilde{l_i^2/j_i}, where $l_i$ is the initial number of edges of $\Tilde{C}_{i-1}(H_{i-1})$, the sparsified core of its parent $H_{i-1} \in \cH_{i-1}$. As \(\widetilde C_{i-1}\) contains only $\atmosttilde{j_{i-1} \log U}$ many edges, it follows that $l_i = \atmosttilde{j_{i-1} \log U}$, and so the total time is $\atmosttilde{ j_{i-1}^2/j_i \log^2 U }$. 

\item \underline{For \(\widetilde C_i(H_i)\):}
by \Cref{thm:sparsifiedcore}, the total time spent to initialize and maintain a sparsified core \(\widetilde C_i(H_i)\) over $O(j_i)$ updates is $\Otil(l_i \log U)$, where $l_i$ is the initial number of edges in $H_i$. 
Since \(H_i\) is built on \(\widetilde C_{i-1}\), \(l_i \leq |E(\widetilde C_{i-1})| = \atmosttilde{j_{i-1} \log U}\).
Thus, the total time spent to maintain \(\widetilde C_i\) before its re-initialization is $\atmosttilde{j_{i-1} \log^2 U}$. 
\end{itemize}
Therefore, the total time 
for maintaining \(\widetilde C_i\) and \(H_i\) before re-initialization is 
\[
\atmosttilde{
\frac{j_{i-1}^2}{j_i} \log^2 U + j_{i-1} \log^2 U}.
\]

By \Cref{clm:recomputation-number}, \(H_i\) is re-initialized  \atmost{l \log ^{\atmost{i}} n /j_i} times.
Hence, the total time for maintaining \(H_i\) is
\begin{align*}
\frac{l \log ^{\atmost{i}} n}{j_i} \left( \frac{j_{i-1}^2}{j_i} + j_{i-1} \right)\log^2 U
=& l \left( \frac{j_{i-1}^2}{j_i^2} + \frac{j_{i-1}}{j_i} \right) \log^{\atmost{i}}n. 
\end{align*}
\end{proof}

\begin{lemma} \label{lem:total-time_Hi}
Given an \(n\)-vertex graph \(G=(V,E,u)\) with an upper bound \(U = \poly{n}\) on its capacity ratio.
Then, following \(l = \Omega(j_1)\) total updates in \(G\), the total update time for maintaining \(\mathcal H_i\) is
\[l \left( \frac{n}{j} \right)^{2/L}\log^{\atmost{i}} n\]
for every \(1 \leq i \leq L\).
\end{lemma}
\begin{proof}
Combining the total update time for each \(H_i \in \mathcal H_i\) from \Cref{lem:total-time-single} with the size of \(\mathcal H_i\) from \Cref{clm:size-Hi}, 
we conclude the total time for \(\mathcal H_i\) is
\begin{align*}
\left( l \left( \frac{j_{i-1}^2}{j_i^2} + \frac{j_{i-1}}{j_i} \right) \log^{\atmost{i}}n \right) \log ^{\atmost{i}} n
&= 
l \left( \frac{j_{i-1}^2}{j_i^2} \right) \log^{\atmost{i}}n
\\
&= l \left( \frac{n}{j} \right)^{2/L} \log ^{\atmost{i}} n,
\end{align*}
since \(j_{i-1}/j_i = (n/j)^{1/L} \) by \Cref{eq:js}.
\end{proof}


In the remainder of this section,
we prove \Cref{th:main}.


\subsubsection{Proof of \Cref{th:main}.} \label{subsub:proof}
We conclude the section by proving \Cref{th:main}.

The fact that the algorithm correctly maintains an $(L, j)$-hierarchy with high probability follows from the correctness of \Cref{th:single-level}, by (1) union bounding the failure probabilities of all data structures of \Cref{thm:fullydynamicsparsifier}, and (2) noting that the graphs in $\cH_L$ at all times constitute $O(j_i)$-trees by the fact that \Cref{alg:j-tree-hierarchy} re-initializes \(\mathcal H_i, \dots, \mathcal H_L\) once there is an \atmost{j_i}-tree \(H_i \in \mathcal H_i\) with core size \(|V(C_i)|>cj_i\), where \(c\) is a constant specified by the algorithm before the procedure begins.
The correctness of \(\mathcal H\) then follows from \Cref{lma:chaintree}.

It remains to prove (1) \(\mathcal H\) is set that $\Otil(\log)^L$-preserves the cuts in \(G\) with high probability, (2) the bound on \(|\mathcal H|\),
(3) \(\mathcal H\) is explicitly maintained,
(4) the capacity ratio,
and (5) the update time.

\underline{$\Otil(\log)^L$-preserving the cuts in \(G\) with high probability:}
by definition, we  need to show that at any time and for any cut $(S, V \setminus S)$ in $G$, with high probability, there exists a graph \(H \in \mathcal H\) such that 
$U_{H}(S) \leq \Otil(\log n)^L U_G(S)$. 
 As every $H \in \cH$ corresponds to a unique chain \(H_0, \dots, H_L\) (i.e., \(H = H^{H_0, \dots, H_L}\)), it suffices to show by induction on $i$ that, with high probability, there exists a chain such that $U_{H^{H_0, \dots, H_i}}(S) \leq \Otil(\log n)^i U_G(S)$. 

\begin{itemize}
\item \underline{Base case:}
by \Cref{th:single-level}, with high probability, there exists a graph $H_1 \in \cH_1$ such that $U_G(S) \leq U_{H_1}(S) \leq \Otil(\log n) \cdot U_G(S)$.
The base step follows as \(H^{H_0,H_1} = \Tilde{C}_1 \cup F_1\) by \Cref{eq:chain-graph}, and that \(\Tilde{C}_1 \cup F_1\) is an \atmost{j_1}-tree of \(H_1\) using the proof of \Cref{lma:chaintree}.

\item \underline{Induction step:}
Suppose that with high probability there exists $H_0, H_1, \dots, H_i$ such that $U_{H^{H_0, \dots, H_i}}(S) \leq \Otil(\log n)^i U_G(S)$. 
By \Cref{th:single-level}, with high probability, there exists a graph $H_{i+1} \in \cH(H_i)$, with core \(C_{i+1}\) and forest \(F_{i+1}\), that $\Otil(\log n)$-preserves the cut $S$ in \(V(\widetilde C_i)\). 
Since $\Tilde{C}_{i+1}$ is a $2$-approximate cut-sparsifier of $C_{i+1}$, it follows that, with high probability, the graph $\Tilde{C}_{i+1} \cup F_{i+1} \cup \dots \cup F_1$ $\Otil(\log n)^{i+1}$-preserves the cut $S$ in \(G\). 
\end{itemize}

\underline{Size of \(\mathcal H\):} 
since there is a bijection between graphs in \(\mathcal H_L\) and chains in the hierarchy, $|\cH| = |\cH_L|$ as discussed in \Cref{clm:size-Hi}.
The claim then follows immediately from \Cref{clm:size-Hi}.

\underline{Explicit maintenance of \(\mathcal H\):}
for every any chain \(H_0,\dots,H_L\) in the hierarchy, the forest \(F_i\) and the sparsified core \(\widetilde C_i\) is explicitly maintained by the data structure.
Thus, we can explicitly maintain \(H^{H_0, \dots, H_L} := \Tilde{C}_L \cup F_L \cup \dots \cup F_1\)
by updating the sets during the maintenance of \Cref{th:main}.
If some sets are re-initialized by the data structure, the set \(\mathcal H\) can be updated in \(L \log ^{\atmost{L}} n\) worst-case update time,  which does not add any overhead in the update time of the data structure.
This is due to the fact that \(|H| = \log ^{\atmost{L}} n\), and each chain is a union of \(L\) sets.

\underline{Capacity ratio:}
since \(\widetilde C_0\) is obtained from \(C_0 = G\), the capacity ratio of \(\widetilde C_0\) is \(\atmost{mU} = \atmost{n^2 U}\) by \Cref{thm:fullydynamicsparsifier}.
For the rest, it simply follows by induction and the fact that, 
  at each level \(1 \leq i \leq L\),
  the data structure adds \atmost{n^2} overhead to the capacity ratio, and that there is \atmost{n} additional overhead due to sparsifying the core \(C_L\):
  (1) by assumption, each \(C_i\) contains \(\atmosttilde{n} = n \log ^{\atmost{1}} n\) edges,
  (2) by \Cref{thm:fullydynamicsparsifier}, the capacity ratio of the sparsified graph \(\widetilde C_i\) is \atmost{n} times that of the core graph \(C_i\),
  and (3) by \Cref{thm:2level}, the capacity ratio of the \atmost{j_{i+1}}-tree \(H_{i+1} = C_{i+1} \cup F_{i+1}\) of \(\widetilde C_i\)
  is \atmost{n^2} times that of \(\widetilde C_i\).

\underline{Update time:} 
using the total update time of \(\mathcal H_i\) from \Cref{lem:total-time_Hi}, the total update time for maintaining the \((L,j)\)-hierarchy after \(m\) updates in \(G\) is
\[
\sum_{i=0}^L m \left( \frac{n}{j} \right)^{2/L}\log^{\atmost{i}} n = L m \left( \frac{n}{j} \right)^{2/L}\log^{\atmost{L}} n = m \left( \frac{n}{j} \right)^{2/L}\log^{\atmost{L}} n,
\]
where the last equality follows from the trivial bound \(L = \atmost{\log n}\) from \Cref{th:main}.
Note that the upper bound on \(L\) in \Cref{th:main} ensures that the term \(\log ^{\atmost{L} }n\) does not dominate \((n/j)^{2/L}\) in the update time, as explained in \Cref{rmk:L}. 
%
%
The amortized update time of the data structure follows immediately
  as the computed time is the total update time after \(m\) updates.

\section{Applications}\label{sec:applications}
In this section, we explain how the data structure of \Cref{th:main} can be used to speed up some central problems in the dynamic setting.
Namely, we discuss the following problems.
\begin{itemize}
\item Dynamic all-pair min-cuts and dynamic all-pair maximum flow, discussed in \Cref{subsec:min-cut}.
\item Dynamic multi-way cut and multi-cut, discussed in \Cref{subsec:multi-cut}.
\item Dynamic sparsest cut, discussed in \Cref{subsec:sparsest-cut}.
\end{itemize}

The common approach is to
 leverage the structural properties of \atmost{j}-trees 
  together with the guarantees of \Cref{th:main}.
This is due to the fact that these problems are easier to solve on trees than general graphs.
As we can maintain an \atmost{j}-tree using \Cref{th:main},
we will argue how we can efficiently maintain the solutions on the forest of the \(j\)-tree, and upon a query, only solve the problem on the sparsified core with \atmosttilde{j} vertices and edges.
Each problem uses specific structural property of \(j\)-trees, and so we discuss them separately.

\subsection{Dynamic Approximate All-Pair Min-Cuts and All-Pair Maximum Flow} \label{subsec:min-cut}
Given a graph \(G=(V,E,u)\), a sink $s \in V$ and a source $t \in V$, a minimum \(s\)-\(t\) cut is a cut $(S, V \setminus S)$ such that $s \in S$, $t \in V \setminus S$ and $U_G(S)$ is minimum among all such cuts separating $s$ and $t$. 

In the fully dynamic all-pairs min-cut problem, an input graph \(G\) undergoes edge insertions and deletions, and the goal is to maintain query access to the \(s\)-\(t\) minimum cut value for any two vertices \(s,t \in V\). In this section, we show how the data structure from \Cref{th:main} can be used to maintain query access to $\Otil(\log n)^L$-approximations of the $s$-$t$-min-cut values with high probability. 

\begin{theorem}\label{thm:max-flow}
Given an undirected \(n\)-vertex graph \(G = (V, E, u)\) with polynomially bounded capacity ratio undergoing polynomially many edge insertions and deletions, parameters \(j \geq 1\) and \( 1 \leq L \leq \littleo{\sqrt{\log (n/j) / \log \log n}}\),
  there is a data structure that supports the following operations against an oblivious adversary:
  \begin{itemize}
      \item $\textsc{InsertEdge}(u, v, c)$: Inserts an edge $(u,v)$ with capacity $c$ to $G$ in $O((n/j)^{2/L} \log^{\atmost{L}} n)$ amortized time.
      \item $\textsc{DeleteEdge}(u, v)$: Removes the edge $(u,v)$ from $G$ in $O((n/j)^{2/L} \log^{\atmost{L}} n)$ amortized time.
      \item $\textsc{MinCut}(s, t)$: Returns an $\Otil(\log n)^L$-approximation to the $s$-$t$ mininimum cut value of $G$ in time $O(j \log^{O(L)}n)$.
  \end{itemize}
The algorithm is correct with high probability.
\end{theorem}
By setting $j = n^{2/(2+L)}$ to balance both query- and update time, and by pushing the query time in the update time through maintaining a solution during updates, we in particular get the following. 
\begin{corollary}
    For any fixed integer $L \geq 1$, there is a dynamic data structure that w.h.p against an oblivious adversary returns $\Otil(\log n)^L$-approximate $s$-$t$-min-cut values in $O(n^{\frac{2}{2+L}}\log^{O(L)} n)$ amortized update time and $O(1)$ query time.
\end{corollary}

\begin{corollary}
    There is a dynamic data structure that w.h.p against an oblivious adversary returns $O(\polylog (n))$-approximate $s$-$t$-min-cut values in $O(n^{0.001})$ amortized update time and $O(1)$ query time.
\end{corollary}
By letting $L$ grow as a function of $n$ and setting $L = \Theta(\sqrt{\log n} / \log \log n)$, we also get a datastructure with sub-polynomial update time.
\begin{corollary}
    There is a dynamic data structure that w.h.p against an oblivious adversary returns $e^{O(\log^{1/2} n \log \log n)}$-approximate $s$-$t$-min-cut values in $e^{O(\log^{1/2} n \log \log n)}$ amortized update time and $O(1)$ query time.
\end{corollary}
Here we remark that the data structure from Theorem $4.1$ in \cite{brand2024almostlineartimealgorithmsdecremental} also has sub-polynomial approximation and query/update time, but additionally works \emph{against an adaptive adversary}. On the other hand, their sub-polynomial factor of $e^{O(\log^{3/4} \log \log n)}$ is worse than that of our data structure. 

Before we prove the theorem, let us state the following straightforward but powerful structural observation.
\begin{observation}\label{lma:max-flow-struct}
    Let $H = C \cup F$ be a $j$-tree with forest $F$ and core $C$. Let $s, t \in V$ such that $s \not \in V(C)$, $t \in V(C)$. Let $e_{\min}^F(s)$ be the minimum capacity edge along the $s$-$\root_F(s)$ path, and let $\nu$ be the maximum flow value between $\root_F(s)$ and $t$. Then the $s$-$t$ maximum flow value is given by $\min(u_H(e_{\min}^F(s)), \nu)$.
\end{observation}

\begin{proof}[Proof of \Cref{thm:max-flow}]
    We use the data structure from \Cref{th:main}, which guarantees us that after any update to $G$ and for for any cut $S$, we have $U_G(S) \leq U_H(S)$ for all $H \in \cH$. Additionally, there with high probability exists an $H \in \cH$ such that $U_H(S) \leq \Otil(\log n)^L U_G(S)$. We can union bound the failure probability of all $O(n^2)$ $s$-$t$ min-cuts over all updates performed to $G$, and thus assume that at all times and for any $s$-$t$ min-cut $S$, there exists an $H$ such that $U_H(S) \leq \Otil(\log n)^L U_G(S)$. Hence, if we can compute a $2$-approximate $s$-$t$ min-cut value for all graphs $H \in \cH$, returning the minimum value found will yield the desired approximation. 

    It remains to show how to efficiently compute a $2$-approximation to the $s$-$t$ min-cut (equivalently the $s$-$t$ max-flow) values of the graphs in $\cH$. To do so, remember that a graph $H \in \cH$ corresponds uniquely to a chain $H_0, H_1, \ldots, H_L$, and that $H = \Tilde{C}_L \cup F_L \cup F_{L-1} \cup \ldots F_1$. If $s, t \in V(\Tilde{C}_L)$, then the $s$-$t$ maximum flow value in $H$ is equal to the $s$-$t$ maximum flow value in $\Tilde{C}_L$. To compute this value, we can use the nearly-linear $2$-approximate max flow algorithm from \cite{peng2015approximateundirectedmaximumflows} which runs in time $\Otil(|E(\Tilde{C}_L)|) = \Otil(j)$, where we used that $|E(\Tilde{C}_L)| = \Otil(j)$. 
    
    Suppose now that $s$ or $t$ or both are not in the core $\Tilde{C}_L$, but rather in the forest $F = F_L \cup \ldots \cup F_1$. Let $u_H(e_{\min}^F(s)), u_H(e_{\min}^F(t))$ be the capacities of the minimum capacity edges in $F$ along the $s$-$\root_F(s)$ resp. $t$-$\root_F(t)$ paths, and let $\nu$ be the $\root_F(s)$-$\root_F(t)$-max-flow value in $\Tilde{C}_L$. 
    
    Then by \Cref{lma:max-flow-struct}, the $s$-$t$-max-flow value of $H$ is given by $\min(\nu, u_H(e_{\min}^F(s)), u_H(e_{\min}^F(t)))$. We can now conclude by noting that we can find in time $O(L \log n)$ the minimum capacity edges $e_{\min}^F(s), e_{\min}^F(t)$ by maintaining a link-cut tree data structure from \Cref{thm:treeds} on each of the forests $F_i$, which is in fact already maintained by the algorithm, as can be seen seen in \Cref{algo:maintain}. Furthermore, we can compute a $2$-approximation of the $\root_F(s)$-$\root_F(t)$ maximum flow value by \cite{peng2015approximateundirectedmaximumflows}.
    
    Consequently, we can determine a $2$-approximate $s$-$t$-min-cut value of a graph $H \in \cH$ in time $\Otil(j)$. Hence, in time $O(j \log^{O(L)} n)$ we can compute them for all graphs in $\cH$, as the set satisfies $|\cH| = O(\log^L n)$.
\end{proof}

\begin{figure}[t]
\center
\includegraphics[width=.6\textwidth]{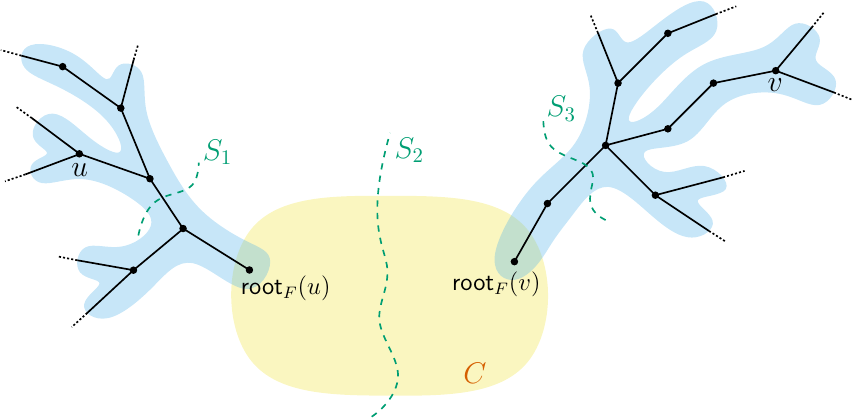}
\label{fig:application}
\caption{Graphical depiction of \Cref{lma:max-flow-struct} and how our algorithm finds an (approximate) $u$-$v$-min-cut of a $j$-tree. $S_1, S_3$ correspond to the minimum cuts separating $u,v$ respectively from their roots in $F$. Then, $S_2$ denotes the minium cut between $\root_F(u)$ and $\root_F(v)$ in the core $C$. The minimum $u$-$v$ cut corresponds to the minimum capacity cut of these three candidates.}
\end{figure}

\subsection{Dynamic Multi-Way Cut and Multi-cut} \label{subsec:multi-cut}

\begin{definition}
Let $G = (V, E, u_G)$ be a graph. Then we call a collection $\{S_i\}_{i=1}^k$ of subsets $S_i \subseteq V$ a partition of $V$ if $V = \bigcup_{i=1}^k S_i$ and, for all $i \neq j$ we have $S_i \cap S_j = \emptyset$.  Given a partition $\{S_i\}$, the total capacity of edges crossing two different components is defined as $U_G(\{S_i\}) := \frac{1}{2} \sum_{i = 1}^k U_G(S_i)$. 
\end{definition}

Given a graph $G = (V,E,u)$ and set of terminals $T \subseteq V$ a multi-way cut of $G$ with respect to $T$ is a partition $\{S_i\}$ of $V$ where each node of $T$ falls in different partitions. The goal of the min cost multi-way cut problem is to find a multi-way cut $\{S_i\}$ of $G$ with respect to $T$ minimizing $U_G(\{S_i\})$.

Given a set of terminal pairs $T = \{s_1,t_1 \dots\}$ where $t_i, s_i \in V$ the multi-cut of $G$ with respect to $T$ is a partition $\{S_i\}$ of $V$ separating each terminal pair in $T$. The goal of the min cost multi-cut problem is to find a multi-cut $\{S_i\}$ of $G$ with respect to terminal pairs $T$ minimizing $U_G(\{S_i\})$.

The following theorem shows how the data structure of \cref{th:main} can be used to be able to answer approximate multi-way and multi-cut queries on graphs undergoing edge insertions and deletions. Its proof is deferred to \cref{appendix:multicut}.

\begin{theorem}
\label{thm:multi}
Given an \(n\)-vertex graph \(G = (V, E, u)\) with polynomially bounded capacity ratio undergoing polynomially bounded edge insertions and deletions, parameters \(j \geq 1\) and \( 1 \leq L \leq \littleo{\sqrt{\log (n/j) / \log \log n}}\),
  there is a data structure that supports the following operations against an oblivious adversary:
  \begin{itemize}
      \item $\textsc{InsertEdge}(u, v, c)$: Inserts an edge $(u,v)$ with capacity $c$ to $G$ in $O((n/j)^{2/L} \log^{\atmost{L}} n)$ amortized time.
      \item $\textsc{DeleteEdge}(u, v)$: Removes the edge $(u,v)$ from $G$ in $O((n/j)^{2/L} \log^{\atmost{L}} n)$ amortized time.
      \item $\textsc{Multi-Way Cut}(T \subseteq V)$: Returns an $\Otil(\log n)^{L+1}$-approximation of the minimum cost multi-way cut value of $G$ with respect to terminal set $T$ in $\atmosttilde{k \cdot (n/j)^{2/L} \log^{\atmost{L}} n + j}$ amortized time.
      \item $\textsc{Multi-Cut}(T = \{(s_i,t_i)|i \in [k]\})$: Returns an $\Otil(\log n)^{L+2}$-approximation of the minimum cost multi-cut value of $G$ with respect to the set of terminal pairs $T$ in $\atmosttilde{k \cdot (n/j)^{2/L} \log^{\atmost{L}} n + (k+j)^2}$ amortized time.
  \end{itemize}
The algorithm is correct with high probability.
\end{theorem}

In specific, \cref{thm:multi} can be used to obtain the following corollaries:
\begin{corollary}
    For any fixed integer $L \geq 1$, there exists a dynamic data structure that w.h.p. and amortized update time $O(\log^{O(L)} n^{\frac{2}{1+L}})$ against an oblivious adversary can return an $\tilde{O}(\log n)^{L+2}$ approximation to the minimum multi-way cut and multi-cut value with respect to a terminal set of size $k$ in amortized query times $\tilde{O}(k \log^{O(L)} n^{\frac{2}{1+L}})$ and $\tilde{O}(k \log^{O(L)} (k+n^{\frac{2}{1+L}}))$ respectively.
\end{corollary}

\begin{corollary}
    There exists a dynamic data structure with amortized update time $O(n^{0.001})$ against an oblivious adversary can return an $O(\polylog(n))$ approximation to the minimum multi-way cut and multi-cut value with respect to a terminal set of size $k$ in amortized query times $\tilde{O}(kn^{0.001})$ and $\tilde{O}(kn^{0.001} + k^2)$ respectively w.h.p.
    
\end{corollary}
Similar to the previous section, by letting $L$ grow as a function of $n$ and setting $L = \Theta(\sqrt{\log n} / \log \log n)$, we also get a datastructure with sub-polynomial update time which answers queries in almost $\poly{k}$ time.
\begin{corollary}
    There exists a dynamic data structure with amortized update time $O(e^{O(\log^{1/2} n \log \log n)})$ against an oblivious adversary can return an $e^{O(\log^{1/2} n \log \log n)}$ approximation to the minimum multi-way cut and multi-cut value with respect to a terminal set of size $k$ in amortized query times $O(k \cdot e^{O(\log^{1/2} n \log \log n)}))$ and $\tilde{O}(ke^{O(\log^{1/2} n \log \log n)} + k^2)$ respectively w.h.p.
\end{corollary}


\subsection{Dynamic Sparsest Cut} \label{subsec:sparsest-cut}
Let \(G=(V,E,u)\) be a graph. The \textit{sparsity} of a cut $(S, V \setminus S)$ is defined as 
\[
\psi_G(S) = \frac{U_G(S)}{\min \{|S|, |V \setminus S|\}},
\]
and the sparsest cut value of $G$ is defined as 
\[\psi(G) := \min_{S \subseteq V} \psi_G(S).\] 
In the fully dynamic sparsest cut problem, an input graph \(G\) undergoes edge insertions and deletions, and the goal is to maintain query access to the sparsest cut value in $G$. In this section, we show the data structure from \Cref{th:main} can be used to maintain query access to an $\Otil(\log n)^{L+1/2}$-approximation of the sparsest cut. The proof is deferred to \Cref{appendix:sparest-cut-proof}.

\begin{theorem}\label{thm:sparsest-cut}
Given an \(n\)-vertex graph \(G = (V, E, u)\) with polynomially bounded capacity ratio undergoing polynomially many edge insertions and deletions, parameters \(j \geq 1\) and \( 1 \leq L \leq \littleo{\sqrt{\log (n/j) / \log \log n}}\),
  there is a data structure that supports the following operations against an oblivious adversary:
  \begin{itemize}
      \item $\textsc{InsertEdge}(u, v, c)$: Inserts an edge $(u,v)$ with capacity $c$ to $G$ in $O((n/j)^{2/L} \log^{\atmost{L}} n)$ amortized time.
      \item $\textsc{DeleteEdge}(u, v)$: Removes the edge $(u,v)$ from $G$ in $O((n/j)^{2/L} \log^{\atmost{L}} n)$ amortized time.
      \item $\textsc{SparsestCut}()$: Returns an $\Otil(\log n)^{L+1}$-approximation of the sparsest cut value of $G$ in time $O(j^{2} \log^{O(L)}n)$.
  \end{itemize}
The algorithm is correct with high probability.
\end{theorem}
By setting $j = n^{1/(1+L)}$ to balance both query- and update time, and by pushing the query time in the update time through maintaining a solution during updates, we in particular get the following. 
\begin{corollary}
    For any fixed integer $L \geq 1$, there is a dynamic data structure that w.h.p against an oblivious adversary returns an $\Otil(\log n)^{L+1}$ approximation of the sparsest cut value in $O(\log^{O(L)} n^{\frac{2}{1+L}})$ amortized update time and $O(1)$ query time.
\end{corollary}

\begin{corollary}
    There is a dynamic data structure that w.h.p against an oblivious adversary returns an $O(\polylog (n))$ approximation of the sparsest cut value in $O(n^{0.001})$ amortized update time and $O(1)$ query time.
\end{corollary}
Similar to the previous section, by letting $L$ grow as a function of $n$ and setting $L = \Theta(\sqrt{\log n} / \log \log n)$, we also get a datastructure with sub-polynomial update time.
\begin{corollary}
    There is a dynamic data structure that w.h.p against an oblivious adversary returns an $e^{O(\log^{1/2} n \log \log n)}$ approximation of the sparsest cut value in $O(e^{O(\log^{1/2} n \log \log n)})$ amortized update time and $O(1)$ query time.
\end{corollary}






\section{Acknowledgement}

Monika Henzinger: Funded by the European union. Views and opinions expressed are however those of the author(s) only and do not necessarily reflect those of the European Union or the European Research Council Executive Agency. Neither the European Union nor the granting authority can be held responsible for them.

\noindent
\parbox[b]{0.65\textwidth}{
\setlength{\parindent}{18pt}    
\indent This project has received funding from the European Research Council (ERC) under the European Union's Horizon 2020 research and innovation programme (MoDynStruct, No. 101019564) and the Austrian Science Fund (FWF) grant  \href{https://www.doi.org/10.55776/I5982}{DOI 10.55776/I5982}.
For open access purposes, the author has applied a CC BY public copyright license to any author-accepted manuscript version arising from this submission.
}
\parbox[b]{0.3\textwidth}{
\hfill    \null\vfill \centering \includegraphics[width=1.1\linewidth]{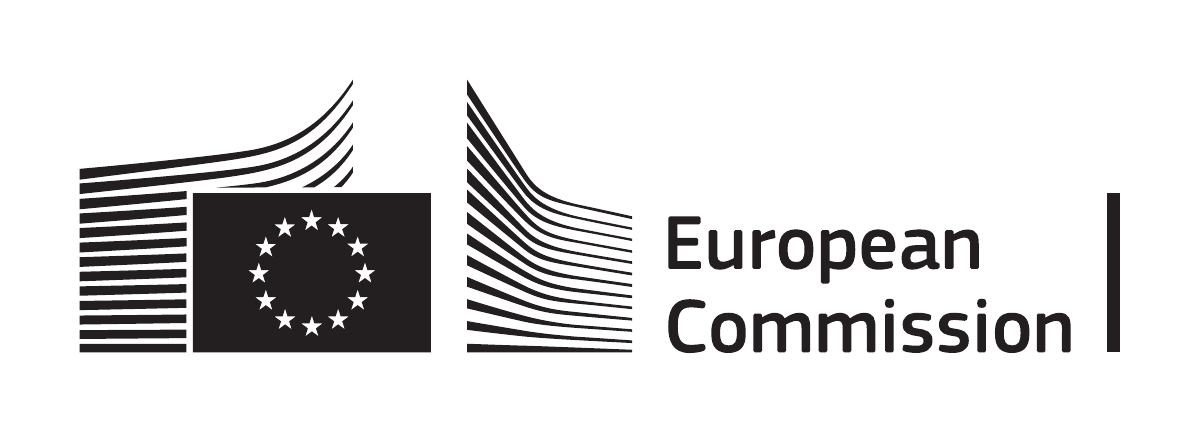} \null\vfill
}

Peter Kiss: This research was funded in whole or in part by the Austrian Science Fund (FWF) 10.55776/ESP6088024.

\newpage
\appendix
\appendix
\section{Proofs of auxiliary lemmas}\label{appendix:splitefficient}
We start with a proof of \Cref{lma:splitefficient}.
\splitefficient*
\begin{proof}
Let $G^{(t)}$ be the graph at time $t$ and consider the following potential function
\[
\Phi^{(t)} := \sum_{v \in V(G^{(t)})} \deg_{G^{(t)}}(v) \log \deg_{G^{(t)}}(v).
\]
Clearly, upon initialization $\Phi^{(0)} \leq 2m \log n$, and if $G^{(t+1)}$ differs from $G^{(t)}$ by an edge insertion/deletion, we have $\Phi^{(t+1)} \leq \Phi^{(t)} + 2\log n$. Thus, the total increase of the potential over all edge updates is bounded by  $2m \log n$. Consider now a vertex split of $v$ into $u$ and $w$. Then the potential function decreases as follows:
\begin{align*}
& \Phi^{(t)} - \Phi^{(t+1)} = \deg_{G^{(t)}}(v) \log \deg_{G^{(t)}}(v) - \deg_{G^{(t+1)}}(u) \log \deg_{G^{(t+1)}}(u) - \deg_{G^{(t+1)}}(w) \log \deg_{G^{(t+1)}}(w) \\
& = \deg_{G^{(t+1)}}(u) \cdot (\log \deg_{G^{(t)}}(v) - \log \deg_{G^{(t+1)}}(u)) + \deg_{G^{(t+1)}}(w) \cdot (\log \deg_{G^{(t)}}(v) - \log \deg_{G^{(t+1)}}(w)) \\
& \geq \deg_{G^{(t+1)}}(u) \cdot (\log \deg_{G^{(t)}}(v) - \log \deg_{G^{(t+1)}}(u)) \\
& \geq \deg_{G^{(t+1)}}(u) \cdot (\log 2\deg_{G^{(t+1)}}(u) - \log \deg_{G^{(t+1)}}(u)) \\
& = \Omega(\deg_{G^{(t+1)}}(u))
\end{align*}
Here we used that $\deg_{G^{(t)}}(v) = \deg_{G^{(t+1)}}(u) + \deg_{G^{(t+1)}}(w)$ and that without loss of generality we may assume that $\deg_{G^{(t+1)}}(u) \leq \deg_{G^{(t+1)}}(v)/2$. 

Now note that we may simulate the vertex split of $v$ by first deleting the edges of $v$ that would be incident on $u$  and then inserting them incident on an isolated vertex. This requires $O(\deg_{G^{(t+1)}}(u)) = O(\Phi^{(t)} - \Phi^{(t+1)})$ edge updates. Since the potential function always satisfies $\Phi \geq 0$, initially satisfies $\Phi^{(0)} \leq 2m \log n$, and can increase through edge insertions by at most $2 m \log n$ over time, all vertex splits can be simulated through $O(m \log n)$ edge updates.
\end{proof}

\section{Counterexample to Tree Embeddability}\label{appendix:madry-almost-tight}
Here, we present a well-known argument that there exist sparse graphs $G$ that require $k = \Omega(n/\alpha)$ spanning trees $T_i \subseteq G$ in order to satisfy embeddability of their average $k^{-1} \sum_{i=1}^k T_i \preceq_{\alpha} G$. This also shows that in \Cref{thm:jtree} and \Cref{thm:madry}, the results are almost tight.

The construction is based on constant degree expanders. We remind the reader that a $d$-regular graph $G = (V, E)$ on $n$ vertices is called a $\delta$-expander if, for any set $S \subseteq V$ with $|S| \leq n/2$, we have $|E(S, V \setminus S)| \geq \delta \cdot d \cdot |S|$.
\begin{theorem}[\cite{pinsker1973concentrator}]
    There exists a fixed $\delta > 0$, such that for any $d \geq 3$ and even integer $n$, there is a $d$-regular $\delta$-expander $G$ on $n$ vertices.
\end{theorem}
We now show the following.
\begin{theorem}
    Let $G = (V, E)$ be a $d$-regular $\delta$-expander, and $T_1, \ldots, T_k \subseteq G$ be spanning trees of $G$. Let $\alpha \geq 1$ and suppose that 
    \[
    \frac{1}{k} \sum_{i=1}^k T_i \preceq_{\alpha} G.
    \]
    Then $k  \geq \delta n / (2 \alpha)$.
\end{theorem}
\begin{proof}
    Let $T \subseteq V$ be a spanning tree of $G$, and let $v \in V$ be a vertex-separator of $T$, i.e., a vertex $v$ such that upon removal, all resulting sub-trees have size at most $n/2$. As $G$ has maximum degree $d$ (and consequently so does $T \subseteq G$), there must exist a resulting sub-tree $T'$ with $|T'| \geq  n/(2d)$. But then, as $G$ is a $\delta$-expander, we know that $E(T', V \setminus T') \geq \delta d \cdot n /(2d)$. Any embedding $\f$ of $G$ into $T$ thus must route $\delta n /2$ units of flow along the edge connecting $T'$ to $v$. Consequently, there exists an edge $e \in T$ with $u^T(e) \geq \delta n/2$. 

    In particular, if $H := \frac{1}{k} \sum_{i=1}^k T_i$, it is the case that for the given edge $e$, $u_{H}(e) \geq \delta n/(2k)$. Consequently, in order to have $H \preceq_{\alpha} G$, we need to have $\delta n/(2k) \leq \alpha$, i.e., $k \geq \delta n/(2\alpha)$. 
\end{proof}

\begin{remark}
    A similar argument shows that if $G$ is a $d$-regular $\delta$-expander and $H$ is an induced $j$-tree of $G$, there must exist an edge with $u_H(e) \geq \Omega(\delta n /  j)$, so that $k = \Omega(\delta n / (j \alpha))$ many induced $j$-trees are required in order for their average to $\alpha$-embed back into $G$.
\end{remark}


\section{Proof of \Cref{thm:madry}}\label{appendix:madryproof}
Here, we want to provide a complete proof of \Cref{thm:madry}. We believe the proof is somewhat simpler than the one provided by Madry in \cite{Madry:2010ab}. The main difference in our proof is that we can use a more standard application of the MWU framework by directly bounding the \emph{width} of the instance. This requires introducing a new regularization term in the Low-Stretch-Spanning-Tree based approach of Madry. 

We remind the reader that for a spanning tree $T$ of $G$ and edge $e \in E_G$, we defined $\fcong_G^T(e) = u^T(e) / u_G(e)$ if $e \in E_T$ and $\fcong_G^T(e) = 1$ otherwise. For convenience of the reader, we now re-state the theorem we will prove. 
\madry*
\subsection{Application of MWU}
In this section, we will show the proof of \Cref{thm:madry} conditioned on the following result, which we will prove in the next section.
\begin{restatable}{lemma}{MWUlemma}\label{lma:MWUlemma}
Let $G$ be a graph, $j$ be a parameter and $w \in \R_{>0}^E$ be an arbitrary weight vector. Then in time $\Otil(m)$ one can find a spanning tree $T$ of $G$ and a set of edges $S \subseteq E_T$ satisfying $|S| \leq j$ such that
\begin{enumerate}
    \item $\sum_{e \in E_G} w(e) \cdot \fcong_G^T(e) \leq \Otil(\log n) \cdot \norm{w}_1$, and
    \item $\max\limits_{e \not \in S} \fcong_G^T(e) \leq \Otil(\log n) \cdot m/j$.
\end{enumerate}
\end{restatable}
\begin{proof}[Proof of \Cref{thm:madry}]
    We follow the proof of Lemma $6.6$ in \cite{Chen:2023aa} closely, and construct the trees sequentially by a standard application of the Multiplicative Weights Update \cite{arora2012multiplicative}. 
    
    Let $\gamma \in \Otil(\log n)$ be such that the trees from \Cref{lma:MWUlemma} satisfy $\sum_{e \in E_T} w(e) \frac{u^T(e)}{u_G(e)} \leq \gamma \cdot \norm{w}_1$ and $\max_{e \in E_T \setminus S} \frac{u^T(e)}{u_G(e)} \leq \gamma \cdot m/j$. Let $k = 10 \gamma m/j = \Otil(\log n) m/j$. We initially set $w_0 \in \R^E$ as $w_0 = \vecone^E$ and compute a tree $T_0$ and set $S_0$ using \Cref{lma:MWUlemma} with respect to this weight function. 
    
    Afterwards for $t \in [k]$ iterations we define $w_{t}(e) = \exp(\frac{1}{k} \mathbf{1}\{e \not \in S_{t-1}\} \fcong_G^{T_{t-1}}(e))$ for all $e \in E_G$ and use \cref{lma:MWUlemma} to obtain tree $T_t$ and edge set $S_t$ using weights $w_t$.
    

    Observe that for all $t \in [k]:$
    \begin{align*}
        \max_{e \in E_G} \frac{1}{k} \sum_{i=0}^{t-1} \mathbf{1}\{e \not \in S_i\} \fcong_G^{T_i}(e) \leq \log \rbrack{\sum_{e \in E_G} \exp\rbrack{\frac{1}{k}\sum_{i=0}^{t-1} \mathbf{1}\{e \not \in S_i\} \fcong_G^{T_i}(e)}} = \log \norm{w_{t}}_1
    \end{align*}

    By our choice of $\gamma$ we have that for all $t \in [k]$:

    \begin{align*}
        \frac{1}{k}\mathbf{1}\{e \not \in S_t\} \fcong_G^{T_t}(e) \leq \gamma m / (jk) \leq 1/10
    \end{align*}

    Hence, we can conclude that for all $t \in [k]:$
    \begin{align*}
        \norm{w_{t}}_1 &= \sum_{e \in E_G} w_{t-1}(e) \exp\rbrack{\frac{1}{k}\mathbf{1}\{e \not \in S_{t-1}\} \fcong_G^{T_{t-1}}(e)} \\
        &\leq \sum_{e \in E_G} w_{t-1}(e) \rbrack{1+\frac{2}{k}\mathbf{1}\{e \not \in S_{t-1}\} \fcong_G^{T_{t-1}}(e)} \\
        & = \norm{w_{t-1}}_1 + \frac{2}{k} \sum_{e \in E_G} \mathbf{1}\{e \not \in S_{t-1}\} \fcong_G^{T_{t-1}}(e) \\
        & \leq (1 + 2 \gamma / k) \norm{w_{t-1}}_1
    \end{align*}


    Here we used that $\exp(x) \leq 1+ 2x$ for $x \leq 1/10$ and the second item of \Cref{lma:MWUlemma}.

    Finally, we can conclude that:
    
    \begin{align*}
        \exp \rbrack{\max_{e \in E_G} \frac{1}{k} \sum_{i=1}^k \mathbf{1}\{e \not \in S_i\} \fcong_G^{T_i}(e)} \leq \norm{w_{k}}_1 \leq (1+2 \gamma /k)^k \norm{w_1}_1 \leq \exp(2 \gamma) m
    \end{align*}

    Taking the logarithm of both sides concludes the lemma.
     
\end{proof}

\subsection{Proof of \Cref{lma:MWUlemma}}
The only building-block that our algorithm will require is a result on Low-Stretch-Spanning-Trees, which we describe now.

\paragraph{Low Stretch Spanning Trees} Let $G = (V, E)$ be a graph with a length vector $l \in \R_{> 0}^E$ on its edges. Let $T$ be a spanning tree of $G$, and let $d^l_G, d_T^l$ denote the shortest path metric induced by $l$ on $G$ and $T$, respectively. Of course, $d^l_G(u, v) \leq d^l_T(u,v)$ for all $u,v$. Define the stretch $\str_T^l(e)$ of an edge $e = (u, v) \in E_G$ as 
\[
\str_T^l(e) := \frac{d^l_T(u, v)}{d^l_G(u, v)}.
\]
Then the following theorem shows that one can efficiently compute a spanning tree $T$ for which the stretch \emph{on average} is a rather small factor of $\Otil(\log n)$.
\begin{theorem}[\cite{Abraham:2019aa}]\label{thm:LSST}
    Given a (multi-)graph $G = (V, E)$ and a length vector $l \in \R_{> 0}^E$, there is an algorithm that computes in $\Otil(m)$ time a spanning tree $T$ of $G$ such that,
    \[
    \frac{1}{m}\sum_{e \in E} \str_T^l(e) \leq \Otil(\log n).
    \]
\end{theorem}
From this theorem, it is well known how to show the following.
\begin{lemma}\label{lma:LSSTweighted}
    Given a (multi-)graph $G = (V, E)$, a length vector $l \in \R_{> 0}^E$ and a weight vector $w \in \R_{> 0}^E$, there is an algorithm that computes in $\Otil(m)$ time a spanning tree $T$ of $G$ such that,
    \[
    \sum_{e \in E} w(e)\str_T^l(e) \leq \Otil(\log n) \cdot \norm{w}_1.
    \]
\end{lemma}
\begin{proof}
    Construct an auxiliary graph $G_w$ on vertex set $V$, where for each edge $e \in G$, we add $\lceil m w(e) / \norm{w}_1 \rceil$ multi-edges. Then any spanning tree $T$ of $G_w$ is also a spanning tree of $G$, and, moreover, the number of edges in $G_w$ can be bounded as
    \[
    \sum_{e \in E_G} \lceil m w(e) / \norm{w}_1 \rceil \leq \sum_{e \in E_G} (1 + m w(e) / \norm{w}_1) = 2m.
    \]
    We now invoke \Cref{thm:LSST} to compute in time $\Otil(m)$ a spanning tree $T$ on $G_w$ with length function $l$ (here, each multi-edge of $G_w$ inherits the length of its original weight in $G$.) Of course, for each copy $e' \in G_w$ of the edge $e \in G$, we have that its stretch with respect to $T$ is the same in $G_w$ as the stretch of $e$ in $G$. Hence, by the guarantees of the spanning tree $T$ in $G_w$, we have that 
    \begin{align*}
        \sum_{e \in E_G} \str_T^l(e) \cdot  \lceil m w(e) / \norm{w}_1 \rceil = \sum_{e \in G_w} \str_T^l(e) \leq \Otil(\log n) \cdot m.
    \end{align*}
    We can now finish by calculating
    \begin{align*}
         \sum_{e \in E_G} w(e) \str_T^l(e) &= \frac{\norm{w}_1}{m}\sum_{e \in E_G} \frac{w(e) m}{\norm{w}_1} \str_T^l(e) \leq \frac{\norm{w}_1}{m}\sum_{e \in E_G} \left\lceil \frac{w(e) m}{\norm{w}_1} \right\rceil \str_T^l(e) \\
        & \leq \frac{\norm{w}_1}{m} \cdot \Otil(\log n) \cdot m = \Otil(\log n) \norm{w}_1.
    \end{align*}
    
\end{proof}
Using this theorem, we now show the following.
\begin{lemma}[See Lemma $5.4$ in \cite{Madry:2010ab}]
    Given a (multi-)graph $G = (V, E_G)$, and a length vector $l \in \R_{> 0}^E$, there is an algorithm that computes in $\Otil(m)$ time a spanning tree $T$ of $G$ such that
    \[
    \sum_{e \in E_G} l(e)u^T(e) \leq \Otil(\log n) \cdot \sum_{e \in E_G} l(e) u_G(e)
    \]
\end{lemma}
\begin{proof}
    Apply \Cref{lma:LSSTweighted} on $G$ with length vector $l$ and weight vector $w(e) = l(e) u_G(e)$. Notice that we have
    \begin{align*}
        \sum_{f \in E_T} l(f) u^T(f) & = \sum_{f \in E_T} l(f) \sum_{e = (u, v) \in E: f \in T[u, v]} u_G(e) \\
        & =  \sum_{e = (u, v) \in E} d_T^l(u, v) u_G(e) \\
        &\leq \sum_{e = (u, v) \in E} \frac{d_T^l(u, v)}{d_G^l(u, v)} u_G(e) l(e) \\
        & =  \sum_{e = (u, v) \in E} \str_T^l(e) w(e) \\
        &\leq \Otil(\log n) \sum_{e \in E_G} w(e) \\
        & = \Otil(\log n) \sum_{e \in E_G} l(e) u_G(e).
    \end{align*}
\end{proof}
We can conclude with the
\begin{proof}[Proof of \Cref{lma:MWUlemma}]
    Apply the previous lemma with length function $l(e) = \frac{w(e) + \norm{w}_1/m}{u_G(e)}$. Then we get a spanning tree $T$ satisfying
\begin{align*}
    \sum_{e \in E_G} l(e) u^T(e) &= \sum_{e \in E_G} \frac{u^T(e)}{u_G(e)} \rbrack{w(e) + \frac{\norm{w}_1}{m}} \leq \Otil(\log n) \cdot \sum_{e \in E_G} l(e) u_G(e) \\
    & = \Otil(\log n) \cdot \sum_{e \in E_G} \rbrack{w(e) + \frac{\norm{w}_1}{m}} = \Otil(\log n) \cdot 2 \norm{w}_1 := \gamma \cdot \norm{w}_1,
\end{align*}
which of course in particular implies that $\sum_{e \in E_G} w(e)\frac{u^T(e)}{u_G(e)} \leq \gamma \cdot \norm{w}_1$ for some $\gamma \in \Otil(\log n)$, and hence shows the first property of the theorem. Consider now the set $S = \{e \in E_T: u^T(e) / u_G(e) \geq \gamma m/j\}$ and observe that 
\begin{align*}
     \gamma \cdot \norm{w}_1 &\geq \sum_{e \in E_G} l(e) u^T(e) \geq \sum_{e \in S} l(e) u^T(e) = \sum_{e \in S} \frac{u^T(e)}{u_G(e)} \rbrack{w(e) + \frac{\norm{w}_1}{m}} \\
    & \geq \sum_{e \in S} \frac{u^T(e)}{u_G(e)} \frac{\norm{w}}{m} \geq |S| \cdot \frac{\norm{w}}{m} \cdot \gamma m/j.
\end{align*}
Consequently, $|S| \leq j$ and by design it satisfies the desired condition.
\end{proof}


 \section{Proof of \Cref{thm:multi}} \label{appendix:multicut}

In order to employ the cut sparsifier of \Cref{th:main} to work against multi and multi-way cuts we will first state some of the key claims and lemmas concerning cut preservation of $j$-trees in terms of multi-cuts.

\begin{lemma}\label{lma:multicutpreserve}
    Let $G, H$ be two graphs on the same vertex set and suppose that $H \preceq_{\alpha} G$. Then for any partition $\{S_i\}$ we have that $U_H(\{S_i\}) \leq \alpha \cdot U_G(\{S_i\})$. 
\end{lemma}
\begin{proof}
    By \Cref{lma:cutpreserve}, we have that $u_H(S_i) \leq \alpha \cdot U_G(S_i)$ for all $i = 1, 2, \ldots, k$. In particular, $U_H(\{S_i\}) = \frac{1}{2} \sum_{i = 1}^k U_H(S_i) \leq \frac{1}{2} \sum_{i = 1}^k \alpha \cdot U_G(S_i) = \alpha \cdot U_G(\{S_i\})$.
\end{proof}

The following lemma can be seen as an extension of \cref{lma:jtreesample} to multi-way cuts. 

\begin{lemma}
\label{lma:jtreesample:multiway}
    Let $H_1, \ldots, H_k$ be the $O(j)$-trees from \Cref{thm:jtree} satisfying $G \preceq k^{-1} \sum_{i=1}^k H_i \preceq_{\alpha} G$. Let $\{S_i\}$ be a partition of $V$. Then, if we sample $l = O(\log n)$ $O(j)$-trees $H_{i_1}, \ldots, H_{i_l}$ uniformly at random, it holds that
    \[
    U_G(\{S_i\}) \leq \min_{1 \leq j \leq l} U_{H_{i_j}}(\{S_i\}) \leq 2 \alpha \cdot U_G(\{S_i\})
    \]
    with high probability. 
\end{lemma}
\begin{proof}
    By the previous lemma, we have 
    \[
    U_G(\{S_i\}) \leq k^{-1} \sum_{i=1}^k U_{H_i}(\{S_i\}) \leq \alpha \cdot U_G(\{S_i\}).
    \]
    Consequently, by Markov's inequality, for any randomly sampled $O(j)$-tree $H_{i_j}$, it holds that $\Pr[U_{H_{i_j}}(\{S_i\}) \geq 2 \alpha U_G(\{S_i\})] \leq 1/2$. Hence, sampling $l = O(\log(n))$ $O(j)$-trees uniformly at random suffices.
\end{proof}

Recall that on a high level we have shown that the hierarchy $\mathcal{H}$ maintained by \cref{th:main} will approximately preserve the cuts of the underlying graph through an inductive argument: by \cref{thm:fullydynamicsparsifier} for any $i \in [L]$ the sparsified cores $\widetilde{C}(H_{i-1})$ approximately preserve the cuts in their denser counterparts $C(H_{i-1})$ and by \cref{lma:jtreesample} the set of $O(j_{i+1})$ trees of $\mathcal{H}(H_i) \subseteq \mathcal{H}_{i+1}$ approximately preserve the cuts of their larger counterparts in $\widetilde{C}(H_{i})$. One can confirm that throughout this inductive argument by replacing \cref{lma:jtreesample} with \cref{lma:jtreesample:multiway} it can be shown that $\mathcal{H}$ preserves the sizes of multi-way cuts as well up to an $\tilde{O}(\log n)^L$ factor.

Upon a query with a set of $k$ terminals or terminal pairs the data structure first moves all of the $O(k)$-terminals into the core of every $O(j)$-tree $H \in \mathcal{H}$ without rebuilding the lowest level of the hierarchy when the cores grow beyond of size $\Theta(j)$. This can be achieved by inserting an arbitrary edge between all terminal vertices and the core and then removing it. By \Cref{th:main} this takes $O(k \cdot (n/j)^{2/L} \poly{\log^L n} \log U)$ time.

Once the terminals are moved to the core, any minimum size cut separating the terminals or terminal pairs will not cut any of the edges of the forest incident on the core. Hence, we can restrict attention to the sparsified cores of each $O(j)$-tree in $H \in \mathcal{H}$. Due to the movement of the terminals to the cores, each of them has $O(k+j)$ vertices and $\atmosttilde{(j+k)}$ edges (recall that the cores of the $O(j)$-trees are sparsified through \cref{thm:fullydynamicsparsifier}). After the query is answered, we revert the data structure to its state before the query.

Finally, we re-insert all $O(k)$ edges again and we delete them from the graph, this time handling them normally with rebuilding the core if it grows above $\Theta(j)$ size. This seemingly redundant step serves as analytic trick to be able to amortize the cost of the $O(k)$-edge updates required to push the terminal vertices into the core. Recall that \cref{th:main} only shows that the amortized update time of edge updates is bounded by $O(k \cdot (n/j)^{2/L} \poly{\log^L n} \log U)$. Hence, it could be that the $O(k)$ edge updates required to pull the terminal vertices into the core is expensive.

If we were to simply reverse the affects of these updates in could occur that they are always expensive. However, note that the second set of updates are normal edge updates which are not reverted, hence their cost amortizes to be $O(k \cdot (n/j)^{2/L} \poly{\log^L n} \log U)$. Furthermore, the cost of the first set of updates is upper bounded by the second set of updates. Hence, we can move the terminals into the core in $O(k \cdot (n/j)^{2/L} \poly{\log^L n} \log U)$ update time.

It remains to show how to obtain approximate the value of a multi-way and multi-cut in the core up to $\polylog(n)$-factors, before we have reverted the first set of updates. Equivalently, we have to show how for a graph $G= (V,E,u)$ on $n$ vertices with $\tilde{O}(n)$ edges we may approximate the size of a minimum multi-way cut and multi-cut within a $O(\polylog(n))$-factor with respect to a terminal or terminal pair set of size $k \leq n$. Multiple algorithms exist for these purposes in literature with $\poly{n}$ running time. In order to optimize the $k$-dependency of our theorem we will describe some simple methods using standard techniques, for which we claim no novelty, running in $\tilde{O}(n)$ and $\tilde{O}(n^2)$ times respectively for multi-way cut and multi-cut.

\begin{lemma}

Let $G = (V,E,u)$ be a graph on $n$ vertices and $\tilde{O}(n)$ edges and $T \subseteq V$ be a set of terminals. There is an algorithm which returns the size of a $O(\log n)$-approximate multi-way cut of $G$ with respect to $T$ in $\tilde{O}(n)$ time.

\end{lemma}

\begin{proof}

Partition the terminals $T$ into two sets some $t = O(\log n)$ times such that any pair of vertices in $T$ in at least one of the partitions falls into a different set. This can be achieved, for example, through random sampling with high probability. Let the sets of partitions be $T_{i_1},T_{i_2} : i \in [t]$.

For each partitioning create a graph $G_i$, each are extensions of $G$ with a pair of added sink and source vertices $s,t$, each of which are connected to one side of the partitioning with an edge of very high capacity. In each $G_i$ calculate an $\alpha$-approximate $s-t$ min cut for some $\alpha = 1$ in $\tilde{O}(n)$ time (this can be done for example through \cite{Li2023NearLinear}). The algorithm returns the size of the union of the cuts.

Observe that the union of the cut is a multi-way cut of $G$ with respect to $T$ as any pair of vertices must be cut by one of the $s-t$ cuts. Furthermore, for any $G_i$ the minimum size multi-way cut of $G$ with respect to $T$ separates $T_{i_1}$ and $T_{i_2}$ hence all of the $s-t$ cuts calculated have cost at most $\alpha$ times the minimum size multi-way cut of $G$. Therefore, the union of the $s-t$ cuts is $O(\log n)$ approximate with respect to the min cost multi-way cut of $G$ with respect to $T$.

\end{proof}

\begin{lemma}

Let $G = (V,E,u)$ be a graph on $n$ vertices and $\tilde{O}(n)$ edges and $T = \{(s_1,t_1), \dots, (s_k,t_k)\}$ be a set of at most $n$ terminal pairs in $V$. There is an algorithm which returns the size of a $\tilde{O}(\log^2 n)$-approximate multi-cut of $G$ with respect to $T$ in $\tilde{O}(n^2)$ time.

\end{lemma}

\begin{proof}

Using \Cref{thm:madry} we obtain a collection of $n$ trees ($j$-trees for $j = 1$) in $\tilde{O}(n^2)$ time. By \cref{lma:multicutpreserve} and \cref{lma:jtreesample:multiway} a sample of $\tilde{O}(n)$ of these trees $\tilde{O}(\log n)$-preserves the multi-cut values of $G$ with respect to $T$. Hence, it remains to obtain a $O(\log n)$-approximation to the multi-cut value on each of the sampled trees in $\tilde{O}(n^2)$ time.

We may turn the multi-cut problem into an instance of the weighted set cover problem the following way: we define the sets to correspond to the edges with the edges weight, the elements to correspond to the terminal pairs and say that a set contains an element if the associated edge cuts the associated terminal pair. We may initialize the weighted set cover through enumerating the unique paths between any pairs of terminals of the tree in $O(n^2)$ time.

Finally, note that its a folklore fact that the greedy algorithm for weighted set cover obtains an $O(\log n)$-approximation and runs in $O(n^2)$ time (number of sets times the number of elements). 

\end{proof}

 \section{Proof of \Cref{thm:sparsest-cut}}\label{appendix:sparest-cut-proof}
In this section, we will provide a proof of \Cref{thm:sparsest-cut}, i.e., we will show how to use the dynamic $O(j)$-tree data structure from \Cref{th:main} in order to dynamically maintain query access to approximate sparsest cut values. 

Our proof will consist of two parts. First, we state and prove a simple structural lemma which shows that, for any $j$-tree $H$, there exists a sparsest cut that either cuts only edges with endpoints in the core $C(H)$, or it cuts a single forest edge $e \in F(H)$. Remember that for a $j$-tree $H = C(H) \cup F(H)$, the graph $F = F(H)$ is a rooted forest and the core $C = C(H)$ is an induced sub-graph on the root set $R$ of $F$ (in particular, $V(C) = R$). 

Next, we show how to efficiently maintain the sparsest cut among the tree edges, i.e., $\min_{e \in E_F} \psi(S(e))$. This will be done by slightly modifying the data structure from \Cref{thm:jtree} by employing ideas from \cite{Chen:2020aa} that force the paths in the forests $F_i$ to be sufficiently short. 

Using both parts, we can then show that upon query, we only need to run an $\Otil(\log n)$-approximate (generalized) sparsest cut algorithm on the (sparsified) cores $\Tilde{C}(H)$ for all $H \in \cH$. As the cores $\Tilde{C}(H)$ are only on $\Otil(j)$ edges, the sparsest cut algorithm takes only $\Otil(j^{2})$ time. 

\paragraph{The Structural Lemma}
Let $H \in \cH$ be an $O(j)$-tree with forest $F$ and (sparsified) core $\Tilde{C}$. For an edge $e =(u, v) \in E_F$ such that $u$ is the child of $v$, define $S(e) := u^{\downarrow} \subseteq V$ as the cut induced by considering the set of descendants of $u$ (which includes $u$). Remembering that each $v \in V(\Tilde{C})$ is the root of a unique tree in $F$, we may also define $S^{\downarrow} = \bigcup_{v \in S} v^{\downarrow}$ for sets $S\subseteq V(\Tilde{C})$. The main lemma shows that in order to find a sparsest cut on a $j$-tree, it suffices to consider cuts inside the core and forest separately. That is, it shows that there always exists a sparsest cut $S$ that only cuts \emph{either} edges in the core \emph{or} edges in the forest, but never both.

We start with an easy auxiliary lemma.
\begin{lemma}
Let $H$ be a graph. Then there always exists a sparsest cut $S$ such that $H[S]$ is connected.
\end{lemma}
\begin{proof}
Let $S$ be a sparsest cut with $|S| \leq n/2$, and suppose that $H[S]$ consists of connected components $S_1, \ldots, S_k$. Then
\[
\psi_G(S) = \frac{U_H(S_1) + \ldots + U_H(S_k)}{|S_1| + \ldots + |S_k|} \geq \min_{1 \leq j \leq k} \frac{U_H(S_j)}{|S_j|} = \min_{1 \leq j \leq k} \psi_H(S_j),
\]
and so one of the connected components must also constitute a sparsest cut.
\end{proof}

We can now state and proof the main structural lemma.
\begin{lemma}\label{lma:sparsest-cut-structure}
    Let $H$ be a $j$-tree with core $C(H)$ and forest $F(H)$. Then there exists a sparsest cut $S$ such that, either
    \begin{enumerate}
        \item $S = S(e)$ for some forest edge $e \in E_F$, or,
        \item $S = X^{\downarrow}$ for some set $X \subseteq V(C(H))$.
    \end{enumerate}
\end{lemma}
\begin{proof}
Assume that there does not exist a forest edge $e \in E_F$ such that $\psi(S(e)) = \psi_H(S)$, i.e., that $U_H(S(e))/\min\{|S(e)|, n - |S(e)|\} > U_H(S) / \min\{|S|, n-|S|\}$, as otherwise there is nothing to show. Furthermore, we can assume that $|S| \leq n/2$, and, by the previous lemma, that $H[S]$ is connected. 

Our goal is now to show that $S = X^{\downarrow}$ for some set $X \subseteq V(C(H))$. This is equivalent to showing that $E(S, V \setminus S)$ does not contain any forest edge $e \in E_F$. We will prove the claim by contradiction. Assume that there exists a forest edge $e \in E_F$ such that $e \in E(S, V \setminus S)$. Then, we will show that the cut $S' = S \cup S(e)$ satisfies $\psi_H(S') < \psi_H(S)$, which constitutes the contradiction. 

Notice that by assumption of $H[S]$ being connected, we know that $S \cap S(e) = \emptyset$. In particular, this implies that $U_H(S') = U_H(S) - U_H(S(e))$. We now first note that

\begin{align*}
& \psi_H(S') = \frac{U_H(S) - U_H(S(e))}{\min\{|S'|, n - |S'|\}} = \frac{U_H(S) - U_H(S(e))}{\min\{|S|+|S(e)|, n - |S| - |S(e)|\}}< \frac{U_H(S)}{|S|} = \psi_H(S)\\
& \Leftrightarrow \quad \rbrack{U_H(S) - U_H(S(e))} |S| < U_H(S) \min\{|S|+|S(e)|, n - |S| - |S(e)|\} \\
& \Leftrightarrow \quad U_H(S) \rbrack{|S| - \min\{|S|+|S(e)|, n - |S| - |S(e)|\}} < U_H(S(e)) |S|.
\end{align*}

Using the assumption that $\psi_H(S) < \psi_H(S(e))$, it thus suffices to show that 
\[U_H(S) \rbrack{|S| - \min\{|S|+|S(e)|, n - |S| - |S(e)|\}} \leq U_H(S) \min\{|S(e)|, n - |S(e)|\},\]
i.e., that 
\[
|S| - \min\{|S|+|S(e)|, n - |S| - |S(e)|\} \leq \min\{|S(e)|, n - |S(e)|\}.
\]
Consider first the case that $\min\{|S|+|S(e)|, n - |S| - |S(e)|\} = |S|+|S(e)|$. Then the statement is trivially true as $|S| - (|S|+|S(e)|) = -|S(e)| \leq \min\{|S(e)|, n - |S(e)|\}$. Consider now the case that $\min\{|S|+|S(e)|, n - |S| - |S(e)|\} = n-|S|-|S(e)|$ and that $\min\{|S(e)|, n - |S(e)|\} = |S(e)|$. Then we need to show that $2|S| - n + |S(e)| \leq |S(e)|$, i.e., that $2|S| \leq n$, which is true as $|S| \leq n/2$. Finally, consider the case that $\min\{|S|+|S(e)|, n - |S| - |S(e)|\} = n-|S|-|S(e)|$ and that $\min\{|S(e)|, n - |S(e)|\} = n - |S(e)|$. Then we need to show that $2 |S| + |S(e)| - n \leq n - |S(e)|$, i.e., that $2|S| + 2|S(e)| \leq 2n$. This is true as $2|S|+2|S(e)| = 2 (|S|+|S(e)|) = 2 |S \cup S(e)| \leq 2n$.
\end{proof}

\paragraph{Maintaining The Sparsest Cut in The Forest}
We now want to  prove an extension to \Cref{th:single-level}. Before we can do this, let us generalize the notion of a sparsest cut. Suppose that the graph $H$ is additionally equipped with vertex weights $w \in \R_{>0}^{n}$, and for  a subset $S \subseteq V$ define $w(S) = \sum_{v \in S} w(v)$. Then the sparsity of a cut $S$ is defined as $\psi_H(S) = U_G(S)/\min \{w(S), w(V) -w(S)\}$, and the sparsest cut value as $\psi(H) := \min_{S \subseteq V} \psi_H(S)$. Notice that we recover the previous definition by setting $w = \vecone$.

\begin{lemma}\label{lma:maintain-forest-cuts}
    Let $H_1, \ldots, H_k$ be the set of $k = O(\log n)$ many $O(j)$-trees maintained by \Cref{th:single-level}. Suppose further that the underlying dynamic graph $G$ on which the algorithm is running is equipped with vertex weights $w \in \R_{>0}^{n}$. 
    
    Suppose that whenever the underlying graph $G$ on which \Cref{th:single-level} is running changes by an edge insertion or deletion, additionally $O(1)$ entries of the weight function $w$ get updated in an arbitrary fashion. 
    
    Then, for every $H_i$, we can dynamically maintain the minimum (weighted) sparsity among its forest edges, i.e., we maintain the value $\min_{e \in E(F(H_i))} U_{H_i}(S(e)) / \min \{w(S), w(V) -w(S)\}$. This can be done in time $\Otil(n/j)$ per update. In particular, the asymptotic runtime of \Cref{th:single-level} does not change.
\end{lemma}
\begin{proof}
    We employ a path-length reduction idea from \cite{Chen:2020aa} by making sure that after initializing the forest $F_i$, all paths on the forest have length at most $\Otil(n/j)$. Since the forest is decremental, the paths can only get shorter over time. Assuming such short path lengths, note first that \Cref{thm:madry} guarantees us that after the underlying graph receives an update, the forest $F_i$ changes by at most $O(1)$ edge deletions. If the forest $F_i$ changes by the deletion of an edge $e = (u, v)$, where $v$ is the parent of $u$, the only edges in the forest whose sparsity could have changed are those on the $v$-$\root_{F_i}(v)$ path. By running the data structure from \Cref{thm:ett} on the dynamic forest $F_i$, we can furthermore maintain under updates to $F_i$ and $w$, in $O(\log n)$ update time query access to $w(u^{\downarrow})$, i.e., the sum of the weights of the sub-tree rooted at vertex $u$. Thus, we can (re-)compute the sparsity of each edge on the $v$-$\root_{F_i}(v)$ path in $O(\log n)$ time each. As the path has length at most $O(n/j)$, this shows that we can maintain the sparsity of every edge in $F$ in $\Otil(n/j)$ time.

    For completeness, we now also show how to initialize $F_i$ to satisfy the bound on its path lengths. To do so, remember the proof of \Cref{thm:madry}, where we initially specified a branch-root set $R_i$ of size $|R_i| = O(j)$, and received from it the forest $F_i$ by cutting tree edges $e_{\min}^{T_i}(u, v)$ between vertices $u, v$ in the initial spanning tree $T_i$ of \Cref{thm:madry}. In the following, we show that we can find a set $S$ of size $|S| = O(j)$, such that after setting $R_i' = R_i \cup S$, making $R_i'$ branch-free, by using the set $R_i'$ as the initial set in \Cref{thm:jtree}, the corresponding forest $F_i$ satisfies that all its paths have length at most $O(n/j)$.

    To find the set $S$, we proceed as follows. If there exists a path of length $n/j$ from a leaf $v$ to its root $r$, we find an edge $e = (x, y)$ along the path such that the length of the path from $v$ to $x$ is at least $n/(2j)$. Then, we add both endpoints of $x, y$ to the set $R'$. We then know that \Cref{thm:jtree} updates the forest $F_i$ by removing the edge $e$. We then repeat this procedure until there are no more such paths. Each time this procedure is performed, a connected component $T$ of $F_i$ decreases by size at least $n/(2j)$. As for connected components of size less than $n/(2j)$, all paths have length at most $n/j$, it thus follows that the total number of iterations are $O(j)$, as desired. Each iteration can be performed in $O(n/j)$ time, so the total amount of time spent to find the set is $O(n)$.
\end{proof}

\paragraph{Putting Everything Together}


Let $H$ be an $O(j)$-tree from the set $\cH$, such that $\Tilde{C}_L$ is its core and $F = F_L \cup \ldots \cup F_1$ the corresponding forest. Let $w \in \R_{>0}^{V(\Tilde{C}_L)}$ be vertex weights satisfying $w(v) = |v^{\downarrow}|$, i.e., the size of the component $T$ in $F$ with root $v$. Then \Cref{lma:sparsest-cut-structure} tells us that we can find a sparsest cut in $H$ by
\begin{enumerate}
    \item Computing the sparsity of the cuts induced by removing a single forest edge, i.e., compute $\min_{e \in E_F} \psi_H(S(e))$.
    \item Computing the value of the sparsest cut that only cuts edges with endpoints in the core. This is equivalent to computing the value of a sparsest cut in the graph $\Tilde{C}_L$ with vertex weights $w(v) = |v^{\downarrow}|$.
    \item Finally, return the minimum value of both type of cuts.
\end{enumerate}

Our algorithm thus works by, for each $H \in \cH$, (1) maintaining the edge $e \in E_F$ that minimizes $\psi_H(S(e))$ (2) upon query employing a standard non-uniform sparsest cut algorithm. We first state the guarantees of the static algorithm we are using, which is obtained as a special case of approximate generalized sparsest cut algorithms. 

\begin{lemma}\label{thm:non-uniform-sparsest-cut}
There is an algorithm that, on a graph $G$ with vertex weights $w$, returns an $\Otil(\log n)$-approximate value to the vertex-weighted sparsest cut in $\Otil(m^2)$ time. 
\end{lemma}
\begin{proof}
Using \Cref{thm:madry} we obtain a collection of $\Otil(m)$ many trees ($j$-trees for $j = 1$) in $\tilde{O}(m^2)$ time. By \Cref{lma:cutsample}, sampling $O(\log n)$ of these trees uniformly at random guarantees that, with probability $1-1/n^c$, the vertex-weighted sparsest cut is preserved up to an $\Otil(\log n)$ factor by computing it for each tree separately and then returning the minimum value found. On a tree, however, this can trivially be done in $O(n)$ time by separately checking each edge.
\end{proof}

\medbreak


\begin{proof}[Proof of \Cref{thm:sparsest-cut}]
The data structure from \Cref{th:main} guarantees us that after any update to $G$ and for for any cut $S$, we have $U_G(S) \leq U_H(S)$ for all $H \in \cH$. Additionally, there with high probability exists an $H \in \cH$ such that $U_H(S) \leq \Otil(\log n)^L U_G(S)$. In particular, with high probability there exists an $H \in \cH$ for which the sparsest cut $S$ in $G$ satisfies $\psi_H(S) \leq \Otil(\log n)^L \psi_G(S)$. We can union bound the failure probability of over all updates performed to $G$, and thus assume that at all times there exists an $H$ such that $\psi(H) \leq \Otil(\log n)^L \psi(G)$. Hence, if we can compute an $\Otil(\log n)$-approximation of the sparsest cut for all $H \in \cH$, returning the minimum value found will yield the desired approximation.

Given a graph $H = \Tilde{C}_L \cup F$ with $F = F_L \cup \ldots \cup F_1$, by \Cref{lma:sparsest-cut-structure} we can achieve this goal by (1) maintaining $\min_{e \in E_F} \psi_H(S(e))$ and (2) upon query computing an $\Otil(\log n)$ approximation to the minimum sparsity of all cuts $S$ that cut only edges inside the core. We start by discussing the former.

\underline{Sparsest Cut in $F$}: In order to compute $\min_{e \in E_F} \psi_H(S(e))$, we compute $\min_{e \in E_{F_i}} \psi_H(S(e))$ for all forests $F_i$ and then return the minimum over all $i = 1, \ldots, L$. Remember, however, that $F_i$ as maintained by \Cref{thm:jtree} is a forest on vertex set $V(\Tilde{C}_{i-1})$ and \emph{not} on the full vertex set $V(H)$. That means that if we want to use \Cref{lma:maintain-forest-cuts} to maintain the sparsities, we need to introduce vertex weights $w(v) = |v^{\downarrow}|$, the size of the sub-tree of the \emph{full forest} $F$ rooted at $v$. However, this is precisely what the data structure from \Cref{lma:maintain-forest-cuts} enables us to do. Consequently, we can maintain the values of the sparsest cut in $F$ at all times. 

\underline{Sparsest Cut in $\Tilde{C}_L$}: Each graph $\Tilde{C}_L$ is a sparsified core and contains only $\Otil(j)$ edges. Consequently, we can employ the algorithm of \Cref{thm:non-uniform-sparsest-cut} with vertex weights $w(v) = |v^{\downarrow}|$ to compute in time $\Otil(j^{2})$ an $\Otil(\log n)$-approximation to the minimum sparsity of all cuts $S$ that cut only edges inside the core.


We can now conclude the proof by noting that $|\cH| = O(\log^L n)$, so that the total time spent upon query is $O(j^{2} \log^{O(L)}n)$.





\end{proof}

\bibliographystyle{siamplain}
\bibliography{references.bib}

@article{henzinger1999randomized,
	author = {Monika Rauch Henzinger and Valerie King},
	bibsource = {dblp computer science bibliography, https://dblp.org},
	doi = {10.1145/320211.320215},
	journal = {J. {ACM}},
	number = {4},
	pages = {502--516},
	title = {Randomized Fully Dynamic Graph Algorithms with Polylogarithmic Time per Operation},
	url = {https://doi.org/10.1145/320211.320215},
	volume = {46},
	year = {1999}}

@inproceedings{AbrahamBN08,
	author = {Ittai Abraham and Yair Bartal and Ofer Neiman},
	bibsource = {dblp computer science bibliography, https://dblp.org},
	booktitle = {49th \FOCS},
	pages = {781--790},
	publisher = {{IEEE} Computer Society},
	title = {Nearly Tight Low Stretch Spanning Trees},
	url = {https://doi.org/10.1109/FOCS.2008.62},
	year = {2008}}

@inproceedings{Abboud:2023aa,
	author = {Amir Abboud and Jason Li and Debmalya Panigrahi and Thatchaphol Saranurak},
	bibsource = {dblp computer science bibliography, https://dblp.org},
	booktitle = {64th \FOCS},
	doi = {10.1109/FOCS57990.2023.00137},
	pages = {2204--2212},
	publisher = {{IEEE}},
	title = {All-Pairs Max-Flow is no Harder than Single-Pair Max-Flow: Gomory-Hu Trees in Almost-Linear Time},
	url = {https://doi.org/10.1109/FOCS57990.2023.00137},
	year = {2023}}

@inproceedings{Chechik:2020aa,
	author = {Shiri Chechik and Tianyi Zhang},
	bibsource = {dblp computer science bibliography, https://dblp.org},
	booktitle = {\SODA},
	doi = {10.1137/1.9781611975994.28},
	editor = {Shuchi Chawla},
	pages = {463--475},
	publisher = {{SIAM}},
	title = {Dynamic Low-Stretch Spanning Trees in Subpolynomial Time},
	url = {https://doi.org/10.1137/1.9781611975994.28},
	year = {2020}}

@inproceedings{Forster:2019aa,
	author = {Sebastian Forster and Gramoz Goranci},
	bibsource = {dblp computer science bibliography, https://dblp.org},
	booktitle = {51st \STOC},
	doi = {10.1145/3313276.3316381},
	editor = {Moses Charikar and Edith Cohen},
	pages = {377--388},
	publisher = {{ACM}},
	title = {Dynamic low-stretch trees via dynamic low-diameter decompositions},
	url = {https://doi.org/10.1145/3313276.3316381},
	year = {2019}}

@inproceedings{Harrelson:2003aa,
	author = {Chris Harrelson and Kirsten Hildrum and Satish Rao},
	bibsource = {dblp computer science bibliography, https://dblp.org},
	booktitle = {15th \SPAA},
	doi = {10.1145/777412.777419},
	editor = {Arnold L. Rosenberg and Friedhelm Meyer auf der Heide},
	pages = {34--43},
	publisher = {{ACM}},
	title = {A polynomial-time tree decomposition to minimize congestion},
	url = {https://doi.org/10.1145/777412.777419},
	year = {2003}}

@inproceedings{Bienkowski:2003aa,
	author = {Marcin Bienkowski and Miroslaw Korzeniowski and Harald R{\"{a}}cke},
	bibsource = {dblp computer science bibliography, https://dblp.org},
	booktitle = {15th \SPAA},
	doi = {10.1145/777412.777418},
	editor = {Arnold L. Rosenberg and Friedhelm Meyer auf der Heide},
	pages = {24--33},
	publisher = {{ACM}},
	title = {A practical algorithm for constructing oblivious routing schemes},
	url = {https://doi.org/10.1145/777412.777418},
	year = {2003}}

@inproceedings{Racke:2014aa,
	author = {Harald R{\"{a}}cke and Chintan Shah and Hanjo T{\"{a}}ubig},
	bibsource = {dblp computer science bibliography, https://dblp.org},
	booktitle = {25th \SODA},
	doi = {10.1137/1.9781611973402.17},
	editor = {Chandra Chekuri},
	pages = {227--238},
	publisher = {{SIAM}},
	title = {Computing Cut-Based Hierarchical Decompositions in Almost Linear Time},
	url = {https://doi.org/10.1137/1.9781611973402.17},
	year = {2014}}

@article{Alon:1995aa,
	author = {Noga Alon and Richard M. Karp and David Peleg and Douglas B. West},
	bibsource = {dblp computer science bibliography, https://dblp.org},
	doi = {10.1137/S0097539792224474},
	journal = {{SIAM} J. Comput.},
	number = {1},
	pages = {78--100},
	title = {A Graph-Theoretic Game and Its Application to the k-Server Problem},
	url = {https://doi.org/10.1137/S0097539792224474},
	volume = {24},
	year = {1995}}

@article{Abraham:2019aa,
	author = {Ittai Abraham and Ofer Neiman},
	bibsource = {dblp computer science bibliography, https://dblp.org},
	doi = {10.1137/17M1115575},
	journal = {{SIAM} J. Comput.},
	number = {2},
	pages = {227--248},
	title = {Using Petal-Decompositions to Build a Low Stretch Spanning Tree},
	url = {https://doi.org/10.1137/17M1115575},
	volume = {48},
	year = {2019}}

@inproceedings{Brand:2023ab,
	author = {Jan van den Brand and Li Chen and Richard Peng and Rasmus Kyng and Yang P. Liu and Maximilian Probst Gutenberg and Sushant Sachdeva and Aaron Sidford},
	bibsource = {dblp computer science bibliography, https://dblp.org},
	booktitle = {64th \FOCS},
	doi = {10.1109/FOCS57990.2023.00037},
	pages = {503--514},
	publisher = {{IEEE}},
	title = {A Deterministic Almost-Linear Time Algorithm for Minimum-Cost Flow},
	url = {https://doi.org/10.1109/FOCS57990.2023.00037},
	year = {2023}}

@inproceedings{Dong:2022aa,
	author = {Sally Dong and Yu Gao and Gramoz Goranci and Yin Tat Lee and Richard Peng and Sushant Sachdeva and Guanghao Ye},
	bibsource = {dblp computer science bibliography, https://dblp.org},
	booktitle = {\SODA},
	doi = {10.1137/1.9781611977073.7},
	editor = {Joseph (Seffi) Naor and Niv Buchbinder},
	pages = {124--153},
	publisher = {{SIAM}},
	title = {Nested Dissection Meets IPMs: Planar Min-Cost Flow in Nearly-Linear Time},
	url = {https://doi.org/10.1137/1.9781611977073.7},
	year = {2022}}

@inproceedings{Madry:2016aa,
	author = {Aleksander Madry},
	bibsource = {dblp computer science bibliography, https://dblp.org},
	booktitle = {57th \FOCS},
	doi = {10.1109/FOCS.2016.70},
	editor = {Irit Dinur},
	pages = {593--602},
	publisher = {{IEEE} Computer Society},
	title = {Computing Maximum Flow with Augmenting Electrical Flows},
	url = {https://doi.org/10.1109/FOCS.2016.70},
	year = {2016}}

@inproceedings{Brand:2023aa,
	author = {Jan van den Brand and Yang P. Liu and Aaron Sidford},
	bibsource = {dblp computer science bibliography, https://dblp.org},
	booktitle = {55th \STOC},
	doi = {10.1145/3564246.3585135},
	editor = {Barna Saha and Rocco A. Servedio},
	pages = {1215--1228},
	publisher = {{ACM}},
	title = {Dynamic Maxflow via Dynamic Interior Point Methods},
	url = {https://doi.org/10.1145/3564246.3585135},
	year = {2023}}

@inproceedings{Daitch:2008aa,
	author = {Samuel I. Daitch and Daniel A. Spielman},
	bibsource = {dblp computer science bibliography, https://dblp.org},
	booktitle = {40th \STOC},
	doi = {10.1145/1374376.1374441},
	pages = {451--460},
	title = {Faster approximate lossy generalized flow via interior point algorithms},
	url = {https://doi.org/10.1145/1374376.1374441},
	year = {2008}}

@inproceedings{Sherman:2017aa,
	author = {Jonah Sherman},
	bibsource = {dblp computer science bibliography, https://dblp.org},
	booktitle = {28th \SODA},
	doi = {10.1137/1.9781611974782.49},
	editor = {Philip N. Klein},
	pages = {772--780},
	publisher = {{SIAM}},
	title = {Generalized Preconditioning and Undirected Minimum-Cost Flow},
	url = {https://doi.org/10.1137/1.9781611974782.49},
	year = {2017}}

@inproceedings{Kelner:2014aa,
	author = {Jonathan A. Kelner and Yin Tat Lee and Lorenzo Orecchia and Aaron Sidford},
	bibsource = {dblp computer science bibliography, https://dblp.org},
	booktitle = {25th \SODA},
	doi = {10.1137/1.9781611973402.16},
	editor = {Chandra Chekuri},
	pages = {217--226},
	publisher = {{SIAM}},
	title = {An Almost-Linear-Time Algorithm for Approximate Max Flow in Undirected Graphs, and its Multicommodity Generalizations},
	url = {https://doi.org/10.1137/1.9781611973402.16},
	year = {2014}}

@inproceedings{sherman2013nearlymaximumflowsnearly,
	author = {Jonah Sherman},
	bibsource = {dblp computer science bibliography, https://dblp.org},
	booktitle = {54th \FOCS},
	doi = {10.1109/FOCS.2013.36},
	pages = {263--269},
	publisher = {{IEEE} Computer Society},
	title = {Nearly Maximum Flows in Nearly Linear Time},
	url = {https://doi.org/10.1109/FOCS.2013.36},
	year = {2013}}

@inproceedings{Li:2025aa,
	author = {Jason Li and Satish Rao and Di Wang},
	bibsource = {dblp computer science bibliography, https://dblp.org},
	booktitle = {\SODA},
	doi = {10.1137/1.9781611978322.68},
	editor = {Yossi Azar and Debmalya Panigrahi},
	pages = {2111--2131},
	publisher = {{SIAM}},
	title = {Congestion-Approximators from the Bottom Up},
	url = {https://doi.org/10.1137/1.9781611978322.68},
	year = {2025}}

@inbook{Shmoys:1996aa,
	author = {Shmoys, David B.},
	booktitle = {Approximation Algorithms for NP-Hard Problems},
	isbn = {0534949681},
	numpages = {44},
	pages = {192--235},
	publisher = {PWS Publishing Co.},
	title = {Cut problems and their application to divide-and-conquer},
	url = {https://people.orie.cornell.edu/shmoys/pdf/multicut.pdf},
	year = {1996}}

@inproceedings{brand2024almostlineartimealgorithmsdecremental,
	author = {Jan van den Brand and Li Chen and Rasmus Kyng and Yang P. Liu and Simon Meierhans and Maximilian Probst Gutenberg and Sushant Sachdeva},
	bibsource = {dblp computer science bibliography, https://dblp.org},
	booktitle = {65th \FOCS},
	doi = {10.1109/FOCS61266.2024.00120},
	pages = {2010--2032},
	publisher = {{IEEE}},
	title = {Almost-Linear Time Algorithms for Decremental Graphs: Min-Cost Flow and More via Duality},
	url = {https://doi.org/10.1109/FOCS61266.2024.00120},
	year = {2024}}

@inproceedings{Racke:2002aa,
	author = {Harald R{\"{a}}cke},
	bibsource = {dblp computer science bibliography, https://dblp.org},
	booktitle = {43rd \FOCS},
	doi = {10.1109/SFCS.2002.1181881},
	pages = {43--52},
	title = {Minimizing Congestion in General Networks},
	url = {https://doi.org/10.1109/SFCS.2002.1181881},
	year = {2002}}

@inproceedings{peng2015approximateundirectedmaximumflows,
	author = {Richard Peng},
	bibsource = {dblp computer science bibliography, https://dblp.org},
	booktitle = {27th \SODA},
	doi = {10.1137/1.9781611974331.CH130},
	editor = {Robert Krauthgamer},
	pages = {1862--1867},
	publisher = {{SIAM}},
	title = {Approximate Undirected Maximum Flows in \emph{O}(\emph{m}polylog(\emph{n})) Time},
	url = {https://doi.org/10.1137/1.9781611974331.ch130},
	year = {2016}}

@inproceedings{Brand:2021aa,
	author = {Jan van den Brand and Yin Tat Lee and Yang P. Liu and Thatchaphol Saranurak and Aaron Sidford and Zhao Song and Di Wang},
	bibsource = {dblp computer science bibliography, https://dblp.org},
	booktitle = {53rd \STOC},
	doi = {10.1145/3406325.3451108},
	editor = {Samir Khuller and Virginia Vassilevska Williams},
	pages = {859--869},
	publisher = {{ACM}},
	title = {Minimum cost flows, MDPs, and $\ell _1$-regression in nearly linear time for dense instances},
	url = {https://doi.org/10.1145/3406325.3451108},
	year = {2021}}

@inproceedings{Goranci:2021ab,
	author = {Gramoz Goranci and Harald R{\"{a}}cke and Thatchaphol Saranurak and Zihan Tan},
	bibsource = {dblp computer science bibliography, https://dblp.org},
	booktitle = {32nd \SODA},
	doi = {10.1137/1.9781611976465.132},
	pages = {2212--2228},
	title = {The Expander Hierarchy and its Applications to Dynamic Graph Algorithms},
	url = {https://doi.org/10.1137/1.9781611976465.132},
	year = {2021}}

@inproceedings{Racke:2008aa,
	author = {Harald R{\"{a}}cke},
	bibsource = {dblp computer science bibliography, https://dblp.org},
	booktitle = {40th \STOC},
	doi = {10.1145/1374376.1374415},
	pages = {255--264},
	title = {Optimal hierarchical decompositions for congestion minimization in networks},
	url = {https://doi.org/10.1145/1374376.1374415},
	year = {2008}}

@article{Karger:2000aa,
	author = {David R. Karger},
	bibsource = {dblp computer science bibliography, https://dblp.org},
	doi = {10.1145/331605.331608},
	journal = {J. {ACM}},
	number = {1},
	pages = {46--76},
	title = {Minimum cuts in near-linear time},
	url = {https://doi.org/10.1145/331605.331608},
	volume = {47},
	year = {2000}}

@article{Arora:2009aa,
	author = {Sanjeev Arora and Satish Rao and Umesh V. Vazirani},
	bibsource = {dblp computer science bibliography, https://dblp.org},
	doi = {10.1145/1502793.1502794},
	journal = {J. {ACM}},
	number = {2},
	pages = {5:1--5:37},
	title = {Expander flows, geometric embeddings and graph partitioning},
	url = {https://doi.org/10.1145/1502793.1502794},
	volume = {56},
	year = {2009}}

@article{sleator1981data,
	author = {Daniel Dominic Sleator and Robert Endre Tarjan},
	bibsource = {dblp computer science bibliography, https://dblp.org},
	doi = {10.1016/0022-0000(83)90006-5},
	journal = {J. Comput. Syst. Sci.},
	number = {3},
	pages = {362--391},
	title = {A Data Structure for Dynamic Trees},
	url = {https://doi.org/10.1016/0022-0000(83)90006-5},
	volume = {26},
	year = {1983}}

@article{chen2022maximumflowminimumcostflow,
	author = {Chen, Li and Kyng, Rasmus and Liu, Yang and Peng, Richard and Probst Gutenberg, Maximilian and Sachdeva, Sushant},
	doi = {10.1145/3728631},
	issn = {1557-735X},
	journal = {Journal of the ACM},
	month = may,
	number = {3},
	pages = {1--103},
	publisher = {Association for Computing Machinery (ACM)},
	title = {Maximum Flow and Minimum-Cost Flow in Almost-Linear Time},
	url = {http://dx.doi.org/10.1145/3728631},
	volume = {72},
	year = {2025}}

@inproceedings{7782947,
	author = {Ittai Abraham and David Durfee and Ioannis Koutis and Sebastian Krinninger and Richard Peng},
	bibsource = {dblp computer science bibliography, https://dblp.org},
	booktitle = {57th \FOCS},
	doi = {10.1109/FOCS.2016.44},
	pages = {335--344},
	title = {On Fully Dynamic Graph Sparsifiers},
	url = {https://doi.org/10.1109/FOCS.2016.44},
	year = {2016}}

@article{Chen:2023aa,
	author = {Li Chen and Rasmus Kyng and Yang P. Liu and Richard Peng and Maximilian Probst Gutenberg and Sushant Sachdeva},
	bibsource = {dblp computer science bibliography, https://dblp.org},
	doi = {10.1145/3610940},
	journal = {Commun. {ACM}},
	number = {12},
	pages = {85--92},
	title = {Almost-Linear-Time Algorithms for Maximum Flow and Minimum-Cost Flow},
	url = {https://doi.org/10.1145/3610940},
	volume = {66},
	year = {2023}}

@inproceedings{Sherman:2009aa,
	author = {Jonah Sherman},
	bibsource = {dblp computer science bibliography, https://dblp.org},
	booktitle = {50th \STOC},
	doi = {10.1109/FOCS.2009.66},
	pages = {363--372},
	publisher = {{IEEE} Computer Society},
	title = {Breaking the Multicommodity Flow Barrier for O(vlog n)-Approximations to Sparsest Cut},
	url = {https://doi.org/10.1109/FOCS.2009.66},
	year = {2009}}

@inproceedings{chen2024almost,
	author = {Chen, Li and Kyng, Rasmus and Liu, Yang P and Meierhans, Simon and Probst Gutenberg, Maximilian},
	booktitle = {56th \STOC},
	doi = {10.1145/3618260.3649745},
	pages = {1165--1173},
	title = {Almost-linear time algorithms for incremental graphs: Cycle detection, sccs, $s$-$t$ shortest path, and minimum-cost flow},
	year = {2024}}

@article{arora2012multiplicative,
	author = {Arora, Sanjeev and Hazan, Elad and Kale, Satyen},
	journal = {Theory of computing},
	number = {1},
	pages = {121--164},
	publisher = {Theory of Computing Exchange},
	title = {The multiplicative weights update method: a meta-algorithm and applications},
	url = {https://theoryofcomputing.org/articles/v008a006},
	volume = {8},
	year = {2012}}

@inproceedings{pinsker1973concentrator,
	author = {M. Pinsker},
	booktitle = {Proceedings of the 7th Annual Teletraffic Conference},
	pages = {1--4},
	title = {On the Complexity of a Concentrator},
	url = {https://citeseerx.ist.psu.edu/document?repid=rep1&type=pdf&doi=71c9fd11ff75889aaa903b027af3a06e750e8add},
	year = {1973}}

@inproceedings{Chen:2020aa,
	author = {Li Chen and Gramoz Goranci and Monika Henzinger and Richard Peng and Thatchaphol Saranurak},
	bibsource = {dblp computer science bibliography, https://dblp.org},
	booktitle = {61st \FOCS},
	doi = {10.1109/FOCS46700.2020.00109},
	pages = {1135--1146},
	title = {Fast Dynamic Cuts, Distances and Effective Resistances via Vertex Sparsifiers},
	url = {https://doi.org/10.1109/FOCS46700.2020.00109},
	year = {2020}}

@inproceedings{Madry:2010ab,
	author = {Aleksander Madry},
	bibsource = {dblp computer science bibliography, https://dblp.org},
	booktitle = {51th \FOCS},
	doi = {10.1109/FOCS.2010.30},
	pages = {245--254},
	title = {Fast Approximation Algorithms for Cut-Based Problems in Undirected Graphs},
	url = {https://doi.org/10.1109/FOCS.2010.30},
	year = {2010}}

@article{10.1145/502090.502095,
	address = {New York, NY, USA},
	author = {Holm, Jacob and de Lichtenberg, Kristian and Thorup, Mikkel},
	doi = {10.1145/502090.502095},
	issn = {0004-5411},
	issue_date = {July 2001},
	journal = {J. ACM},
	month = jul,
	number = {4},
	numpages = {38},
	pages = {723--760},
	publisher = {Association for Computing Machinery},
	title = {Poly-logarithmic deterministic fully-dynamic algorithms for connectivity, minimum spanning tree, 2-edge, and biconnectivity},
	url = {https://doi.org/10.1145/502090.502095},
	volume = {48},
	year = {2001}}

@inproceedings{generalsparse,
	address = {New York, NY, USA},
	author = {Fung, Wai Shing and Hariharan, Ramesh and Harvey, Nicholas J.A. and Panigrahi, Debmalya},
	booktitle = {43rd \STOC},
	doi = {10.1145/1993636.1993647},
	isbn = {9781450306911},
	location = {San Jose, California, USA},
	numpages = {10},
	pages = {71--80},
	publisher = {Association for Computing Machinery},
	series = {STOC '11},
	title = {A general framework for graph sparsification},
	url = {https://doi.org/10.1145/1993636.1993647},
	year = {2011}}

@article{kargersparse,
	author = {Bencz\'{u}r, Andr\'{a}s A. and Karger, David R.},
	doi = {10.1137/070705970},
	eprint = {https://doi.org/10.1137/070705970},
	journal = {SIAM Journal on Computing},
	number = {2},
	pages = {290-319},
	title = {Randomized Approximation Schemes for Cuts and Flows in Capacitated Graphs},
	url = {https://doi.org/10.1137/070705970},
	volume = {44},
	year = {2015}}

@inproceedings{Li2023NearLinear,
	author = {Jason Li and Danupon Nanongkai and Debmalya Panigrahi and Thatchaphol Saranurak},
	bibsource = {dblp computer science bibliography, https://dblp.org},
	booktitle = {\SODA},
	doi = {10.1137/1.9781611977554.CH10},
	editor = {Nikhil Bansal and Viswanath Nagarajan},
	pages = {240--275},
	publisher = {{SIAM}},
	title = {Near-Linear Time Approximations for Cut Problems via Fair Cuts},
	url = {https://doi.org/10.1137/1.9781611977554.ch10},
	year = {2023}}

@article{AlonM85,
	author = {Noga Alon and V. D. Milman},
	bibsource = {dblp computer science bibliography, https://dblp.org},
	doi = {10.1016/0095-8956(85)90092-9},
	journal = {J. Comb. Theory {B}},
	number = {1},
	pages = {73--88},
	title = {$\lambda _1$, Isoperimetric inequalities for graphs, and superconcentrators},
	url = {https://doi.org/10.1016/0095-8956(85)90092-9},
	volume = {38},
	year = {1985}}

@inproceedings{AndersenP09,
	author = {Reid Andersen and Yuval Peres},
	bibsource = {dblp computer science bibliography, https://dblp.org},
	booktitle = {41st \STOC},
	doi = {10.1145/1536414.1536449},
	editor = {Michael Mitzenmacher},
	pages = {235--244},
	publisher = {{ACM}},
	title = {Finding sparse cuts locally using evolving sets},
	url = {https://doi.org/10.1145/1536414.1536449},
	year = {2009}}

@article{LeightonR99,
	author = {Frank Thomson Leighton and Satish Rao},
	bibsource = {dblp computer science bibliography, https://dblp.org},
	doi = {10.1145/331524.331526},
	journal = {J. {ACM}},
	number = {6},
	pages = {787--832},
	title = {Multicommodity max-flow min-cut theorems and their use in designing approximation algorithms},
	url = {https://doi.org/10.1145/331524.331526},
	volume = {46},
	year = {1999}}

@article{AroraHK10,
	author = {Sanjeev Arora and Elad Hazan and Satyen Kale},
	bibsource = {dblp computer science bibliography, https://dblp.org},
	doi = {10.1137/080731049},
	journal = {{SIAM} J. Comput.},
	number = {5},
	pages = {1748--1771},
	title = {O(sqrt(log(n)) Approximation to {SPARSEST} {CUT} in {\~{O}}(n\({}^{\mbox{2}}\)) Time},
	url = {https://doi.org/10.1137/080731049},
	volume = {39},
	year = {2010}}

@article{KhandekarRV09,
	author = {Rohit Khandekar and Satish Rao and Umesh V. Vazirani},
	bibsource = {dblp computer science bibliography, https://dblp.org},
	doi = {10.1145/1538902.1538903},
	journal = {J. {ACM}},
	number = {4},
	pages = {19:1--19:15},
	title = {Graph partitioning using single commodity flows},
	url = {https://doi.org/10.1145/1538902.1538903},
	volume = {56},
	year = {2009}}

@article{AroraK16,
	author = {Sanjeev Arora and Satyen Kale},
	bibsource = {dblp computer science bibliography, https://dblp.org},
	doi = {10.1145/2837020},
	journal = {J. {ACM}},
	number = {2},
	pages = {12:1--12:35},
	title = {A Combinatorial, Primal-Dual Approach to Semidefinite Programs},
	url = {https://doi.org/10.1145/2837020},
	volume = {63},
	year = {2016}}

@inproceedings{OrecchiaSVV08,
	author = {Lorenzo Orecchia and Leonard J. Schulman and Umesh V. Vazirani and Nisheeth K. Vishnoi},
	bibsource = {dblp computer science bibliography, https://dblp.org},
	booktitle = {40th \STOC},
	doi = {10.1145/1374376.1374442},
	editor = {Cynthia Dwork},
	pages = {461--470},
	publisher = {{ACM}},
	title = {On partitioning graphs via single commodity flows},
	url = {https://doi.org/10.1145/1374376.1374442},
	year = {2008}}

@article{gomory1961multi,
	author = {Gomory, Ralph E and Hu, Tien Chung},
	doi = {https://doi.org/10.1137/0109047},
	journal = {Journal of the Society for Industrial and Applied Mathematics},
	number = {4},
	pages = {551--570},
	publisher = {SIAM},
	title = {Multi-terminal network flows},
	volume = {9},
	year = {1961}}

@article{AbboudKT21,
	author = {Amir Abboud and Robert Krauthgamer and Ohad Trabelsi},
	bibsource = {dblp computer science bibliography, https://dblp.org},
	doi = {10.4086/TOC.2021.V017A005},
	journal = {Adv. Math. Commun.},
	pages = {1--27},
	title = {New Algorithms and Lower Bounds for All-Pairs Max-Flow in Undirected Graphs},
	url = {https://doi.org/10.4086/toc.2021.v017a005},
	volume = {17},
	year = {2021}}

@inproceedings{AbboudKT21b,
	author = {Amir Abboud and Robert Krauthgamer and Ohad Trabelsi},
	bibsource = {dblp computer science bibliography, https://dblp.org},
	booktitle = {53rd \STOC},
	doi = {10.1145/3406325.3451073},
	editor = {Samir Khuller and Virginia Vassilevska Williams},
	pages = {1725--1737},
	publisher = {{ACM}},
	title = {Subcubic algorithms for Gomory-Hu tree in unweighted graphs},
	url = {https://doi.org/10.1145/3406325.3451073},
	year = {2021}}

@inproceedings{AbboudKT21a,
	author = {Amir Abboud and Robert Krauthgamer and Ohad Trabelsi},
	bibsource = {dblp computer science bibliography, https://dblp.org},
	booktitle = {62nd \FOCS},
	doi = {10.1109/FOCS52979.2021.00112},
	pages = {1135--1146},
	publisher = {{IEEE}},
	title = {{APMF} {\textless} APSP? Gomory-Hu Tree for Unweighted Graphs in Almost-Quadratic Time},
	url = {https://doi.org/10.1109/FOCS52979.2021.00112},
	year = {2021}}

@inproceedings{0006PS21,
	author = {Jason Li and Debmalya Panigrahi and Thatchaphol Saranurak},
	bibsource = {dblp computer science bibliography, https://dblp.org},
	booktitle = {62nd \FOCS},
	doi = {10.1109/FOCS52979.2021.00111},
	pages = {1124--1134},
	publisher = {{IEEE}},
	title = {A Nearly Optimal All-Pairs Min-Cuts Algorithm in Simple Graphs},
	url = {https://doi.org/10.1109/FOCS52979.2021.00111},
	year = {2021}}

@inproceedings{Zhang22,
	author = {Tianyi Zhang},
	bibsource = {dblp computer science bibliography, https://dblp.org},
	booktitle = {49th \ICALP},
	doi = {10.4230/LIPICS.ICALP.2022.109},
	editor = {Mikolaj Bojanczyk and Emanuela Merelli and David P. Woodruff},
	pages = {109:1--109:18},
	publisher = {Schloss Dagstuhl - Leibniz-Zentrum f{\"{u}}r Informatik},
	series = {LIPIcs},
	title = {Faster Cut-Equivalent Trees in Simple Graphs},
	url = {https://doi.org/10.4230/LIPIcs.ICALP.2022.109},
	volume = {229},
	year = {2022}}

@inproceedings{BrandCKLPGSS24,
	author = {Jan van den Brand and Li Chen and Rasmus Kyng and Yang P. Liu and Richard Peng and Maximilian Probst Gutenberg and Sushant Sachdeva and Aaron Sidford},
	bibsource = {dblp computer science bibliography, https://dblp.org},
	booktitle = {\SODA},
	doi = {10.1137/1.9781611977912.106},
	editor = {David P. Woodruff},
	pages = {2980--2998},
	publisher = {{SIAM}},
	title = {Incremental Approximate Maximum Flow on Undirected Graphs in Subpolynomial Update Time},
	url = {https://doi.org/10.1137/1.9781611977912.106},
	year = {2024}}

@inproceedings{GoranciH23,
	author = {Gramoz Goranci and Monika Henzinger},
	bibsource = {dblp computer science bibliography, https://dblp.org},
	booktitle = {50th \ICALP},
	doi = {10.4230/LIPICS.ICALP.2023.69},
	editor = {Kousha Etessami and Uriel Feige and Gabriele Puppis},
	pages = {69:1--69:14},
	publisher = {Schloss Dagstuhl - Leibniz-Zentrum f{\"{u}}r Informatik},
	series = {LIPIcs},
	title = {Efficient Data Structures for Incremental Exact and Approximate Maximum Flow},
	url = {https://doi.org/10.4230/LIPIcs.ICALP.2023.69},
	volume = {261},
	year = {2023}}

@inproceedings{GuptaK21,
	author = {Manoj Gupta and Shahbaz Khan},
	bibsource = {dblp computer science bibliography, https://dblp.org},
	booktitle = {4th \SOSA},
	doi = {10.1137/1.9781611976496.10},
	editor = {Hung Viet Le and Valerie King},
	pages = {86--91},
	publisher = {{SIAM}},
	title = {Simple dynamic algorithms for Maximal Independent Set, Maximum Flow and Maximum Matching},
	url = {https://doi.org/10.1137/1.9781611976496.10},
	year = {2021}}

@inproceedings{GoranciHRS25,
	author = {Gramoz Goranci and Monika Henzinger and Harald R{\"{a}}cke and A. R. Sricharan},
	bibsource = {dblp computer science bibliography, https://dblp.org},
	booktitle = {52nd \ICALP},
	doi = {10.4230/LIPICS.ICALP.2025.91},
	editor = {Keren Censor{-}Hillel and Fabrizio Grandoni and Jo{\"{e}}l Ouaknine and Gabriele Puppis},
	pages = {91:1--91:20},
	publisher = {Schloss Dagstuhl - Leibniz-Zentrum f{\"{u}}r Informatik},
	series = {LIPIcs},
	title = {Incremental Approximate Maximum Flow via Residual Graph Sparsification},
	url = {https://doi.org/10.4230/LIPIcs.ICALP.2025.91},
	volume = {334},
	year = {2025}}

@inproceedings{GoranciHT16,
	author = {Gramoz Goranci and Monika Henzinger and Mikkel Thorup},
	bibsource = {dblp computer science bibliography, https://dblp.org},
	booktitle = {24th \ESA},
	doi = {10.4230/LIPICS.ESA.2016.46},
	editor = {Piotr Sankowski and Christos D. Zaroliagis},
	pages = {46:1--46:17},
	publisher = {Schloss Dagstuhl - Leibniz-Zentrum f{\"{u}}r Informatik},
	series = {LIPIcs},
	title = {Incremental Exact Min-Cut in Poly-logarithmic Amortized Update Time},
	url = {https://doi.org/10.4230/LIPIcs.ESA.2016.46},
	volume = {57},
	year = {2016}}

@inproceedings{GoranciHNSTW23,
	author = {Gramoz Goranci and Monika Henzinger and Danupon Nanongkai and Thatchaphol Saranurak and Mikkel Thorup and Christian Wulff{-}Nilsen},
	bibsource = {dblp computer science bibliography, https://dblp.org},
	booktitle = {\SODA},
	doi = {10.1137/1.9781611977554.CH3},
	editor = {Nikhil Bansal and Viswanath Nagarajan},
	pages = {70--86},
	publisher = {{SIAM}},
	title = {Fully Dynamic Exact Edge Connectivity in Sublinear Time},
	url = {https://doi.org/10.1137/1.9781611977554.ch3},
	year = {2023}}

@inproceedings{El-HayekH025,
	author = {Antoine El{-}Hayek and Monika Henzinger and Jason Li},
	bibsource = {dblp computer science bibliography, https://dblp.org},
	booktitle = {\SODA},
	doi = {10.1137/1.9781611978322.22},
	editor = {Yossi Azar and Debmalya Panigrahi},
	pages = {750--784},
	publisher = {{SIAM}},
	title = {Fully Dynamic Approximate Minimum Cut in Subpolynomial Time per Operation},
	url = {https://doi.org/10.1137/1.9781611978322.22},
	year = {2025}}

@inproceedings{KyngMG24,
	author = {Rasmus Kyng and Simon Meierhans and Maximilian Probst Gutenberg},
	bibsource = {dblp computer science bibliography, https://dblp.org},
	booktitle = {56th \STOC},
	doi = {10.1145/3618260.3649767},
	editor = {Bojan Mohar and Igor Shinkar and Ryan O'Donnell},
	pages = {1174--1183},
	publisher = {{ACM}},
	title = {A Dynamic Shortest Paths Toolbox: Low-Congestion Vertex Sparsifiers and Their Applications},
	url = {https://doi.org/10.1145/3618260.3649767},
	year = {2024}}

@inproceedings{young,
	author = {Young, N.E.},
	booktitle = {Proceedings 42nd IEEE Symposium on Foundations of Computer Science},
	doi = {10.1109/SFCS.2001.959930},
	pages = {538-546},
	title = {Sequential and parallel algorithms for mixed packing and covering},
	year = {2001}}

@inproceedings{bernstein_et_al,
	address = {Dagstuhl, Germany},
	author = {Bernstein, Aaron and van den Brand, Jan and Probst Gutenberg, Maximilian and Nanongkai, Danupon and Saranurak, Thatchaphol and Sidford, Aaron and Sun, He},
	booktitle = {49th \ICALP},
	doi = {10.4230/LIPIcs.ICALP.2022.20},
	editor = {Boja\'{n}czyk, Miko{\l}aj and Merelli, Emanuela and Woodruff, David P.},
	isbn = {978-3-95977-235-8},
	issn = {1868-8969},
	pages = {20:1--20:20},
	publisher = {Schloss Dagstuhl -- Leibniz-Zentrum f{\"u}r Informatik},
	series = {Leibniz International Proceedings in Informatics (LIPIcs)},
	title = {{Fully-Dynamic Graph Sparsifiers Against an Adaptive Adversary}},
	url = {https://drops.dagstuhl.de/entities/document/10.4230/LIPIcs.ICALP.2022.20},
	urn = {urn:nbn:de:0030-drops-163611},
	volume = {229},
	year = {2022}}
\end{document}